\documentclass[aps,pra,twocolumn,nofootinbib, superscriptaddress,10pt]{revtex4-2}

\usepackage{graphicx}
\usepackage{epstopdf}
\usepackage{amsmath}
\usepackage{amssymb}
\usepackage{mathtools}
\usepackage{mathrsfs}
\usepackage{amsthm}
\usepackage{bm}
\usepackage{bbm}
\usepackage{url}
\usepackage[T1]{fontenc}
\usepackage{csquotes}
\MakeOuterQuote{"}
\usepackage{thmtools, thm-restate}
\usepackage{longtable}
\usepackage{array}
\usepackage{diagbox}

\newtheoremstyle{note}
{\topsep/2}               
{\topsep/2}               
{}                      
{\parindent}            
{\itshape}              
{.}                     
{5pt plus 1pt minus 1pt}
{}

\theoremstyle{note}
\newtheorem{thm}{Theorem}
\newtheorem{lem}{Lemma}
\newtheorem{conjecture}{Conjecture}

\theoremstyle{definition}

\theoremstyle{remark}

\newtheorem{thm*}{Theorem}

\makeatletter
\expandafter\newcommand\csname thethm*default\endcsname{\thethm*}
\newcommand{\thmstarnum}[1]{\expandafter\gdef\csname thethm*\endcsname{#1*}}
\thmstarnum{\thethm}
\expandafter\g@addto@macro\csname endthm*\endcsname{\thmstarnum{\thethm}}
\makeatother

\newtheorem{lem*}{Lemma}

\makeatletter
\expandafter\newcommand\csname thelem*default\endcsname{\thelem*}
\newcommand{\lemstarnum}[1]{\expandafter\gdef\csname thelem*\endcsname{#1*}}
\lemstarnum{\thelem}
\expandafter\g@addto@macro\csname endlem*\endcsname{\lemstarnum{\thelem}}
\makeatother

\def\vec#1{{\bm{#1}}} 

\newcommand{\tr}{\operatorname{tr}}

\newcommand{\sgn}{\operatorname{sgn}}

\newcommand{\imply}{\mathrel{\Rightarrow}}

\newcommand{\rmc}{\mathrm{c}}
\newcommand{\rmd}{\mathrm{d}}

\newcommand{\rmi}{\mathrm{i}}
\newcommand{\rmo}{\mathrm{o}}

\newcommand{\rmt}{\mathrm{t}}

\newcommand{\rmT}{\mathrm{T}}

\newcommand{\caH}{\mathcal{H}}

\newcommand{\caV}{\mathcal{V}}

\newcommand{\bbC}{\mathbb{C}}

\newcommand{\sym}{\mathrm{sym}}

\newcommand{\scrM}{\mathscr{M}}

\newcommand{\scrS}{\mathscr{S}}

\newcommand{\be}{\begin{equation}}
\newcommand{\ee}{\end{equation}}
\newcommand{\ba}{\begin{align}}
\newcommand{\ea}{\end{align}}

\def\<{\langle}  
\def\>{\rangle}  
\newcommand{\ket}[1]{| #1\>}
\newcommand{\bra}[1]{\< #1|}

\def\eqref#1{\textup{(\ref{#1})}}  
\newcommand{\eref}[1]{Eq.~\textup{(\ref{#1})}}

\newcommand{\esref}[2]{Eqs.~\textup{(\ref{#1})} and \textup{(\ref{#2})}}

\newcommand{\Esref}[2]{Equations~\textup{(\ref{#1})} and \textup{(\ref{#2})}}

\newcommand{\ecref}[2]{Eqs.~\textup{(\ref{#1})}-\textup{(\ref{#2})}}

\newcommand{\fref}[1]{Fig.~\ref{#1}}
\newcommand{\Fref}[1]{Figure~\ref{#1}}

\newcommand{\tref}[1]{Table~\ref{#1}}

\newcommand{\sref}[1]{Sec.~\ref{#1}}

\newcommand{\thref}[1]{Theorem~\ref{#1}}
\newcommand{\Thref}[1]{Theorem~\ref{#1}}
\newcommand{\thsref}[1]{Theorems~\ref{#1}}

\newcommand{\lref}[1]{Lemma~\ref{#1}}
\newcommand{\Lref}[1]{Lemma~\ref{#1}}

\newcommand{\cref}[1]{Conjecture~\ref{#1}}
\newcommand{\Cref}[1]{Conjecture~\ref{#1}}

\newcommand{\aref}[1]{Appendix~\ref{#1}}

\newcommand{\rcite}[1]{Ref.~\cite{#1}}
\newcommand{\rscite}[1]{Refs.~\cite{#1}}

\begin{document}
\title{Efficient verification of Affleck-Kennedy-Lieb-Tasaki states}

\author{Tianyi Chen}
\email{These authors contributed equally to this work.}
\affiliation{State Key Laboratory of Surface Physics and Department of Physics, Fudan University, Shanghai 200433, China}
\affiliation{Institute for Nanoelectronic Devices and Quantum Computing, Fudan University, Shanghai 200433, China}
\affiliation{Center for Field Theory and Particle Physics, Fudan University, Shanghai 200433, China}

\author{Yunting Li}
\email{These authors contributed equally to this work.}
\affiliation{State Key Laboratory of Surface Physics and Department of Physics, Fudan University, Shanghai 200433, China}
\affiliation{Institute for Nanoelectronic Devices and Quantum Computing, Fudan University, Shanghai 200433, China}
\affiliation{Center for Field Theory and Particle Physics, Fudan University, Shanghai 200433, China}

\author{Huangjun Zhu}
\email{zhuhuangjun@fudan.edu.cn}
\affiliation{State Key Laboratory of Surface Physics and Department of Physics, Fudan University, Shanghai 200433, China}
\affiliation{Institute for Nanoelectronic Devices and Quantum Computing, Fudan University, Shanghai 200433, China}
\affiliation{Center for Field Theory and Particle Physics, Fudan University, Shanghai 200433, China}

\begin{abstract}
Affleck-Kennedy-Lieb-Tasaki (AKLT) states are an important class of many-body quantum states that are useful in quantum information processing, including measurement-based quantum computation in particular. 
Here we propose a general approach for constructing efficient verification protocols for  AKLT states  on  arbitrary graphs with local spin measurements. 
Our verification protocols build on  bond verification protocols and matching covers (including edge coloring) of the underlying graphs, which have a simple geometric and graphic picture. 
We also provide rigorous performance guarantee that is required for practical applications. 
With our approach, most AKLT states of wide interest, including those defined on 1D and 2D lattices, can be verified with a constant sample cost, which is independent 
of the system size and is dramatically more efficient than all previous approaches. As an illustration, we 
construct concrete verification protocols for  AKLT states on various lattices and on arbitrary graphs up to five vertices.  
\end{abstract}

\date{\today}
\maketitle


\section{Introduction}

Ground states of local Hamiltonians play crucial roles in many-body physics and have also found increasing applications in quantum information processing \cite{StepWPW17,Wei18, StepNBE19,DaniAM20,GoihWET20,WeiRA22}. The  Affleck-Kennedy-Lieb-Tasaki (AKLT) states \cite{AfflKLT87,AfflKLT88} are of special interest because they are the ground states of exactly solvable models and are tied to the famous Haldane conjecture \cite{Hald83C,Hald83N}. These states are originally defined on spin chains and  have  been  generalized to arbitrary graphs later \cite{KiriK09,XuK08,KoreX10}. Recently,  AKLT states have attracted increasing attention because of their connection with symmetry-protected topological orders \cite{ChenGLW12, Sent15,ChiuTSA16}. Moreover,  AKLT states  on many 2D lattices, including the honeycomb lattice, are  universal resource states for measurement-based quantum computation \cite{RausB01,KaltLZB10,WeiAR11,Miya11,WeiAR12}.

In practice it is not easy to prepare many-body states, such as AKLT states, perfectly. Therefore, it is crucial to verify these states within a desired precision efficiently. However, traditional tomographic approaches  are too resource consuming to achieve this goal for  large and intermediate quantum systems. Recently, great efforts have been directed to addressing this problem \cite{EiseHWR20,CarrEKK21,KlieR21,YuSG22,MorrSGD22},  and various alternative approaches have been proposed, including compressed sensing \cite{GrosLFB10}, direct fidelity estimation \cite{FlamL11,SilvLP11}, and shadow estimation \cite{Aaro18,HuanKP20} etc.

Here we are particularly interested in a promising  approach known as quantum state verification (QSV), which can achieve high efficiency based on local measurements 
\cite{HayaMT06,CramPFS10,AoliGKE15,HangKSE17, PallLM18,TakeM18,ZhuH19AdS,ZhuH19AdL,WuBGL21,LiuSHZ21,GocaSD22}
. So far efficient verification protocols have been constructed for  bipartite pure states \cite{HayaMT06,Haya09G,PallLM18,ZhuH19O,LiHZ19,WangH19,YuSG19}, stabilizer states (including graph states and Greenberger-Horne-Zeilinger states in particular) \cite{HayaM15,FujiH17,HayaH18, PallLM18,MarkK18,ZhuH19E,ZhuH19AdL,LiHZ20}, hypergraph states \cite{ZhuH19E}, weighted graph states \cite{HayaT19},  Dicke states \cite{LiuYSZ19}, and phased Dicke states (including Slater determinant states) \cite{LiuYSZ19,LiHSS21}. In addition, several verification  protocols have been demonstrated successfully  in experiments \cite{ZhanZCP20,ZhanLYP20,LuXCC20,JianWQC20}.  Moreover, this approach can be generalized to the verification of quantum gates and processes \cite{WuB19,LiuSYZ20,ZhuZ20,ZengZL20}, which have also been demonstrated in experiments \cite{ZhanHTS22,LuooZZ22}. Unfortunately, efficient verification protocols known so far are usually tailored to quantum states with special structures and rely on explicit expressions of the states under consideration. For the ground states of local Hamiltonians, 
although several verification protocols have been proposed \cite{CramPFS10,HangKSE17,GluzKEA18,TakeM18,CruzBTS22},  it is still too resource consuming to verify large and intermediate quantum systems.

In this paper, following the simple recipe proposed in the companion paper \cite{ZhuLC22}, we propose a general approach for constructing efficient verification protocols for  AKLT states defined on  arbitrary graphs.  Notably, explicit expressions for the AKLT states are not necessary. Our verification protocols are based on local spin measurements and are thus easy to implement in experiments. 
In addition, these verification protocols have very simple description in terms of 
elementary geometric and graph theoretic concepts. Moreover, we provide rigorous upper bounds on the number of  tests (sample cost) required to achieve a given precision. With our approach, most AKLT states of practical interest, including those defined on 1D and 2D lattices, can be verified with a number of tests that is independent 
of the system size, which is dramatically more efficient than previous approaches \cite{CramPFS10,HangKSE17,TakeM18,CruzBTS22}.
In addition, we construct concrete verification protocols for all AKLT states defined on arbitrary graphs up to five vertices.

The rest of this paper is organized as follows. In \sref{sec:QSV} we first review the basic framework of QSV and then introduce the idea of subspace verification as a generalization. In \sref{sec:AKLT} we review the definition and basic properties of AKLT states that are relevant in later studies. In \sref{sec:BondVerify} we  clarify potential  bond test operators based on spin measurements and construct various optimal and efficient 
 bond verification protocols. In \sref{sec:ApproachGen}  we propose a general approach for constructing efficient verification protocols for  AKLT states together with rigorous performance guarantee. In \sref{sec:1DAKLT} we discuss in detail the verification of 1D AKLT states. In \sref{sec:VAKLTgenGraph} we construct concrete verification protocols for AKLT states defined on general graphs up to five vertices. In \sref{sec:summary} we summarize this paper. Several technical proofs and a table are relegated to the Appendix.

\section{\label{sec:QSV}Quantum state verification}
In this section we first review the general framework of QSV following \rscite{PallLM18,ZhuH19AdS,ZhuH19AdL}. Then we generalize the idea to subspace verification, which is closely tied to the verification of ground states of local Hamiltonians.

\subsection{Basic framework}
A primitive in quantum information processing is to produce a given quantum state $|\Psi\rangle$ with prescribed properties. In practice, the device we employ is never perfect, and the states produced in individual runs may be different from the target state and also different from each other. 
So it is crucial to verify whether the deviation from the target state, usually quantified by the infidelity, is tolerable. To address this problem, in each run
we can perform a random two-outcome measurement $\{T_l,1-T_l\}$ determined by the test operator $T_l$, where the two outcomes correspond to passing and failing the test, respectively. 
To guarantee that the target state $|\Psi\rangle$ can  pass the test with certainty, the test operator $T_l$ needs to  satisfy the following requirement
\begin{align}\label{eq:TestTargetCon}
T_l|\Psi\>=|\Psi\>. 
\end{align}

Let $p_l$ be the probability of performing the test   $T_l$ and define the \emph{verification operator}
\begin{equation}\label{eq:VO}
\Omega=\sum_l p_l T_l. 
\end{equation}
Suppose $\sigma$ is a quantum state whose fidelity with the target state is at most $1-\epsilon$, that is,  $\bra{\Psi}\sigma \ket{\Psi} \leq 1-\epsilon$; then the probability that $\sigma$ can pass each test on average is bounded from above as follows  \cite{PallLM18,ZhuH19AdS,ZhuH19AdL},
\begin{equation}\label{eq:PassProb}
\max_{\<\Psi|\sigma|\Psi\>\leq 1-\epsilon }\tr(\Omega \sigma)=1- [1-\beta(\Omega)]\epsilon=1- \nu(\Omega)\epsilon,
\end{equation}
where $\beta(\Omega)$ denotes the second largest eigenvalue of the verification operator $\Omega$, and  $\nu(\Omega)=1-\beta(\Omega)$ is referred to as the \emph{spectral gap}.

As a corollary of \eref{eq:PassProb}
the probability that the states $\sigma_1, \sigma_2, \ldots, \sigma_N$ can pass all $N$ tests satisfies \cite{PallLM18,ZhuH19AdS,ZhuH19AdL}
\begin{align}\label{eq:PassProb1}
\prod_{j=1}^N\tr(\Omega \sigma_j)\leq \prod_{j=1}^N [1- \nu(\Omega)\epsilon_j]\leq [1- \nu(\Omega)\bar{\epsilon}]^N,
\end{align}
where $\epsilon_j=1-\<\Psi|\sigma_j|\Psi\>$ is the infidelity of the state $\sigma_j$ 
and $\bar{\epsilon}=\sum_j\epsilon_j/N$ is the average infidelity. 
To verify the
target state $|\Psi\rangle$ within infidelity $\epsilon$ and significance level $\delta$, which means $\delta\geq [1- \nu(\Omega)\epsilon]^N$, the minimum number of tests required is given by \cite{PallLM18,ZhuH19AdS,ZhuH19AdL}
\begin{equation}\label{eq:nu}
N=\biggl\lceil\frac{ \ln \delta}{\ln[1-\nu(\Omega)\epsilon]}\biggr\rceil\approx \frac{ \ln (\delta^{-1})}{\nu(\Omega)\epsilon},
\end{equation}
which decreases monotonically with  the spectral gap $\nu(\Omega)$. To achieve a high efficiency, we need to construct a verification operator with a large spectral gap. Here we shall focus on verification protocols that can be realized by local projective measurements, which are most amenable to practical applications.

The verification operator $\Omega$ is \emph{homogeneous} \cite{ZhuH19O,ZhuH19AdS,ZhuH19AdL} if it has the form 
\begin{equation}
\Omega=|\Psi\>\<\Psi|+\lambda(1-|\Psi\>\<\Psi|). 
\end{equation}
In this case, the probability that $\sigma$ can pass each test on average is completely determined by its fidelity with the target state $|\Psi\>$, 
\begin{equation}
\tr(\Omega\sigma)=\lambda+\nu F=1-\nu \epsilon,
\end{equation}
where $F=\<\Psi|\sigma|\Psi\>$ and $\epsilon=1-F$. 
Such verification protocols are of special interest because they can also be used for fidelity estimation.

\subsection{Subspace verification}
The basic  idea of QSV can also be applied to subspace verification, which  emerges naturally in the verification of multipartite pure states, such as the ground states local Hamiltonians. To see this point, suppose we want to verify the  multipartite pure state $|\Psi\>$ whose reduced states are mixed and are supported in certain subspaces.  To this end, we can verify that each reduced state is supported in  a particular subspace. Quite often it
turns out that the target state $|\Psi\>$ can be verified in this way without additional steps. Notably, this strategy is particularly useful to verifying the ground states of frustrate-free Hamiltonians, including AKLT states in particular.

Let $\caV$ be a given subspace of  the Hilbert space $\caH$ under consideration and $Q$  the corresponding projector. 
Our task is to verify whether the state produced is supported in this subspace. 
To address this problem, we can construct a set of tests and perform a random test from this set in each run. Every test corresponds to a two-outcome measurement $\{T_l,1-T_l\}$, which is determined by the test operator $T_l$, as in the verification of a pure state. Now  the condition
in \eref{eq:TestTargetCon} is replaced by
\begin{align}\label{eq:TestTargetProj}
T_l Q=Q, 
\end{align}
so that all states supported in $\caV$ can pass each test with certainty.  Let $p_l$ be the probability of performing the test 
 $T_l$ and define the verification operator 
\begin{align}
\Omega=\sum_l p_l T_l
\end{align} 
as in QSV [cf.~\eref{eq:VO}]. Then $\Omega$ satisfies the condition $\Omega Q=Q$
thanks to \eref{eq:TestTargetProj}.

Suppose the quantum state $\sigma$ under consideration satisfies the condition $\tr(Q\sigma)\leq 1-\epsilon$; then the  probability that $\sigma$ can pass each test on average is bounded from above as follows,
\begin{equation}\label{eq:PassProbSS}
\max_{\tr(Q\sigma)\leq 1-\epsilon }\tr(\Omega \sigma)=1- [1-\beta(\Omega)]\epsilon=1- \nu(\Omega)\epsilon,
\end{equation}
where 
\begin{align}
\beta(\Omega)&=\|\bar{\Omega}\|,  \quad \nu(\Omega)=1-\beta(\Omega),\\ \bar{\Omega}&=(1-Q)\Omega(1-Q).
\end{align}
 As a corollary of \eref{eq:PassProbSS},
the probability that the states $\sigma_1, \sigma_2, \ldots, \sigma_N$ produced in $N$ runs can pass all $N$ tests satisfies
\begin{align}
\prod_{j=1}^N\tr(\Omega \sigma_j)\leq \prod_{j=1}^N [1- \nu(\Omega)\epsilon_j]\leq [1- \nu(\Omega)\bar{\epsilon}]^N,
\end{align}
where $\epsilon_j=1-\tr(Q\sigma_j)$ and  $\bar{\epsilon}=(\sum_j \epsilon_j)/N$. This result has the same form as the counterpart in QSV. To verify the subspace $\caV$ within infidelity $\epsilon$ and significance level $\delta$, the number of tests required reads 
\begin{equation}
N=\biggl\lceil\frac{ \ln \delta}{\ln[1-\nu(\Omega)\epsilon]}\biggr\rceil\approx \frac{ \ln (\delta^{-1})}{\nu(\Omega)\epsilon},
\end{equation}
which also has the same form as the counterpart in QSV as presented in  \eref{eq:nu}.

The concept of homogeneous verification operators has a natural generalization in the context of subspace verification. Now the verification operator $\Omega$ is \emph{homogeneous}  if it has the form 
 \begin{equation}
 \Omega=Q+\lambda(1-Q). 
 \end{equation}
 In this case, the probability that $\sigma$ can pass each test on average is completely determined by the overlap $\tr(Q\sigma)$, 
 \begin{equation}
 \tr(\Omega\sigma)=\lambda+\nu \tr(Q\sigma)=1-\nu [1-\tr(Q\sigma)].
 \end{equation}
  In other words, the overlap $\tr(Q\sigma)$ can be estimated from the passing probability $\tr(\Omega\sigma)$. Such verification protocols will play an important role in the verification of AKLT states, as we shall see shortly.

\section{\label{sec:AKLT}AKLT States}
In this section we briefly review AKLT states defined on general graphs \cite{AfflKLT87,AfflKLT88,KiriK09,XuK08,KoreX10}. For the convenience of the readers, basic facts about spin operators and graphs are introduced in advance.

\subsection{Spin operators}
The spin operator associated with a spin-$S$ particle is denoted by $\bm{S}=(S_x,S_y,S_z)$, where $S_x$, $S_y$, $S_z$ are the spin operators along directions $\hat{x}$, $\hat{y}$, $\hat{z}$, respectively, which act on a Hilbert space of dimension $2S+1$. Note that $S_x$, $S_y$, $S_z$ are nondegenerate and  have eigenvalues $S, S-1, \ldots, -S$. 
Given an eigenvalue  $m$ of $S_z$, the corresponding eigenstate is denoted by  $|S, m\>$ or  $|m\>$ when $S$ is clear from the context. Then the operators $S_x, S_y, S_z$ can be expressed as 
\begin{align}
S_x=\frac{S_++S_-}{2},\;\; S_y=\frac{S_+-S_-}{2\rmi},\;\; S_z=\sum_{m=-S}^S m|m\>\<m|,  \label{eq:Spin}
\end{align}
where
\begin{equation}
\begin{aligned}
 S_{+}&=\sum_{m=-S}^{S-1} \sqrt{S(S+1)-m(m+ 1)}|m+1\>\<m|,
\\
 S_{-}&=\sum_{m=-S+1}^S \sqrt{S(S+1)-m(m- 1)}|m-1\>\<m|.
\end{aligned}
\end{equation}
When $S=1$ for example, $S_x$, $S_y$, $S_z$ have the following matrix representations,
\begin{equation}
\begin{aligned}
S_x&=\frac{1}{\sqrt{2}}
\begin{pmatrix}
0 & 1 & 0\\
1 & 0 & 1\\
0 & 1 & 0\\
\end{pmatrix},\quad 
S_y=\frac{\rmi}{\sqrt{2}}
\begin{pmatrix}
0 & -1 & 0\\
1 & 0 & -1\\
0 & 1 & 0\\
\end{pmatrix},\\
S_z&=
\begin{pmatrix}
1 & 0 & 0\\
0 & 0 & 0\\
0 & 0 & -1\\
\end{pmatrix}.
\end{aligned}
\end{equation}

Let $\bm{r}$ be a (real) unit vector in dimension 3, then the spin operator along direction $\bm{r}$ reads
\begin{align}\label{eq:sr}
S_{\vec{r}}:=\bm{r}\cdot\bm{S}=r_x S_x +r_y S_y  + r_z S_z.
\end{align}
Note that $S_{\vec{r}}$ has the same eigenvalues as $S_z$ for any unit vector $\vec{r}$ in dimension 3. The eigenstate of $S_{\vec{r}}$ associated with the eigenvalue $m$ is denoted by $|S,m\>_\vec{r}$ or $|m\>_\vec{r}$ when $S$ is clear from the context. When $m=S$ ($m=-S$), the eigenstate is also denoted by $|+\>_{\vec{r}}$ ($|-\>_\vec{r}$). 
The projector onto $|m\>_\vec{r}$ can be expressed as 
\begin{align}
|m\>_\vec{r}\<m|=\prod_{k=-S, k\neq m}^S \frac{S_{\vec{r}}-k}{m-k}, \label{eq:SpinProj}
\end{align}
which implies that
\begin{equation}
\begin{aligned}
|+\>_\vec{r}\<+|&=|S\>_\vec{r}\<S|=\prod_{k=-S}^{S-1} \frac{S_{\vec{r}}-k}{S-k},\\
|-\>_\vec{r}\<-|&=|-S\>_\vec{r}\<-S|=\prod_{k=-S+1}^{S} \frac{S_{\vec{r}}-k}{-S-k}. \label{eq:eigSpm}
\end{aligned}
\end{equation}
In the special case $S=1/2$, \eref{eq:eigSpm} yields
\begin{equation}
|\pm\>_\vec{r}\<\pm|=\Bigl|\pm\frac{1}{2}\Bigr\>_\vec{r}\Bigl\<\pm \frac{1}{2}\Bigr|=\frac{1}{2}\pm S_\vec{r}= \frac{1\pm\vec{r}\cdot\vec{\sigma}}{2},\\
\end{equation}
where $\vec{\sigma}=(\sigma_x,\sigma_y,\sigma_z)$ is the vector composed of the three Pauli operators. When $S=1$ by contrast, \eref{eq:SpinProj} yields
\begin{equation}\label{eq:proj}
|\pm\>_{\vec{r}}\<\pm|=\frac{S_{\vec{r}}(S_{\vec{r}}\pm 1)}{2},\quad   |0\>_{\vec{r}}\<0|= 1-S_\vec{r}^2.
\end{equation}

In addition, given two unit vectors $\bm{r},\bm{s}$ in dimension 3,
the fidelities between $|\pm\>_\vec{r}$ and $|\pm\>_\vec{s}$ read (cf. Sec. III D in \rcite{ZhanFG90})
\begin{equation}\label{eq:rsOverlap}
\begin{split}
|\sideset{_{\bm{r}}}{_{\bm{s}}}{\mathop{\<+|+\>}}|^2=|\sideset{_{\bm{r}}}{_{\bm{s}}}{\mathop{\<-|-\>}}|^2=\Bigl(\frac{1+\bm{r}\cdot\bm{s}}{2}\Bigr)^{2S},\\
|\sideset{_{\bm{r}}}{_{\bm{s}}}{\mathop{\<+|-\>}}|^2=|\sideset{_{\bm{r}}}{_{\bm{s}}}{\mathop{\<-|+\>}}|^2=\Bigl(\frac{1-\bm{r}\cdot\bm{s}}{2}\Bigr)^{2S}. 
\end{split}
\end{equation}
When $S=1/2$, \eref{eq:rsOverlap} yields the familiar fidelity formula for a qubit,
\begin{equation}\label{eq:rsOverlapS1}
\begin{split}
|\sideset{_{\bm{r}}}{_{\bm{s}}}{\mathop{\<+|+\>}}|^2=|\sideset{_{\bm{r}}}{_{\bm{s}}}{\mathop{\<-|-\>}}|^2=\frac{1+\bm{r}\cdot\bm{s}}{2},\\
|\sideset{_{\bm{r}}}{_{\bm{s}}}{\mathop{\<+|-\>}}|^2=|\sideset{_{\bm{r}}}{_{\bm{s}}}{\mathop{\<-|+\>}}|^2=\frac{1-\bm{r}\cdot\bm{s}}{2}. 
\end{split}
\end{equation}

\subsection{Graph basics}
A graph $G(V,E)$ is specified by a vertex set $V$ and edge set $E$, where each edge is a two-vertex subset of $V$ (here we only consider graphs without loops) \cite{Volo09book}. Two distinct vertices $j,k\in V$ are adjacent if $\{j,k\}\in E$, in which case $j,k$ are also called neighbors. The degree of a vertex $j$ is the number of its neighbors and is denoted by $\deg(j)$. The degree of $G$ is the maximum vertex degree and is denoted by $\Delta(G)$. The graph $G$ is $k$-regular if all the vertices have  degree $k$. The graph $G$ is connected if 
for any pair of distinct vertices $i,j$, there exists a sequence of vertices $i_1, i_2, \ldots,i_h$ with $i_1=i$ and $i_h=j$ such that each pair of consecutive vertices are adjacent, that is, $\{i_k, i_{k+1}\}\in E$  for $k=1, 2, \ldots, h-1$.

Two distinct edges of $G$ are adjacent if they share a same vertex and nonadjacent otherwise. A \emph{matching} $M$ is a subset of $E$ in which no two edges are adjacent.  The matching is a \emph{maximal matching} if it is not contained in any other matching; it is a \emph{maximum matching} if it contains the largest number  of edges. The \emph{matching number} is the cardinality of a maximum matching and is denoted by $\upsilon(G)$. 
A \emph{matching cover} $\scrM$ is a set of matchings that covers the edge set $E$, which means $\cup_{M\in \scrM}=E$. An \emph{edge coloring} of $G$ is an assignment of colors to its edge such that only nonadjacent edges can have the same color. An edge coloring is trivial if all edges have different colors. 
 By definition, each edge coloring determines a matching cover composed of disjoint matchings, and vice versa. 
The \emph{chromatic index} (or edge chromatic number) of $G$ is the minimum number of colors required to color the edges of $G$  and is denoted by $\chi'(G)$. Meanwhile, $\chi'(G)$ is also the minimum number of matchings required to cover the edge set $E$. According to Vizing's theorem \cite{Vizi64,Volo09book}, the chromatic index of $G$ satisfies 
\begin{align}
\Delta(G)\leq \chi'(G)\leq \Delta(G)+1. \label{eq:VizingBound}
\end{align}

\subsection{AKLT Hamiltonians and AKLT States}
Let  $G(V,E)$ be a connected graph with $n$ vertices, 
to define the AKLT Hamiltonian associated with this graph, we first assign a Hilbert space $\caH_j$ of dimension $\deg(j)+1$ to each vertex $j$. The whole Hilbert space is a tensor product of $\caH_j$, that is, $\caH=\bigotimes_{j\in V}\caH_j$.  
Then we can assign a spin operator $\vec{S}_j=(S_{j,x}, S_{j,y}, S_{j,z})$ of spin value $S_j=\deg(j)/2$ on $\caH_j$. Next, for each edge $e=\{j,k\}\in E$ of the graph, define $S_e=S_j+S_k$, then $S_e$ is the maximum possible value of the total spin of the two nodes. Define
\begin{align}
S_E:=\max_{e\in E}S_e=\max_{\{j,k\}\in E}(S_j+S_k); \label{eq:S_E}
\end{align}
then we have
\begin{align}\label{eq:SEDeltaG}
S_E\leq \Delta(G),
\end{align}
where the inequality is saturated if $G$ is regular.

Given an edge $e=\{j,k\}$,  denote by  $P_e=P_{S_e}(\vec{S}_j+\vec{S}_k)$ the projector onto the spin-$S_e$ subspace of spins $j$ and $k$. To be concrete, the projector can be expressed as follows,
 \begin{align}
 P_e=P_{S_e}(\vec{S}_j+\vec{S}_k)=\prod_{l=|S_j-S_k|}^{ S_j+S_k-1}\frac{(\vec{S}_j+\vec{S}_k)^2 -l(l+1)}{S_e(S_e+1)-l(l+1)},\label{eq:Pe}
 \end{align}
where
\begin{align}
(\vec{S}_j+\vec{S}_k)^2&=(S_{j,x}+S_{k,x})^2+(S_{j,y}+S_{k,y})^2\nonumber\\
&\quad +(S_{j,z}+S_{k,z})^2.
\end{align}
When $S_j=S_k=1$, \eref{eq:Pe} can be simplified as 
\begin{equation}\label{eq:spin2}
P_2(\vec{S}_j+\vec{S}_k)=\frac{\bm{S}_{j}\cdot \bm{S}_{k}}{2}+\frac{\left(\bm{S}_{j}\cdot \bm{S}_{k}\right)^2}{6}+\frac{1}{3}. 
\end{equation}
When $S_j=S_k=3/2$, \eref{eq:Pe} can be simplified as 
\begin{align}\label{eq:P3}
P_3(\bm{S}_i+\bm{S}_j)&=\frac{27}{160}\bm{S}_i \cdot \bm{S}_j+\frac{29}{360}(\bm{S}_i \cdot \bm{S}_j)^2\nonumber\\
&\quad+\frac{1}{90}(\bm{S}_i \cdot \bm{S}_j)^3+\frac{11}{128}.
\end{align}

Now the AKLT Hamiltonian associated with the graph $G(V,E)$ can be expressed as 
\begin{align}
H_G=\sum_{e\in E} P_e=\sum_{\{j,k\}\in E}P_{S_j+S_k}(\vec{S}_j+\vec{S}_k). \label{eq:H_G}
\end{align}
It is known that this Hamiltonian is frustration free and has a unique ground state \cite{KiriK09,XuK08,KoreX10}, which is denoted by $|\Psi_G\>$ henceforth.  By definition  we have 
\begin{align}
P_e |\Psi_G\>=0 \quad \forall e\in E. \label{eq:Psi_G}
\end{align}
Moreover, $|\Psi_G\>$ is the only state (up to an irrelevant overall phase factor) that satisfies this condition.

Suppose $e,e'\in E$ are two edges of $G$; then $P_e$ and $P_{e'}$ commute unless the intersection $e\cap e'$ is nonempty. Suppose $e=\{j,k\}$, then at most $\deg(j)+\deg(k)-2$ projectors $P_{e'}$ do not commute with $P_e$. Let 
\begin{align}\label{eq:g}
g&=g(H_G)=\max_{j<k} [\deg(j)+\deg(k)-2]=2\max_{e\in E} S_e-2\nonumber\\
&=2S_E-2;
\end{align}
then each projector $P_e$ does not commute with at most $g$ projectors that compose the Hamiltonian $H_G$ in \eref{eq:H_G}. By \eref{eq:SEDeltaG} we have 
\begin{align}
g\leq 2\Delta(G)-2. 
\end{align}

When $G$ is the cycle graph with $n\geq 3$ vertices, we get the prototypical 1D AKLT Hamiltonian,
\begin{align}
H^{\circ}(n):=H_G=\sum_{j=1}^n P_2(\vec{S}_j+\vec{S}_{j+1}), 
\end{align}
where the vertex label $n+1$ is identified with 1 by convention.  Explicit expression for the AKLT state $|\Psi_G\>$ is known \cite{AfflKLT87,AfflKLT88}, but it is not necessary to the current study. By contrast, the AKLT Hamiltonian for the open chain with $n$ nodes is denoted by 
\begin{align}
H_{\frac{1}{2},\frac{1}{2}}(n):=&P_{\frac{3}{2}}(\vec{S}_1+\vec{S}_2)+P_{\frac{3}{2}}(\vec{S}_{n-1}+\vec{S}_n)\nonumber\\
&+\sum_{j=2}^{n-2} P_2(\vec{S}_j+\vec{S}_{j+1});
\end{align} 
here the spin values for the two boundary spins are both equal to $1/2$ as indicated in the subscripts. For the convenience of later discussions, we also define two auxiliary Hamiltonians, 
\begin{align}
H_{\frac{1}{2},1}(n):=&P_{\frac{3}{2}}(\vec{S}_1+\vec{S}_2)+\sum_{j=2}^{n-1} P_2(\vec{S}_j+\vec{S}_{j+1}), \label{eq:Hh1} \\
H_{1,1}(n):=&\sum_{j=1}^{n-1} P_2(\vec{S}_j+\vec{S}_{j+1}),\label{eq:H11}
\end{align} 
where the subscripts indicate the spin values of the two boundary spins. Note that the ground state spaces of $H_{\frac{1}{2},1}(n)$ and $H_{1,1}(n)$ are two-fold degenerate and four-fold degenerate, respectively, in contrast with $H^{\circ}(n)$  and $H_{\frac{1}{2},\frac{1}{2}}(n)$, which are nondegenerate.

\subsection{Spectral gaps of  AKLT Hamiltonians}
The spectral gaps of AKLT Hamiltonians are of key interest in many-body physics. They also  play a crucial role in the verification of AKLT states as we shall see later. In the original papers \cite{AfflKLT87,AfflKLT88}, Affleck, Kennedy, Lieb, and Tasaki proved that the AKLT Hamiltonian on the closed chain is gapped. Recently, researchers further showed that  the AKLT Hamiltonians on several 2D lattices are also gaped. Notably, spectral gaps can be established rigorously for
 the  honeycomb lattice, decorated  honeycomb lattice, and decorated square lattice
\cite{AbduLLN20,PomaW19,PomaW20, LemmSW20,GuoPW21}. The estimated spectral gaps for the 1D chain,  honeycomb lattice, and  square lattice are 0.350, 0.100, and 0.015, respectively \cite{GarcMW13,WeiRA22}. Here we discuss briefly about the spectral gaps in the 1D case and for arbitrary connected graphs up to  five vertices.  The spectral gap of the Hamiltonian $H_G$ is denoted by $\gamma(H_G)$ or $\gamma$ for simplicity when there is no danger of confusion. 

\subsubsection{1D chains}
The spectral gaps of the AKLT Hamiltonians for four types of 1D chains up to 10 nodes are presented in \tref{tab:AKLTgap} and illustrated \fref{fig:AKLTGap}. Numerical calculation shows that the spectral gaps for open chains, corresponding to  $H_{\frac{1}{2},\frac{1}{2}}(n)$,  $H_{\frac{1}{2},1}(n)$, and $H_{1,1}(n)$, decrease monotonically with the number of nodes. In the case of the  closed chain  corresponding to $H^{\circ}(n)$, by contrast, the spectral gap decreases monotonically  if the chain has odd length, but increases monotonically if the chain has even length. In addition, numerical calculation suggests the following relations,
\begin{gather}
\gamma\bigl(H_{\frac{1}{2},\frac{1}{2}}(n)\bigr)\geq \gamma\bigl(H_{\frac{1}{2},1}(n)\bigr)\geq \gamma(H_{1,1}(n)),\quad n\geq 3,\\
\gamma\bigl(H_{\frac{1}{2},\frac{1}{2}}(n)\bigr)\geq \gamma(H^{\circ}(n)), \quad n\geq 4,n\neq5.
\end{gather}

\begin{table}[b]
	\caption{\label{tab:AKLTgap}
		Spectral gaps of the AKLT Hamiltonians $H_{\frac{1}{2},\frac{1}{2}}(n)$, $H_{\frac{1}{2},1}(n)$, $H_{1,1}(n)$, and $H^{\circ}(n)$  for $n=3,4,\ldots,10$ together with the lower bounds $\tilde{c}_n$ and $c_n$ defined in \thsref{thm:gapKnabe} and~\ref{thm:gapGM}. Note that $\tilde{c}_k$ is a  lower bound for $\gamma(H^{\circ}(n))$ when $n>k$, while $c_k$ is a  lower bound for $\gamma(H^{\circ}(n))$  when $n>2k$.}
		\renewcommand\arraystretch{1.2}
	\begin{ruledtabular}
		\begin{tabular}{c|cccccccc}
			$n$	& 3 & 4 & 5 & 6	& 7 & 8 & 9 & 10\\
			\hline
			$H_{\frac{1}{2},\frac{1}{2}}$ & 0.667 & 0.517 & 0.454 & 0.421 & 0.402 & 0.390 & 0.381 & 0.376\\
			$H_{\frac{1}{2},1}$ & 0.592 & 0.473 & 0.431 & 0.408 & 0.393 & 0.384 & 0.377 & 0.372\\
			$H_{1,1}$ & 0.500 & 0.449 & 0.413 & 0.398 & 0.387 & 0.379 & 0.374 & 0.367\\
			$H^{\circ}$	& 0.833 & 0.333 & 0.454 & 0.348 & 0.402 & 0.350 & 0.381 & 0.350\\			
			$\tilde{c}_n$ & 0 & 0.173 & 0.218 & 0.248 & 0.264 & 0.276 & 0.284 & 0.291\\				
			$c_n$ & 0 & 0.207 & 0.254 & 0.280 & 0.290 & 0.296 & 0.299 & 0.301\\
		\end{tabular}
	\end{ruledtabular}
\end{table}

\begin{figure}
	\includegraphics[width=0.42\textwidth]{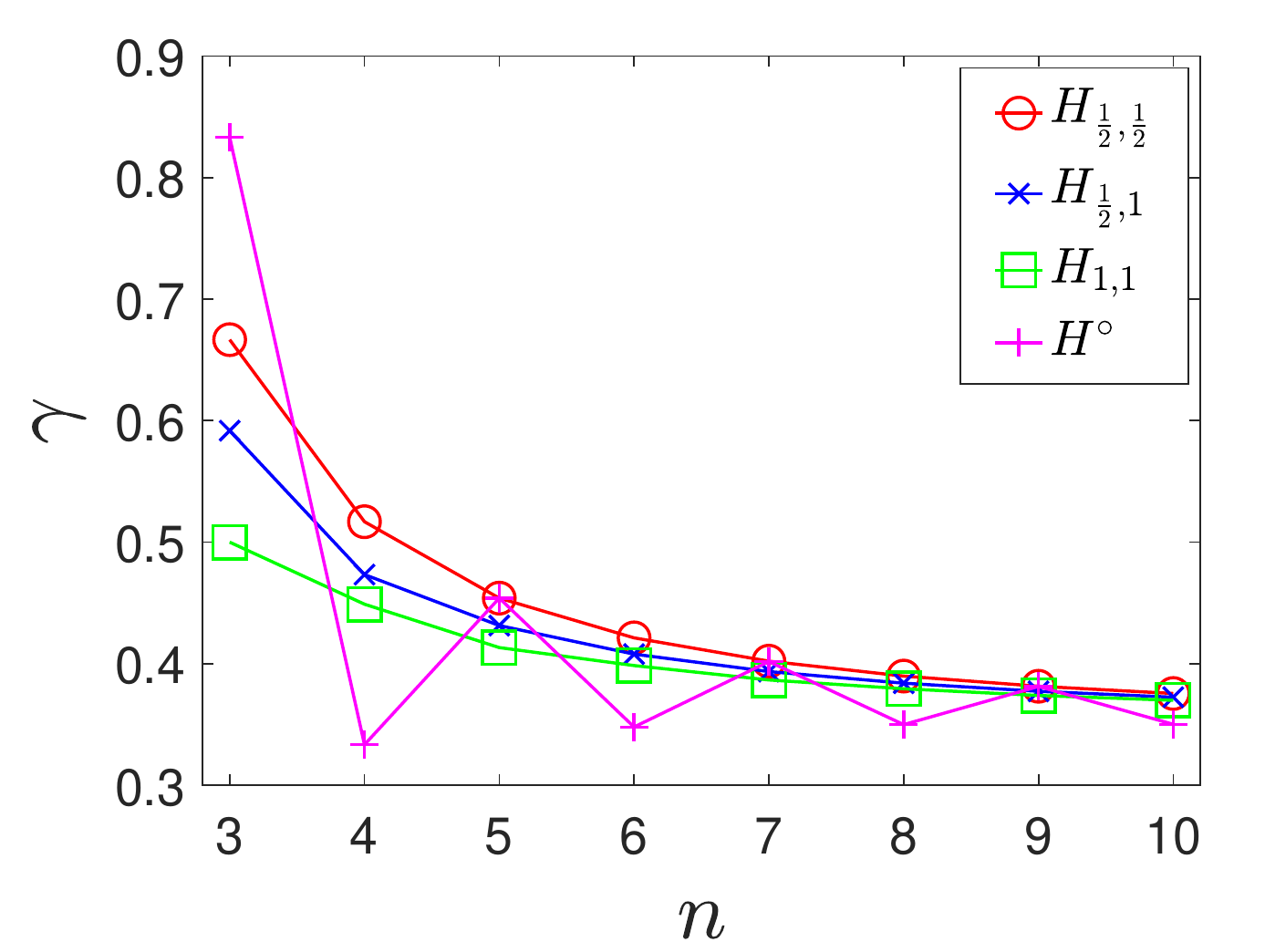}
	\caption{\label{fig:AKLTGap}Spectral gaps of the AKLT Hamiltonians $H_{\frac{1}{2},\frac{1}{2}}(n)$, $H_{\frac{1}{2},1}(n)$, $H_{1,1}(n)$, and $H^{\circ}(n)$  for $n=3,4,\ldots,10$. Here lines are guides for the eyes (similarly for other figures). }	
\end{figure}

In the thermodynamic limit, the AKLT Hamiltonian on the closed chain was shown to be gaped already in the original work of AKLT \cite{AfflKLT87,AfflKLT88}. 
Later, Knabe  derived a lower bound for $\gamma(H^{\circ}(n))$ based on the spectral gap $\gamma(H_{1,1}(k))$ \cite{Knab88}, where $H_{1,1}(k)$ is the Hamiltonian associated with the open chain of $k$ nodes as defined in \eref{eq:H11}.
\begin{thm}[Knabe]\label{thm:gapKnabe}
	Suppose  $n>k>2$. Then $\gamma(H^{\circ}(n))\geq \tilde{c}_k$ with
\begin{equation}
 \tilde{c}_k:= \left(\frac{k-1}{k-2}\right)\left[\gamma(H_{1,1}(k))-\frac{1}{k-1}\right].
\end{equation}
\end{thm}
The Knabe's bound above  is nontrivial whenever $\gamma(H_{1,1}(k))>1/(k-1)$. Recently, Gosset and Mozgunov  proved a stronger result as stated in  the following theorem~\cite{GossM16}. 
\begin{thm}[Gosset-Mozgunov]\label{thm:gapGM}
	Suppose  $k>2$ and $n>2k$. Then $\gamma(H^{\circ}(n))\geq  c_k$ with
\begin{equation}\label{eq:ck}
c_k:= \frac{5}{6}\left(\frac{k^2+k}{k^2-4}\right)\left[\gamma(H_{1,1}(k))-\frac{6}{k(k+1)}\right].
\end{equation}
\end{thm}
Numerical calculation suggests that \thsref{thm:gapKnabe} and \ref{thm:gapGM} still hold if $H^{\circ}(n)$ is replaced by $H_{\frac{1}{2},\frac{1}{2}}(n)$, $H_{\frac{1}{2},1}(n)$, or  $H_{1,1}(n)$ (cf. \tref{tab:AKLTgap}).

\subsubsection{general graphs}

\begin{figure}[t]
	\includegraphics[width=0.46\textwidth]{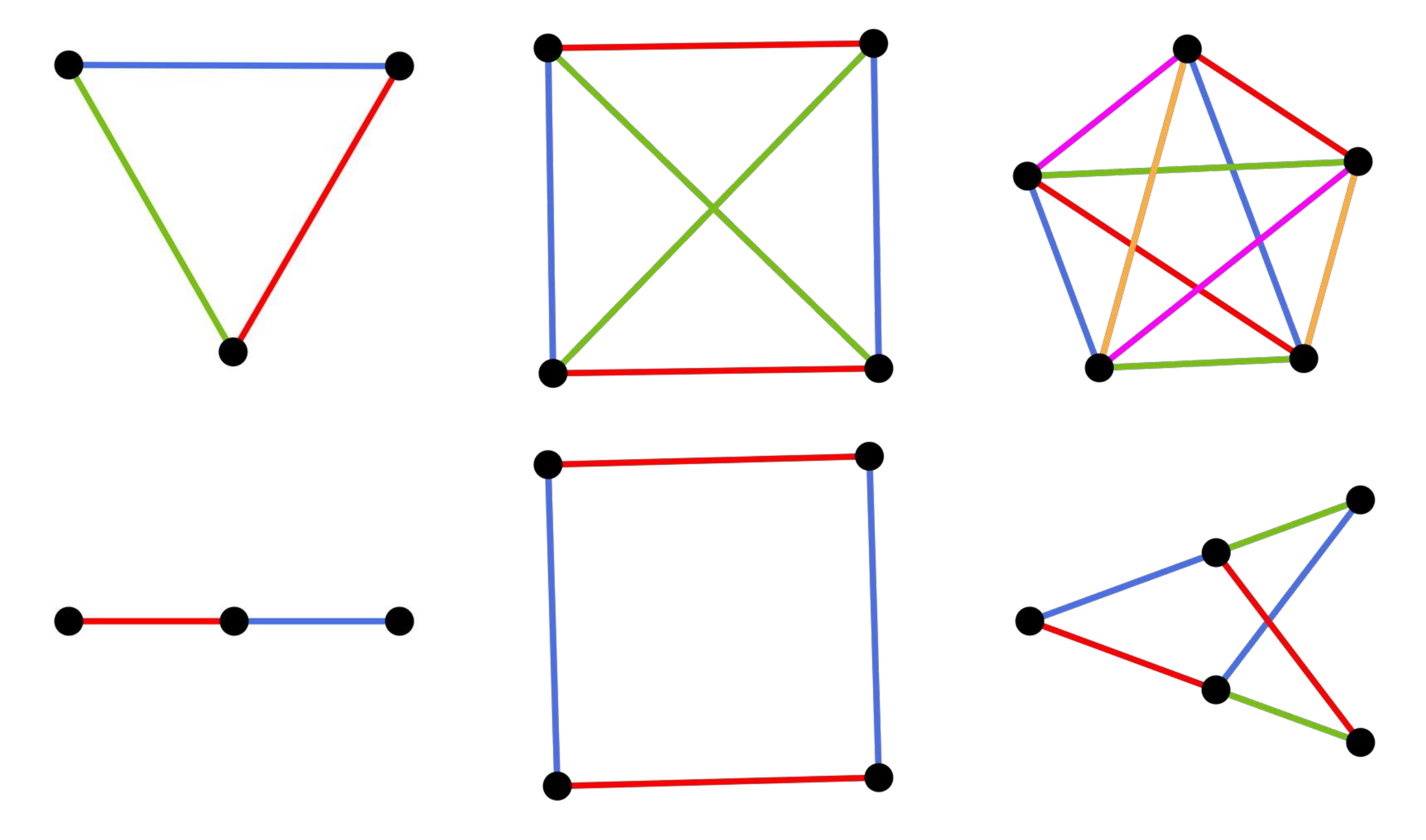} 	
	\caption{\label{fig:gammaMaxMin}Connected graphs of three, four, and five vertices whose AKLT Hamiltonians have the largest spectral gaps (up) and smallest spectral gaps (down).}
\end{figure}

Next, we consider the spectral gaps of AKLT Hamiltonians associated with general connected graphs $G(V,E)$ up to five vertices, that is, $n=|V|\leq 5$. Up to isomorphism there are 1 connected graph of two  vertices,  2 connected graphs of three  vertices,  6 connected graphs of four  vertices, and  21 connected graphs of five  vertices. The corresponding spectral gaps are presented in \tref{tab:VAKLTgen} in \aref{app:VAKLTgen}.  Calculation shows that  the Hamiltonian associated with the complete graph has the largest spectral gap among all graphs with the same number of vertices, as illustrated in \fref{fig:gammaMaxMin}. To be specific, the spectral gaps are $5/6$, $7/10$, and $3/5$ for complete graphs of three, four, and five vertices, respectively.  We guess the same conclusion holds even for graphs with more than five vertices. We have not found a general pattern for the graph that leads to the smallest spectral gap: the minimum is attained at the linear graph when  $n=3$, the cycle graph when $n=4$, and the lower-right graph shown in   \fref{fig:gammaMaxMin} (corresponding to graph No. 18 in  \tref{tab:VAKLTgen}) when  $n=5$.

\section{\label{sec:BondVerify}Bond verification protocols}
To verify the AKLT state associated with a given graph, we first need to construct bond verification protocols for verifying each pair of adjacent nodes, which is tied to the problem of subspace verification. Here we shall focus on  verification protocols that  build on spin measurements.

\subsection{\label{sec:TestOperator}Test operators based on spin measurements}

Let $\vec{S}_1$ and $\vec{S}_2$ be two spin operators of spin values $S_1$ and $S_2$,   respectively. Let $S=S_1+S_2$ and let $P_S=P_S(\vec{S}_1+\vec{S}_2)$ be the projector onto the subspace associated with the maximum total spin $S$ as defined in \eref{eq:Pe}; let $Q=1-P_S$. Here our goal is to verify the null space of $P_S$, that is, the support of $Q$. To this end, we shall construct test operators based on spin measurements.

Suppose the two parties perform spin measurements along  directions $\vec{r}, \vec{s}$, respectively, where $\vec{r}, \vec{s}$ are real unit vectors in dimension 3. The measurement outcomes can be labeled  by  eigenvalues $m_1$ and $m_2$ of $\vec{r}\cdot \vec{S}_1$ and $\vec{s}\cdot \vec{S}_2$, respectively. Let
\begin{align}
|S_1,m_1; S_2,m_2\>_{\vec{r},\vec{s}}:=&|S_1,m_1\>_\vec{r}\otimes |S_2,m_2\>_\vec{s},\\
p_{\vec{r},\vec{s}}(S_1,m_1;S_2,m_2):=&\|P_S |S_1,m_1; S_2,m_2\>_{\vec{r},\vec{s}}\|^2\nonumber\\
=&\tr[P_S P_{\vec{r},\vec{s}}(S_1,m_1; S_2,m_2)],
\end{align}
where $P_{\vec{r},\vec{s}}(S_1,m_1; S_2,m_2)$ is the rank-1 projector onto the state $|S_1,m_1; S_2,m_2\>_{\vec{r},\vec{s}}$.
Then a general test operator $R$ is a linear combination of the  projectors $P_{\vec{r},\vec{s}}(S_1,m_1; S_2,m_2)$ associated with all possible outcomes.
To guarantee that all states supported in the support of $Q$ can pass the test with certainty, $R$ should satisfy the condition
\begin{align}
R\geq P_{\vec{r},\vec{s}}(S_1,m_1; S_2,m_2) \label{eq:TestCondition}
\end{align}
whenever $p_{\vec{r},\vec{s}}(S_1, m_1; S_2,m_2)<1$.

If $\vec{r}$ is neither parallel nor antiparallel to $\vec{s}$, then the inequality in \eref{eq:TestCondition} should hold for all possible outcomes $m_1, m_2$ according to \lref{lem:SpinProjProb} below. So the test operator $R$ is equal to the identity operator, and the test is trivial. When $\vec{s}$ is parallel to $\vec{r}$,  nontrivial test operators can be constructed.  Here we are particularly interested in the canonical test projector associated with $\vec{r}$ as defined as follows,
\begin{align}
R_\vec{r}:=&1-P_{\vec{r}}(S_1; S_2)-P_{\vec{r}}(-S_1; -S_2)\nonumber\\
=&1-|++\>_{\vec{r}}\<++|-|--\>_{\vec{r}}\<--|, \label{eq:TestCanonical}
\end{align}
where 
\begin{align}
P_{\vec{r}}(S_1; S_2):=&P_{\vec{r},\vec{r}}(S_1,S_1; S_2,S_2),\\
P_{\vec{r}}(-S_1; -S_2):=&P_{\vec{r},\vec{r}}(S_1,-S_1; S_2,-S_2),\\ 
|\pm\pm\>_{\vec{r}}:=&|\pm S_1\>_\vec{r}\otimes |\pm S_2\>_\vec{r}.
\end{align}
According to \lref{lem:SpinProjProb} below, any other test operator $R$ based on the same spin measurements satisfy $R\geq R_\vec{r}$ and is thus suboptimal. When $\vec{s}=-\vec{r}$, the set of accessible test operators does not change given that 
\begin{align}
P_{\vec{r},\vec{s}}(S_1,m_1; S_2,m_2)&=P_{\vec{r},-\vec{s}}(S_1,m_1; S_2, -m_2). 
\end{align}
Therefore, it suffices to consider canonical test projectors based on parallel spin measurements. 

The following lemma employed in the above analysis is proved in \aref{sec:SpinProjProbProof}. 
\begin{lem}\label{lem:SpinProjProb}
Suppose $S_1$ and $S_2$ are positive integers or half integers. Then 
\begin{align}
p_{\vec{r},\vec{s}}(S_1,m_1; S_2,m_2)\leq 1,
\end{align}
and the inequality is saturated iff one of the following four conditions holds
\begin{subequations}
\begin{align}
\vec{r}&=\vec{s}, & m_1&=S_1, & m_2&=S_2; \label{eq:SpinCa}\\
\vec{r}&=\vec{s}, & m_1&=-S_1, & m_2&=-S_2; \label{eq:SpinCb} \\
\vec{r}&=-\vec{s},& m_1&=S_1,& m_2&=-S_2;  \label{eq:SpinCc}\\
\vec{r}&=-\vec{s},& m_1&=-S_1,& m_2&=S_2.  \label{eq:SpinCd}
\end{align}
\end{subequations}
\end{lem}

\subsection{Spectral gaps of bond verification protocols}
According to the discussion in \sref{sec:TestOperator}, each canonical test projector is specified by a unit vector $\vec{r}$ in dimension 3. Given any probability distribution $\mu$ on the unit sphere, then  a bond verification protocol can be constructed by performing each test $R_{\vec{r}}$ with a suitable probability. 
The corresponding   verification operator reads
\begin{align}
\Omega_{S_1, S_2}(\mu)=\int R_{\vec{r}} d \mu(\bm{r}). \label{eq:BondVO}
\end{align}
When $S_1$ and $S_2$ are clear from the context, $\Omega_{S_1, S_2}(\mu)$ can be abbreviated as $\Omega(\mu)$ for simplicity. 
The spectral gap of $\Omega_{S_1, S_2}(\mu)$ reads
\begin{align}
\nu(\Omega_{S_1, S_2}(\mu))=1-\|P_S\Omega_{S_1, S_2}(\mu))P_S\|, \label{eq:SpectralGapBond}
\end{align}
where $P_S$ is the projector defined according to \eref{eq:Pe}.
Note that the spectral gap $\nu(\Omega_{S_1, S_2}(\mu))$ is invariant when $\mu$ is subjected to any orthogonal transformation; in addition, $\nu(\Omega(\mu))$ is concave in $\mu$. 
\begin{lem}\label{lem:SpectralGapBond}
The spectral gap $\nu(\Omega_{S_1, S_2}(\mu))$ is independent of $S_1$ and $S_2$ once the sum $S=S_1+S_2$ is fixed. 
\end{lem}	
\Lref{lem:SpectralGapBond} follows from the definition of the test operator $R_\vec{r}$ in \eref{eq:TestCanonical}  and the fact that the representation of the angular momentum operators carried by the states $|S_1, S_2\>_{\vec{r}}$ is independent of $S_1$ and $S_2$ once the sum $S=S_1+S_2$ is fixed.
This result holds even if $S_1=0$ or $S_2=0$, which is very helpful to simplify the computation of the spectral gap. In view of these facts, we shall denote the spectral gap of $\Omega_{S_1, S_2}(\mu)$ by $\nu_S(\mu)$ for simplicity.

Denote by $\mu_\sym$ the average distribution of $\mu$ and its center inversion.
Define 
\begin{align}
\Omega_S(\mu):=&\int(\,1-|S\>_{\vec{r}}\<S|-|-S\>_{\vec{r}}\<-S|\,) d\mu(\vec{r})\nonumber\\
=&
1-2\int |S\>_{\vec{r}}\<S| d\mu_\sym(\vec{r}). \label{eq:OmegaS}
\end{align}
Then 
\begin{align}
\nu_S(\mu)=1-\|\Omega_S(\mu)\|=\lambda_{\min}(O_S(\mu)), \label{eq:nuSmu}
\end{align}
where
\begin{align}
O_S(\mu):=2\int |S\>_{\vec{r}}\<S| d\mu_\sym(\vec{r}), \label{eq:MSmu}
\end{align}
and $\lambda_{\min}$ denotes the smallest eigenvalue. In particular,  $\nu_S(\mu)$ is nonzero iff the operator  $O_S(\mu)$ has full rank.
When $\mu$ is a discrete distribution, to achieve a nonzero spectral gap $\nu_S(\mu)>0$,
 the support of $\mu_\sym$ should contain at least $2S+1$ points, so the support of $\mu$ should contain at least $\lceil S+\frac{1}{2}\rceil$ points. To construct a nontrivial bond verification protocol, therefore, at least $\lceil S+\frac{1}{2}\rceil$ distinct canonical tests are required.

The following lemma proved in \aref{asec:nuSmuProof} is very instructive to understanding the properties of $\nu_S(\mu)$. 
\begin{lem}\label{lem:nuSmu}
Suppose $\mu$ is a probability distribution on the unit sphere. Then $\nu_S(\mu)$ is nonincreasing in $S$. If $S_1\leq S_2$, then 
\begin{align}
\nu_{S_1}(\mu)\geq \frac{2S_2+1}{2S_1+1} \nu_{S_2}(\mu). 
 \label{eq:nuSmuIneq} 
\end{align}
\end{lem}

In the special case $S=1/2$, we have  $\nu_S(\mu)=1$ irrespective of the distribution $\mu$. So \lref{lem:nuSmu} implies that 
\begin{align}
 \nu_S(\mu)\leq \frac{2}{2S+1}, \label{eq:nuMax}
\end{align}
which sets an upper bound for the spectral gap achievable by spin measurements. 
Alternatively, \eref{eq:nuMax} follows from  \eref{eq:OmegaS}, which implies that 
\begin{align}
\tr[\Omega_S(\mu)]=2S-1, \quad \|\Omega_S(\mu)\|\geq \frac{2S-1}{2S+1}. 
\end{align}
Bond verification protocols that saturate the upper bound in \eref{eq:nuMax} are called optimal. Notably, this bound 
is saturated when  $\mu$ is the isotropic (uniform) distribution on the unit sphere, which  leads to  the \emph{isotropic protocol}.

To clarify the condition required for constructing an optimal bond verification protocol, we need to introduce additional concepts. 
Let $t$ be a nonnegative integer. A probability distribution $\mu$ on the unit sphere is a (spherical) $t$-design if the average of any polynomial of  degree less than or equal to 
 $t$ over the distribution is equal to the average over the isotropic distribution \cite{DelsGS77,Seid01,BannB09}. By definition a $t$-design is automatically a $(t-1)$-design for any positive integer $t$. The design strength of the distribution $\mu$ is the largest integer $t$ such that $\mu$ is a $t$-design.
The isotropic distribution  forms a spherical $\infty$-design and has strength $\infty$. If $\mu$ is center symmetric, then $\mu$ is a $2j$-design iff $\mu$ is a $(2j+1)$-design for any positive integer $j$, so the strength of $\mu$ is always an odd integer. The next theorem follows from a similar reasoning used to establish Theorem 3 in the companion paper \cite{ZhuLC22}. A self-contained proof is presented in \aref{asec:OmegaSdesign}. 
 \begin{thm}\label{thm:OmegaSdesign}
	Let $\mu$ be a probability distribution on the unit sphere and $S$ a positive integer or half integer. Then the following four statements are equivalent. 
	\begin{enumerate}
		
		\item $\nu_S(\mu)=\frac{2}{2S+1}$.
		
		\item $\Omega_S(\mu)=\frac{2S-1}{2S+1}$.
		
		\item $\Omega_S(\mu)$ is proportional to  the identity operator. 
		
		\item $\mu_\sym$ forms a spherical $t$-design with $t=2S$. 
	\end{enumerate}
\end{thm}
According to the  discussion before \thref{thm:OmegaSdesign}, when $S$ is a positive integer, $\mu_\sym$ forms a  $2S$-design iff it forms a  $(2S+1)$-design; when $S$ is a positive half integer, $\mu_\sym$ forms a  $2S$-design iff it forms a  $2\lfloor S\rfloor$-design.

\subsection{\label{sec:BondVC}Concrete verification protocols}

\begin{table*}
	\caption{\label{tab:BondSpectralGap}
		Spectral gaps $\nu_S(\mu)$ for $1\leq S\leq 4$  of bond verification protocols based on platonic solids, $\mu_{24}$, $\mu_{32}$, and the isotropic distribution.  Vertex number,  distinct test number, and design strength of each distribution are also shown for completeness.
	}
	\renewcommand\arraystretch{1.2}	
	\begin{math}
	\begin{array}{c|ccccccccccc}
	\hline\hline 
	\mbox{Protocol} &\Omega & \mbox{Vertex number}  & \mbox{Test number} & \mbox{design strength}  &\nu_1  &\nu_{3/2} &\nu_{2} & \nu_{5/2} &\nu_3 &\nu_{7/2}&\nu_4 \\[0.5ex]
	\hline
	\mbox{Tetrahedron}&\Omega_\rmt	&4&4& 2 &\frac{2}{3} &\frac{1}{2} & \frac{1}{3}&\frac{5}{18} &\frac{5}{27}& \frac{5}{54} &0\\[0.8ex]
	\mbox{Octahedron}&\Omega_\rmo	&6&3& 3 &\frac{2}{3} &\frac{1}{2}&\frac{1}{3}&\frac{1}{6}& 0&0&0\\[0.8ex]
	\mbox{Cube}   &\Omega_\rmc   	&8&4&3 & \frac{2}{3} &\frac{1}{2}&\frac{1}{3}&\frac{5}{18}&\frac{5}{27}&\frac{5}{54}&0\\[0.8ex]
	\mbox{Icosahedron}&\Omega_\rmi	&12&6&5 &\frac{2}{3} &\frac{1}{2}&\frac{2}{5}&\frac{1}{3}&\frac{4}{15}&\frac{7}{30}&\frac{14}{75}\\[0.8ex]
	\mbox{Dodecahedron}&\Omega_\rmd	&20&10&5 &\frac{2}{3} &\frac{1}{2}&\frac{2}{5}&\frac{1}{3}&\frac{5}{18}&\frac{2}{9}&\frac{16}{81}\\[0.8ex]
	\mu_{24}&\Omega(\mu_{24})	&24&24&7 &\frac{2}{3} &\frac{1}{2}&\frac{2}{5}&\frac{1}{3}&\frac{2}{7}&\frac{1}{4}&\frac{23}{105}\\[0.8ex]
	\mu_{32} & \Omega(\mu_{32})	&32&16&9 &\frac{2}{3} &\frac{1}{2}&\frac{2}{5}&\frac{1}{3}&\frac{2}{7}&\frac{1}{4} & \frac{2}{9}\\[0.8ex]
	\mbox{Isotropic} &\Omega_{\mathrm{iso}}	&\infty&\infty&\infty &\frac{2}{3} &\frac{1}{2}&\frac{2}{5}&\frac{1}{3}&\frac{2}{7}&\frac{1}{4} & \frac{2}{9}\\[0.5ex]
	\hline\hline
	\end{array}	
	\end{math}	
\end{table*}

In this section we construct a number of concrete bond verification protocols based on discrete distributions on the unit sphere, which are  appealing to practical applications.

First, we consider bond verification protocols based on platonic solids. 
Each platonic solid inscribed in the unit sphere determines a probability distribution on the unit sphere (by convention all vertices have the same weight), which in turn determines a bond verification protocol for any given pair of spins. In this way we can construct five bond verification protocols by virtue of the five platonic solids. To be concrete, the vertices of the regular tetrahedron are chosen to be:
\begin{equation}
\begin{aligned}
&\frac{1}{\sqrt{3}}\left(1,1,1\right), & \frac{1}{\sqrt{3}}\left(1,-1,-1\right),\\
 & \frac{1}{\sqrt{3}}\left(-1,1,-1\right), &\frac{1}{\sqrt{3}}\left(-1,-1,1\right).\\
\end{aligned}
\end{equation}
The vertices of the octahedron are chosen to be:
\begin{equation}
\begin{array}{ccc}
\left(\pm 1,0,0\right), & \left(0,\pm 1,0\right), & \left(0,0,\pm 1\right).\\
\end{array}
\end{equation}
The vertices of the cube are chosen to be:
\begin{equation}
\frac{1}{\sqrt{3}}\left(\pm 1,\pm 1,\pm 1\right).
\end{equation}
The vertices of the icosahedron are chosen to be:
\begin{equation}\label{eq:ico}
\begin{aligned}
&\frac{1}{\sqrt{1+b^2}}(\pm 1,\pm b,0), \quad  \frac{1}{\sqrt{1+b^2}}(\pm b,0,\pm 1),\\ &\frac{1}{\sqrt{1+b^2}}(0,\pm 1,\pm b),
\end{aligned}
\end{equation}
where $b=(1+\sqrt{5})/2$.
The vertices of the dodecahedron are chosen to be:
\begin{equation}\label{eq:de}
\begin{aligned}
&\frac{1}{\sqrt{3}}\biggl(\pm b,\pm\frac{1}{b},0\biggr), & \frac{1}{\sqrt{3}}\biggl(\pm\frac{1}{b},0,\pm b\biggr), \\
&\frac{1}{\sqrt{3}}\biggl(0,\pm b,\pm\frac{1}{b}\biggr), & \frac{1}{\sqrt{3}}(\pm 1,\pm 1,\pm 1).
\end{aligned}
\end{equation}

For the convenience of the following discussions, the verification operators associated with the regular tetrahedron,  octahedron, cube,   icosahedron, and dodecahedron, are denoted by  $\Omega_{\rmt}$,  $\Omega_{\rmo}$, $\Omega_{\rmc}$,  $\Omega_{\rmi}$, and $\Omega_{\rmd}$, respectively. Although these verification operators may depend on the specific choices of vertices, their spectral gaps as shown in \tref{tab:BondSpectralGap} are independent of the specific choices. Except for the regular tetrahedron, every platonic solid is center symmetric, and the two tests based on each pair of antipodal vertices are equivalent; so the total number of distinct tests is equal to one half of the vertex number. It is known that  the regular tetrahedron, octahedron,  cube, icosahedron, and dodecahedron form spherical $t$-designs with $t=2, 3, 3, 5, 5$, respectively. Therefore, the icosahedron and dodecahedron protocols are optimal when $S=S_1+S_2\leq 5/2$.

To construct optimal bond verification protocols for $S\geq 3$, we need to go beyond platonic solids and consider spherical designs with higher strengths. 
For example, a spherical $7$-design can be constructed from an orbit of the rotational symmetry group  of the standard cube (which has order 24): one fiducial vector has the form $(u_1, u_2, u_3)$, where
\begin{equation}
u_j=\sqrt{\frac{1}{3}\biggl(1+2\sqrt{\frac{2}{5}}\cos\frac{\theta+2j\pi}{3}\biggr)},\quad  \theta=\arctan\frac{3\sqrt{10}}{20}
\end{equation}
for $j=1,2,3$. This orbit has 24 vectors, which can be expressed as 
\begin{align}
\{(a_1 u_{\sigma(1)}, a_2 u_{\sigma(2)},a_1 a_2\sgn(\sigma)u_{\sigma(3)})| a_1, a_2=\pm 1, \sigma\in \scrS_3 \}. 
\end{align}
Here $\scrS_3$ denotes the symmetric group of three letters; $\sgn(\sigma)$ is equal to 1 for even permutations and equal to $-1$ for odd permutations. All points corresponding to these vectors have the same weight as before; the resulting distribution on the unit sphere is denoted by $\mu_{24}$ henceforth. Note that  $\mu_{24}$  is not center symmetric.

A spherical 9-design can be constructed from the union of the vertices of the icosahedron in \eref{eq:ico} and that of the dodecahedron in  \eref{eq:de}, which form pentakis dodecahedron (cf. \rcite{HughW20}). The icosahedron has weight $5/14$ in total and each vertex has weight $5/168$, while the dodecahedron has weight $9/14$ in total and each vertex has weight $9/280$. The resulting distribution on the unit sphere is denoted by $\mu_{32}$, which is center symmetric by construction. The spectral gaps $\nu_S(\mu)$ for $1\leq S\leq 4$  of bond verification operators based on $\mu_{24}$ and $\mu_{32}$ are also shown in \tref{tab:BondSpectralGap}.

\section{\label{sec:ApproachGen}Verification of AKLT states: General approach}

\subsection{Construction of verification protocols}

\begin{figure*}
	\includegraphics[width=0.8\textwidth]{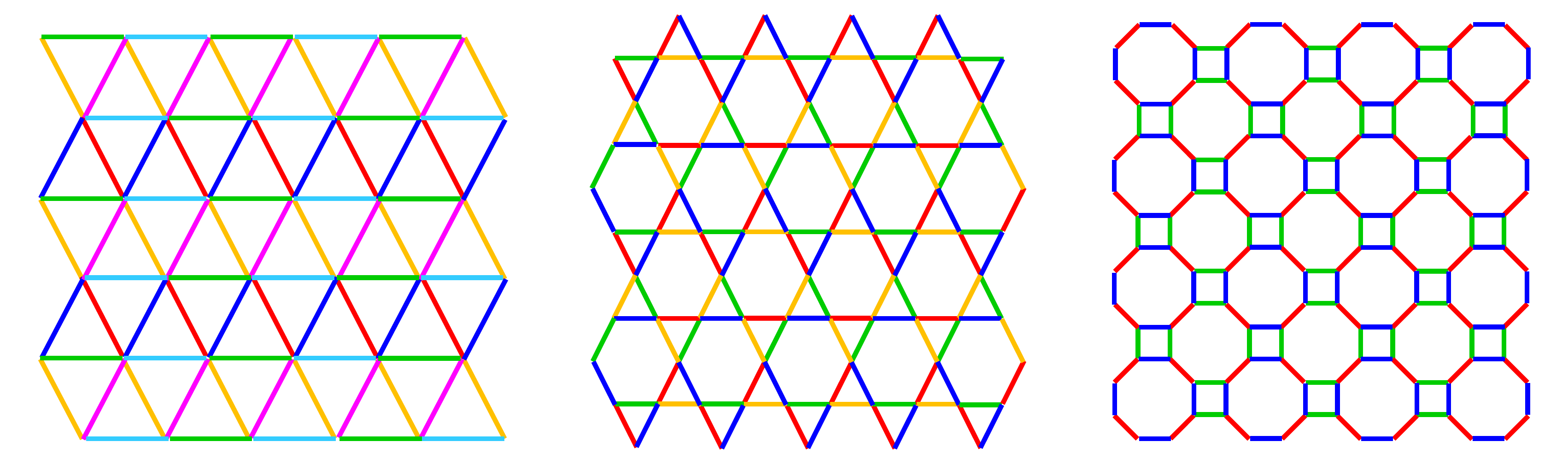}	
	\caption{\label{fig:LatticeColoring}Optimal edge colorings of several common 2D lattices:  triangular lattice, kagome lattice, and square-octagon lattice. These optimal colorings can be used to construct efficient protocols for verifying AKLT states on these lattices, which require constant sample costs that are independent of the lattice size. }	
\end{figure*}

Consider the AKLT Hamiltonian $H_G$ and AKLT state $|\Psi_G\>$ associated with a given graph $G=(V,E)$ of $n$ vertices [cf. \esref{eq:H_G}{eq:Psi_G}].
To verify $|\Psi_G\>$, we need to verify each bond associated with each edge of the graph $G$. More specifically, we need to verify the null space of the projector $P_e$ for each $e\in E$. Here we shall focus on bond verification protocols based on spin measurements, which are determined by probability distributions on the unit sphere as discussed in \sref{sec:BondVerify}. For simplicity we also assume that the same distribution is chosen for each bond, although this is not compulsory.

Let $\mu$ be  a probability distribution  on the unit sphere. According to \sref{sec:BondVerify} we can construct a bond verification protocol for each edge of the graph $G$. The bond verification operator associated with the edge $e\in E$ is denoted by $\Omega_e(\mu)$ [cf.~\eref{eq:BondVO}]. 
Given a matching $M$ of $G$, then a test for $|\Psi_G\>$ can be constructed by performing the bond verification protocols for all the bonds associated with edges in $M$ independently. 
The corresponding  test operator is given by
\begin{align}
T_M(\mu)=\prod_{e\in M} \Omega_e(\mu). \label{eq:TestMatch}
\end{align}
Note that all the bond verification operators $\Omega_e(\mu)$ for $e\in M$ commute with each other, so the order in the above product  is irrelevant. Suppose $\nu(\Omega_e(\mu))>0$ for each $e\in M$, then a quantum state $|\Phi\>$ satisfies   the condition $T_M|\Phi\>=|\Phi\>$ iff $P_e|\Phi\>=0$ for each $e\in M$. Therefore, a  state can pass the test $T_M(\mu)$  with certainty iff it belongs to the null space of the  projector $P_e$ for each $e\in M$. Let $M'$ be another matching of $G$, then we can deduce  from \eref{eq:TestMatch} the following relation:
\begin{align}
T_M(\mu)\geq T_{M'}(\mu)\quad \mbox{  if  }\quad M\subseteq M' \label{eq:TestMatchOrder}.
\end{align}

Let $\scrM=\{M_1, M_2, \ldots, M_m\}$ be a matching cover of $G$ that  consists of $m$ matchings and let $p=(p_1,p_2,\ldots, p_m)$ be a probability distribution on $\scrM$  (we shall assume that the  distribution is uniform when $p$ is not mentioned explicitly). Then a verification protocol for $|\Psi_G\>$ can be constructed by performing each test $T_{M_l}(\mu)$ with probability $p_l$. The resulting verification protocol is specified by the triple $(\mu, \scrM,p)$ (here $p$ can be omitted for the uniform distribution), and the corresponding verification operator reads
\begin{align}
\Omega(\mu, \scrM,p)=\sum_{l=1}^m p_l T_{M_l}(\mu). 
\end{align} 
Suppose $p_l>0$ for $l=1,2,\ldots, m$ and $\nu(\Omega_e(\mu))>0$ for each $e\in \cup_l M_l =E$;  then a quantum state $|\Phi\>$ can pass all the tests with certainty iff $P_e|\Phi\>=0$ for each $e\in E$. So  only the target state $|\Psi_G\>$ can pass all the tests with certainty, which means this verification protocol is effective. According to \eref{eq:TestMatchOrder}, the spectral gap of $\Omega(\mu, \scrM,p)$ does not decrease if $M_l$ is replaced by another matching $M_l'$ that contains $M_l$. To maximize the spectral gap, therefore, it is advisable to choose matching covers composed of
maximal matchings.

If the matchings in $\scrM=\{M_1, M_2, \ldots, M_m\}$ form one orbit under the symmetry group of the graph $G$, then the spectral gap of the verification operator $\Omega(\mu, \scrM,p)$ is maximized when the probability distribution $p$ is uniform. In general, given the distribution $\mu$ and the matching cover $\scrM$, the maximum spectral gap of $\Omega(\mu, \scrM,p)$ can be determined by semidefinite programming (SDP). Define
\begin{align}
\bar{T}_{M_l}(\mu)&:=T_{M_l}(\mu)-|\Psi_G\>\<\Psi_G|,\\
\bar{\Omega}(\mu, \scrM,p)&:=\Omega(\mu, \scrM,p)-|\Psi_G\>\<\Psi_G|=
\sum_{l=1}^m p_l \bar{T}_{M_l}(\mu);
\end{align}
then 
\begin{align}
\nu(\Omega(\mu, \scrM,p))=1-\|\bar{\Omega}(\mu, \scrM,p)\|.
\end{align}
To maximize the spectral gap of $\Omega(\mu, \scrM,p)$,  it is equivalent to minimize the operator norm of $\bar{\Omega}(\mu, \scrM,p)$, which can be realized by the following SDP:
\begin{equation}\label{eq:normSDP}
\begin{aligned}
&\mbox{minimize}\quad \; h\\
&\mbox{subject to}\quad  h\geq \sum_{l=1}^m p_l \bar{T}_{M_l}(\mu),\\
&\quad \quad \quad \quad \quad \;\;  p_l\geq 0,\quad  \sum_{l=1}^m p_l =1.
\end{aligned}
\end{equation}
To construct an optimal matching protocol, in principle we need to consider all maximal matchings before the optimization, which is feasible only for small systems. 
When this approach is too prohibitive, we can consider simple matching protocols and resort to analytical bounds derived in the next subsection.

\subsection{Sample complexity}
Before presenting our main results on the sample complexity, we need to introduce some terminology.
Suppose $P_e$ and $P_{e'}$ are two projectors associated with two edges of $G=(V,E)$ as defined in \eref{eq:Pe}. Denote by $s(P_e P_{e'})$ the largest singular value of $P_e P_{e'}$ that is not equal to~1. By definition $s(P_e P_{e'})=0$ if $e=e'$ or if $e$ and $e'$ are disjoint. If $e=\{1,2\}$ and $e'=\{2,3\}$, then 
\begin{align}
P_e&=P_{12}:=P_{S_1+S_2}(\vec{S}_1+\vec{S}_2),\\ P_{e'}&=P_{23}:=P_{S_2+S_3}(\vec{S}_2+\vec{S}_3). 
\end{align}
So the value of $s(P_e P_{e'})=s(P_{12}P_{23})$ is determined by the spin values $S_1, S_2, S_3$; in addition, this value is invariant if $S_1$ and $S_3$ are exchanged. The specific value of $s^2(P_{12}P_{23})$ for $S_1, S_2, S_3\leq 5/2$ can be found in \tref{tab:s2}, which suggests the following conjecture.
\begin{conjecture}
Suppose $S_1, S_2, S_3$ are positive integers or half integers; then $s^2(P_{12}P_{23})$ is a rational number. If in addition $S_1, S_3\leq S_2$, then 
\begin{align}
s^2(P_{12}P_{23})\leq 1/4,
\end{align}
and the inequality is saturated iff $S_1=S_2=S_3$.
\end{conjecture}

\begin{table}[t]
	\caption{\label{tab:s2}Value of $s^2(P_{12}P_{23})$ with $P_{12}=P_{S_1+S_2}(\vec{S}_1+\vec{S}_2)$ and $P_{23}=P_{S_2+S_3}(\vec{S}_2+\vec{S}_3)$, where $s(P_{12}P_{23})$ is the largest singular value of $P_{12}P_{23}$ that is not equal to 1, and $s^2(P_{12}P_{23})$ is the largest eigenvalue of $P_{12}P_{23}P_{12}$ that is not equal to 1.}
	\renewcommand\arraystretch{1.2}
	\begin{ruledtabular}
		\begin{tabular}{l|cccccc}
			\diagbox[width=8em,height=5.6ex,trim=l]{$(S_1,S_3)$}{$S_2$} & $\frac{1}{2}$ & $1$ & $\frac{3}{2}$ & $2$ & $\frac{5}{2}$ & $3$\\
			\hline
			\hfill\\[-2.5ex]
			\hspace{0.5em}$(\frac{1}{2},\:\frac{1}{2})$ & $\frac{1}{4}$ & $\frac{1}{9}$ & $\frac{1}{16}$ & $\frac{1}{25}$ & $\frac{1}{36}$ & $\frac{1}{49}$\\[0.8ex]
			\hspace{0.5em}$(\frac{1}{2},\:1)$ & $\frac{1}{3}$ & $\frac{1}{6}$ & $\frac{1}{10}$ & $\frac{1}{15}$ & $\frac{1}{21}$ & $\frac{1}{28}$\\[0.8ex]
			\hspace{0.5em}$(\frac{1}{2},\:\frac{3}{2})$ & $\frac{3}{8}$ & $\frac{1}{5}$ & $\frac{1}{8}$ & $\frac{3}{35}$ & $\frac{1}{16}$ & $\frac{1}{21}$\\[0.8ex]
			\hspace{0.5em}$(\frac{1}{2},\:2)$ & $\frac{2}{5}$ & $\frac{2}{9}$ & $\frac{1}{7}$ & $\frac{1}{10}$ & $\frac{2}{27}$ & $\frac{2}{35}$\\[0.8ex]
			\hspace{0.5em}$(\frac{1}{2},\:\frac{5}{2})$ & $\frac{5}{12}$ & $\frac{5}{21}$ & $\frac{5}{32}$ & $\frac{1}{9}$ & $\frac{1}{12}$ & $\frac{5}{77}$\\[0.8ex]
			
			\hspace{0.5em}$(\frac{1}{2},\:3)$ & $\frac{3}{7}$ & $\frac{1}{4}$ & $\frac{1}{6}$ & $\frac{3}{25}$ & $\frac{1}{11}$ & $\frac{1}{14}$\\[0.8ex]

			\hspace{0.5em}$(1,\:1)$ & $\frac{4}{9}$ & $\frac{1}{4}$ & $\frac{4}{25}$ & $\frac{1}{9}$ & $\frac{4}{49}$ & $\frac{1}{16}$\\[0.8ex]
			\hspace{0.5em}$(1,\:\frac{3}{2})$ & $\frac{1}{2}$ & $\frac{3}{10}$ & $\frac{1}{5}$ & $\frac{1}{7}$ & $\frac{3}{28}$ & $\frac{1}{12}$\\[0.8ex]
			\hspace{0.5em}$(1,\:2)$ & $\frac{8}{15}$ & $\frac{1}{3}$ & $\frac{8}{35}$ & $\frac{1}{6}$ & $\frac{8}{63}$ & $\frac{1}{10}$\\[0.8ex]
			\hspace{0.5em}$(1,\:\frac{5}{2})$ & $\frac{5}{9}$ & $\frac{5}{14}$ & $\frac{1}{4}$ & $\frac{5}{27}$ & $\frac{1}{7}$ & $\frac{5}{44}$\\[0.8ex]
			
			\hspace{0.5em}$(1,\:3)$ & $\frac{4}{7}$ & $\frac{3}{8}$ & $\frac{4}{15}$ & $\frac{1}{5}$ & $\frac{12}{77}$ & $\frac{1}{8}$\\[0.8ex]

			\hspace{0.5em}$(\frac{3}{2},\:\frac{3}{2})$ & $\frac{9}{16}$ & $\frac{9}{25}$ & $\frac{1}{4}$ & $\frac{9}{49}$ & $\frac{9}{64}$ & $\frac{1}{9}$\\[0.8ex]
			\hspace{0.5em}$(\frac{3}{2},\:2)$ & $\frac{3}{5}$ & $\frac{2}{5}$ & $\frac{2}{7}$ & $\frac{3}{14}$ & $\frac{1}{6}$ & $\frac{2}{15}$\\[0.8ex]
			\hspace{0.5em}$(\frac{3}{2},\:\frac{5}{2})$ & $\frac{5}{8}$ & $\frac{3}{7}$ & $\frac{5}{16}$ & $\frac{5}{21}$ & $\frac{3}{16}$ & $\frac{5}{33}$\\[0.8ex]
			
			\hspace{0.5em}$(\frac{3}{2},\:3)$ & $\frac{9}{14}$ & $\frac{9}{20}$ & $\frac{1}{3}$ & $\frac{9}{35}$ & $\frac{9}{44}$ & $\frac{1}{6}$\\[0.8ex]

			\hspace{0.5em}$(2,\:2)$ & $\frac{16}{25}$ & $\frac{4}{9}$ & $\frac{16}{49}$ & $\frac{1}{4}$ & $\frac{16}{81}$ & $\frac{4}{25}$\\[0.8ex]
			\hspace{0.5em}$(2,\:\frac{5}{2})$ & $\frac{2}{3}$ & $\frac{10}{21}$ & $\frac{5}{14}$ & $\frac{5}{18}$ & $\frac{2}{9}$ & $\frac{2}{11}$\\[0.8ex]
			
			\hspace{0.5em}$(2,\:3)$ & $\frac{24}{35}$ & $\frac{1}{2}$ & $\frac{8}{21}$ & $\frac{3}{10}$ & $\frac{8}{33}$ & $\frac{1}{5}$\\[0.8ex]

			\hspace{0.5em}$(\frac{5}{2},\:\frac{5}{2})$ & $\frac{25}{36}$ & $\frac{25}{49}$ & $\frac{25}{64}$ & $\frac{25}{81}$ & $\frac{1}{4}$ & $\frac{25}{121}$\\[0.8ex]
			
			\hspace{0.5em}$(\frac{5}{2},\:3)$ & $\frac{5}{7}$ & $\frac{15}{28}$ & $\frac{5}{12}$ & $\frac{1}{3}$ & $\frac{3}{11}$ & $\frac{5}{22}$\\[0.8ex]
			\hspace{0.5em}$(3,\:3)$ & $\frac{36}{49}$ & $\frac{9}{16}$ & $\frac{4}{9}$ & $\frac{9}{25}$ & $\frac{36}{121}$ & $\frac{1}{4}$\\[0.5ex]
			
		\end{tabular}
	\end{ruledtabular}
\end{table}

Define
\begin{align}\label{eq:sG}
s(G):=\max_{e,e'\in E}s(P_e P_{e'})=\max_{e,e'\in E|e\neq e'}s(P_e P_{e'}).
\end{align}
By definition $0\leq s(G)< 1$. 
Here $s(G)$ can be abbreviated as $s$ if there is no danger of confusion. According to \tref{tab:s2}, we have $s(G)=1/2$ for most lattices of practical interest, including the open chain (with at least four nodes), closed chain, square lattice, honeycomb lattice,  triangular lattice, kagome lattice, and square-octagon lattice (cf. \fref{fig:LatticeColoring}). For the open chain with three nodes, we have $s(G)=1/3$.

In the following theorem,  $\gamma=\gamma(H_G)$ is the spectral gap of $H_G$, while $S_E$ and $\nu_{S_E}(\mu)$  are defined in \esref{eq:S_E}{eq:nuSmu}, respectively.

\begin{thm}\label{thm:AKLTnu1}
Let $|\Psi_G\>$ be the AKLT state defined on the graph $G=(V,E)$. 	Suppose $\Omega(\mu, \scrM)$ is the verification operator specified by the probability distribution $\mu$ and the matching cover $\scrM$ composed of $m$ matchings. Then 
\begin{align}
\nu(\Omega(\mu, \scrM)) \geq \frac{\nu_{S_E}(\mu)}{m}f\Bigl(\frac{\gamma}{s^2g^2}\Bigr)\geq \frac{\nu_{S_E}(\mu)\gamma}{24m(S_E-1)^2}, \label{eq:nuLB1}
\end{align}
where  $s=s(G)$, $g=2S_E-2$,  and 
\begin{align}\label{eq:fx}
f(x)=\begin{cases}
\frac{\sqrt{1+x}-1}{\sqrt{1+x}} &m=2,\\[1ex]
\frac{\sqrt{1+x}-1}{\sqrt{1+x}+1} &m\geq 3.
\end{cases}
\end{align} 
If $\mu_\sym$ forms a spherical $t$-design with $t=2S_E$, then 
\begin{align}
\nu(\Omega(\mu, \scrM))&\geq \frac{2}{m(2S_E+1)}f\Bigl(\frac{\gamma}{s^2g^2}\Bigr)\nonumber\\
&\geq \frac{\gamma}{12m(2S_E+1)(S_E-1)^2}. \label{eq:nuLB2}
\end{align}
\end{thm}
\Thref{thm:AKLTnu1} follows from Theorem~1 in the companion paper \cite{ZhuLC22} as well as  \lref{lem:nuSmu} and \thref{thm:OmegaSdesign} in \sref{sec:BondVerify}. In conjunction with \eref{eq:nu}, it is now straightforward to  derive the following upper bound on the the minimum  number of tests required to verify the AKLT state $|\Psi_G\>$ within infidelity $\epsilon$ and significance level $\delta$:
\begin{align}
N&\leq \biggl\lceil \frac{m(2S_E+1)\ln(\delta^{-1})}{2\epsilon f\bigl(\frac{\gamma}{s^2g^2}\bigr)}\biggr\rceil\nonumber\\
&\leq \biggl\lceil\frac{12m(2S_E+1)(S_E-1)^2\ln(\delta^{-1})}{\gamma\epsilon}\biggr\rceil. \label{eq:NUB}
\end{align}
When $x\ll 1$, $f(x)$ can be approximated by $x/2$ for $m=2$ and $x/4$ for $m\geq 3$. If $\gamma/(s^2g^2)\ll1$, then \eref{eq:nuLB1} implies that
\begin{align}
\nu(\Omega(\mu, \scrM)) \gtrsim\begin{cases}
\frac{\gamma\nu_{S_E}(\mu)}{2m s^2g^2} &m=2,\\[1ex]
\frac{\gamma\nu_{S_E}(\mu)}{4m s^2g^2} &m\geq 3.
\end{cases} 
\end{align}
This equation is instructive to understanding the efficiency of the matching protocol.
\Esref{eq:nuLB2}{eq:NUB} can be simplified in a similar way.

Recall that the minimum number of matchings required to cover the edge set of $G$ is equal to the chromatic index $\chi'(G)$. If $m=\chi'(G)$, then \eref{eq:nuLB1} reduces to 
\begin{align}
&\nu(\Omega(\mu, \scrM))\geq \frac{\nu_{S_E}(\mu)}{\chi'(G)}f\Bigl(\frac{\gamma}{s^2g^2}\Bigr)\geq \frac{\nu_{S_E}(\mu)\gamma}{24\chi'(G)(S_E-1)^2}\nonumber\\
&\quad \geq \frac{\nu_{S_E}(\mu)\gamma}{24[\Delta(G)+1][\Delta(G)-1]^2}\geq \frac{\nu_{S_E}(\mu)\gamma}{24\Delta(G)^3}, \label{eq:nuLB3}
\end{align}
where the third inequality follows from the facts that $\chi'(G)\leq \Delta(G)+1$ and $S_E\leq \Delta(G)$. 
Although it is not always easy to find an optimal matching cover,  a nearly optimal matching cover composed of $\chi'(G)+1$ matchings can be found efficiently.

If in addition $\mu$ forms a spherical $t$-design with $t=2S_E$, then \eref{eq:nuLB2} implies that
\begin{align}
&\nu(\Omega(\mu,\scrM)) \geq \frac{2}{\chi'(G)(2S_E+1)}f\Bigl(\frac{\gamma}{s^2g^2}\Bigr)\nonumber\\
&\geq \frac{\gamma}{12\chi'(G)(2S_E+1)(S_E-1)^2}\geq\frac{\gamma}{24\Delta(G)^4}, \label{eq:nuLB4}
\end{align}
Accordingly, \eref{eq:NUB} reduces to 
\begin{align}
 N\leq \biggl\lceil \frac{\chi'(G)(2S_E+1)\ln(\delta^{-1})}{2\epsilon f\bigl(\frac{\gamma}{s^2g^2}\bigr)}\biggr\rceil\leq  \biggl\lceil \frac{24\Delta(G)^4\ln(\delta^{-1})}{\gamma\epsilon} \biggr\rceil. \label{eq:NUB2}
\end{align}
For most AKLT states of practical interest, including those defined on various lattices as illustrated in \fref{fig:LatticeColoring} (see also Fig.~1 in the companion paper \cite{ZhuLC22}), $\Delta(G)$ does not increase with the system size. So these AKLT states can be verified with constant sample cost that is independent of  the system size as long as the spectral gap $\gamma$ has a nontrivial system-independent lower bound. Our verification protocols are much more efficient than protocols known in the literature \cite{CramPFS10,HangKSE17,TakeM18,CruzBTS22} and the sample costs have much better scaling behaviors with respect to the system size, spectral gap of the underlying Hamiltonian, and the precision as quantified by the infidelity.

The next theorem follows from Theorems 2, 3 in the companion paper \cite{ZhuLC22}  and \lref{lem:nuSmu} in \sref{sec:BondVerify}. 
\begin{thm}\label{thm:AKLTnu2}
	Suppose $\scrM$ in \thref{thm:AKLTnu1} is an edge coloring of $G$ and let $p=(|M_1|, |M_2|, \ldots, |M_m|)/|E|$; then 
	\begin{align}
	\nu(\Omega(\mu, \scrM,p))\geq \frac{\nu_{S_E}(\mu)\gamma}{|E|} \geq \frac{2\nu_{S_E}(\mu)\gamma}{n(n-1)} .  \label{eq:SpectralGapColor}
	\end{align}
If $\mu_\sym$ forms a spherical $t$-design with $t=2S_E$, then 
	\begin{align}
\nu(\Omega(\mu, \scrM,p))\geq \frac{2\gamma}{(2S_E+1)|E|} \geq \frac{4\gamma}{(2S_E+1)n(n-1)}.  \label{eq:SpectralGapColor2}
\end{align}	
The first inequality in \eref{eq:SpectralGapColor2} is saturated if $S_e$ is independent of $e\in E$ and $\scrM$ is the trivial edge coloring with $|\scrM|=|E|$. 
\end{thm}
By virtue of  \esref{eq:nu}{eq:SpectralGapColor2}, we can  derive another upper bound on the minimum  number of tests required to verify the AKLT state $|\Psi_G\>$ within infidelity $\epsilon$ and significance level $\delta$:
\begin{align}
N\leq \biggl\lceil \frac{(2S_E+1)|E|\ln(\delta^{-1})}{2\gamma\epsilon}\biggr\rceil. 
\end{align} 
When $G$ is a connected $k$-regular graph with $n\geq 2$ vertices, we have $|E|=nk/2$ and $S_e=S_E=\Delta(G)=k$ for all $e\in E$, so
\eref{eq:SpectralGapColor} reduces to
\begin{align}
\nu(\Omega(\mu, \scrM,p))\geq \frac{2\nu_k(\mu)\gamma}{nk}. 
\end{align}
If in addition $\mu_\sym$ forms a spherical $t$-design with $t=2S_E=2k$, then we have $\nu_k(\mu)=2/(2k+1)$, so  the above equation (cf. \eref{eq:SpectralGapColor2}) reduces to 
\begin{align}
\nu(\Omega(\mu, \scrM,p))\geq \frac{4\gamma}{nk(2k+1)}.
\end{align}
This inequality is saturated if $\scrM$ corresponds to the trivial edge coloring and $p$ is uniform. In conjunction with \eref{eq:nu}, it is straightforward to derive the number of tests required to achieve a given precision.

\section{\label{sec:1DAKLT}Verification of 1D AKLT states}
 
In this section we discuss in more detail the verification of  1D AKLT states, that is, AKLT states  defined on the  closed chain (cycle) and open chain.

\begin{figure}[b]
	\includegraphics[width=0.45\textwidth]{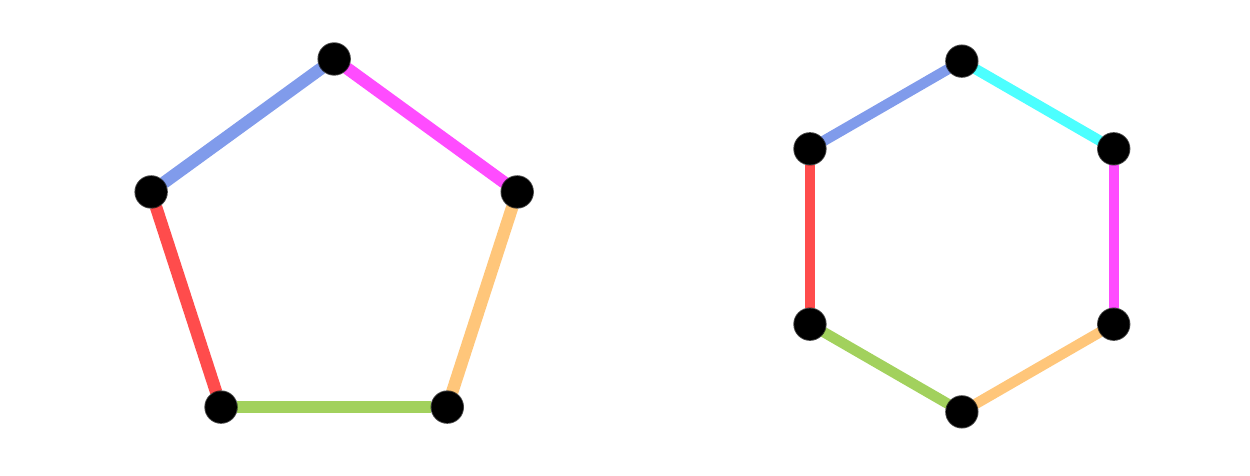}
	\caption{Trivial edge colorings of closed chains with five and six vertices.}
	\label{fig:Closed56color}
\end{figure}

\subsection{Verification of the AKLT state on the closed chain}
\subsubsection{Simplest verification protocols}
Let $G(V,E)$ be the closed chain with $n$ vertices.
Given a bond verification protocol specified by a probability distribution $\mu$ on the unit sphere, then a verification protocol of the AKLT state $|\Psi_G\>$ is specified by a weighted matching cover of $G$. The simplest matching cover, denoted by  $\scrM_\rmT$, corresponds to the trivial edge coloring and consists of $n$ matchings, each of which consists of only one edge  as illustrated in  \fref{fig:Closed56color}. Accordingly, each test operator is associated with one edge. Denote  by $T_j$ the test operator associated with the edge $\{j,j+1\}$ (here $n+1$ is identified with 1 under the periodic boundary condition). Note that  test operators $T_j$ for $j=1,2,\ldots, n$ are related to each other by cyclic permutations, so they should be performed with the same probability to maximize the spectral gap. The resulting verification operator reads
\begin{align}\label{eq:VOAKLT1}
\Omega(\mu,\scrM_\rmT)=\frac{1}{n}\sum_{j=1}^{n} T_j(\mu).
\end{align}
\Fref{fig:cycleSimple} shows the spectral gap of $\Omega(\mu,\scrM_\rmT)$  with $\mu$ constructed from the five platonic solids; in addition, the figure  shows the number of tests required to verify the AKLT state within infidelity $\epsilon=0.01$ and significance level $\delta=0.01$. 

\begin{figure}[t]
	\includegraphics[width=0.42\textwidth]{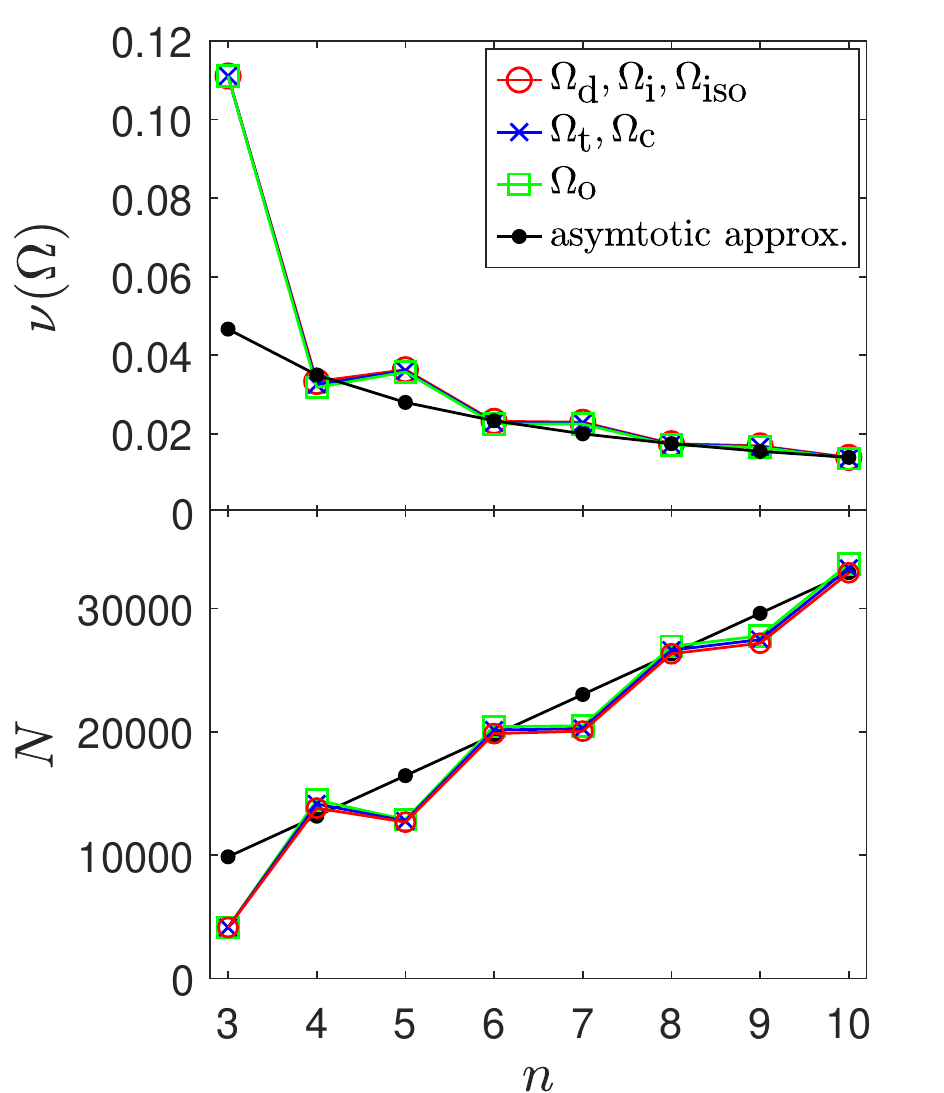}		
	\caption{\label{fig:cycleSimple}
Spectral gap $\nu(\Omega(\mu,\scrM_\rmT))$ and number $N$ of tests required to verify the AKLT state on the closed chain of $n$ nodes
with precision $\delta=\epsilon=0.01$.		
Each protocol is based on the trivial edge coloring, and the underlying	 bond verification protocol is constructed from a platonic solid or the isotropic distribution on the unit sphere as discussed in \sref{sec:BondVC} and indicated in the legend. The black dots represent the asymptotic approximation presented in \eref{eq:AsympApprox}.		
	}
\end{figure}

According to \thref{thm:AKLTnu2}, the spectral gap of $\Omega(\mu,\scrM_\rmT)$ satisfies
\begin{align}
\nu(\Omega(\mu,\scrM_\rmT))&\geq \frac{\nu_2(\mu)\gamma(H^{\circ}(n))}{n}; \label{eq:nuLB}
\end{align}	
here $\nu_2(\mu)$ denotes the spectral gap (from the maximum eigenvalue 1) of the bond verification operator, while $\gamma(H^{\circ}(n))$ denotes the spectral gap of the Hamiltonian (from the minimum eigenvalue corresponding to the ground state).
By virtue of  \thref{thm:gapGM} and \eref{eq:nuLB}
we can further deduce that
\begin{equation}\label{eq:nuLBck}
\nu\left(\Omega(\mu,\scrM_\rmT)\right)\geq \frac{c_k \nu_2(\mu)}{n},\quad n>2k,
\end{equation}
where $c_k$ is defined in \eref{eq:ck}. 
So the number of tests required to verify the AKLT state $|\Psi_G\>$ within infidelity $\epsilon$ and significance level $\delta$ satisfies 
\begin{align}
N&\leq \left\lceil\frac{ \ln (\delta^{-1})}{\nu(\Omega(\mu,\scrM_\rmT))\epsilon}\right \rceil\leq 
\left\lceil\frac{n \ln (\delta^{-1})}{\nu_2(\mu)\gamma(H^{\circ}(n))\epsilon}\right \rceil
\nonumber\\
&\leq\left\lceil\frac{ n\ln (\delta^{-1})}{c_k\,\nu_2(\mu)\,\epsilon}\right \rceil, \quad n>2k,\label{eq:NUBck}
\end{align}
which is linear in the number of spins. 

The inequality  in \eref{eq:nuLB} is saturated  when  $\mu$ forms a spherical 4-design (icosahedron and dodecahedron protocols for example), in which case  we have $\nu_2(\mu)=2/5$ and 
\begin{align}
\nu(\Omega(\mu,\scrM_\rmT))&=\frac{2\gamma(H^{\circ}(n))}{5n}, \label{eq:nutri1D4design}\\
N &\approx \left\lceil\frac{5n \ln (\delta^{-1})}{2\gamma(H^{\circ}(n))\epsilon}\right \rceil
\leq\left\lceil\frac{ 5n\ln (\delta^{-1})}{2c_k\,\epsilon}\right \rceil,  
\end{align}
where the inequality holds whenever $n>2k$.
Although \eref{eq:nutri1D4design} is derived under the 4-design assumption, calculation shows that it holds with high precision (with deviation less than 5\% for $n\leq 10$) for all protocols based on   platonic solids, as illustrated in \fref{fig:cycleSimple}.  When $n=3$ for example, we have 
$\gamma(H^{\circ}(n))=5/6$ and $\nu(\Omega(\mu,\scrM_\rmT))=1/9$ for all protocols based on  platonic solids. This observation indicates that the general lower bound in \eref{eq:nuLB} is usually not tight when $\mu$ does not form a 4-design. In other words, protocols based on tetrahedron, octahedron, and cube are more efficient than expected; the reason is still not very clear now.

In the large-$n$ limit, the spectral gap $\gamma(H^{\circ}(n))$ is approximately equal to $0.350$ \cite{GarcMW13,WeiRA22}. If in addition $\mu$ forms a spherical 4-design, then  we have
\begin{align}\label{eq:AsympApprox}
\nu(\Omega(\mu,\scrM_\rmT))&\approx \frac{0.140}{n},\quad 
N\approx\frac{7.14\, n\ln (\delta^{-1})}{\epsilon}. 
\end{align}
Numerical calculation shows that all protocols based on platonic solids can achieve a similar performance  as illustrated in \fref{fig:cycleSimple}.

\subsubsection{Optimal matching protocols}
\begin{figure}[b]
	\includegraphics[width=0.42\textwidth]{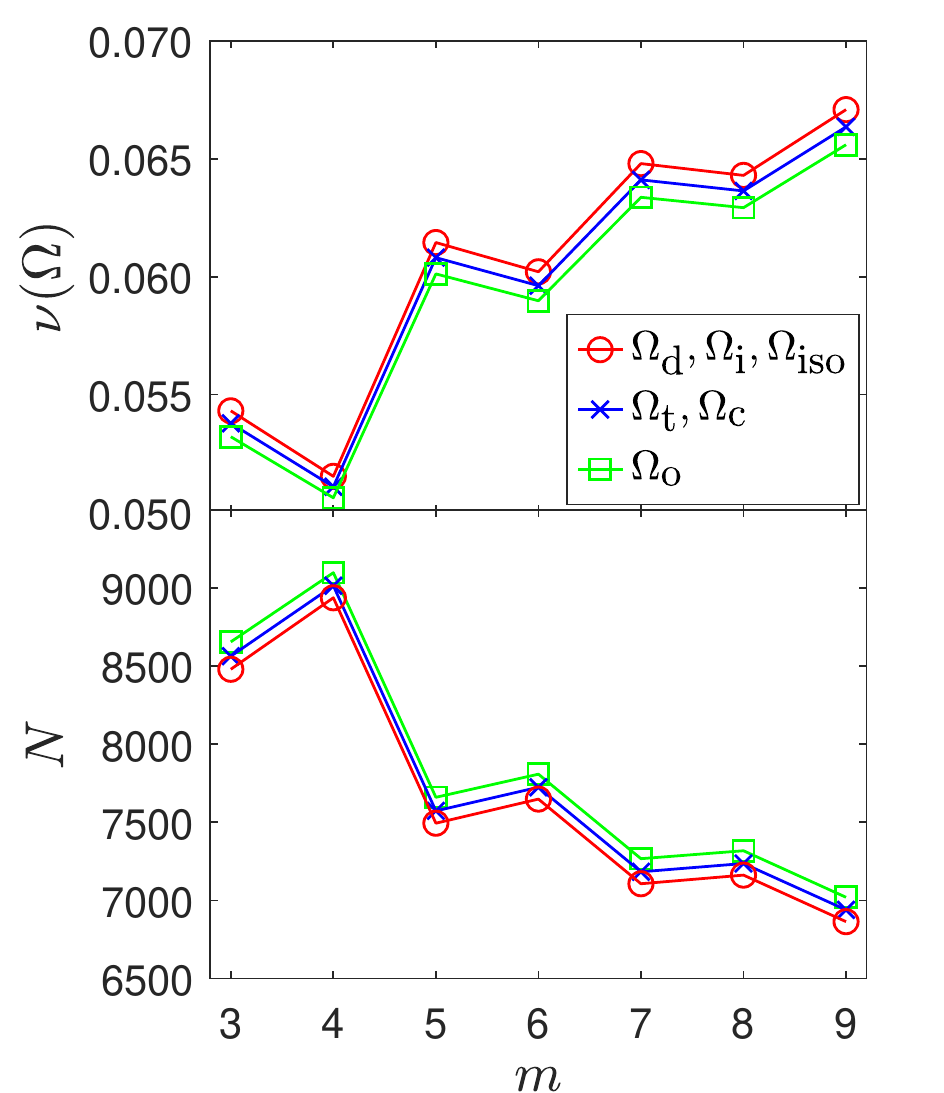}
	\caption{\label{fig:ntests}Verification of the AKLT state on the closed chain of nine nodes based on the  matching covers  $\scrM_m=\{M_j\}_{j=1}^m$. Here $M_j$ is the maximum matching defined in  \eref{eq:Mj}, and $m$ is the number of maximum matchings  and also  the number of distinct tests employed.  Infidelity and significance level are chosen to be $\delta=\epsilon=0.01$  as in \fref{fig:cycleSimple}; the choice of bond verification protocols is also the same. }
\end{figure}

\begin{table*}
	\caption{\label{tab:matching}The matching numbers $\upsilon(G)$ and the numbers of maximal and maximum matchings (shown as triples) for closed chains and open chains of 3 to 10 vertices.}
	\renewcommand\arraystretch{1.2}
	\begin{ruledtabular}
		\begin{tabular}{c|cccccccc}
			$n$ & 3 & 4 & 5 & 6 & 7 & 8 & 9 & 10\\
			\hline
			closed chain & (1, 3, 3) & (2, 2, 2) & (2, 5, 5) & (3, 5, 2) & (3, 7, 7) & (4, 10, 2) & (4, 12, 9) & (5, 17, 2)\\	
			open chain & (1, 2, 2) & (2, 2, 1) & (2, 3, 3) & (3, 4, 1) & (3, 5, 4) & (4, 7, 1) & (4, 9, 5) & (5, 12, 1) \\
		\end{tabular}
	\end{ruledtabular}
\end{table*}

More efficient verification protocols can be constructed from better matching covers. For the cycle graph $G$ with $n$ vertices, the chromatic index is given by 
\begin{align}
\chi'(G)=\begin{cases}
2 & \mbox{if $n$ is even},\\
3 & \mbox{if $n$ is odd}. 
\end{cases}
\end{align}
When $n$ is even,  there exist two maximum matchings, namely,
\begin{equation}
\begin{aligned}
M_1&=\{\{1,2\}, \{3,4\},\ldots, \{n-1,n\}\}, \\ M_2&=\{\{2,3\},\{4,5\},\ldots, \{n,1\}\},
\end{aligned}
\end{equation}
which form the matching cover $\scrM=\{M_1, M_2\}$ and also defines an edge coloring of $G$. 
Given a probability distribution $\mu$ on the unit sphere, then we can construct two test operators $T_{M_1}(\mu), T_{M_2}(\mu)$ according to \eref{eq:TestMatch}. By symmetry the two tests should be performed with the same probability to maximize the spectral gap. According to \thref{thm:AKLTnu1} with $m=g=2$ and $s=1/2$, the spectral gap of the resulting verification operator satisfies
\begin{align}
\nu(\Omega(\mu, \scrM))&\geq  \frac{(\sqrt{1+\gamma}-1)\nu_2(\mu)}{2\sqrt{1+\gamma}}\nonumber\\
&\geq  \frac{(\sqrt{1+c_k}-1)\nu_2(\mu)}{2\sqrt{1+c_k}}, \label{eq:nuLBeven}
\end{align}
where $\gamma=\gamma(H^{\circ}(n))$. Here the second inequality follows from \thref{thm:gapGM} and is applicable for $n>2k$ and $k>2$. 
In the large-$n$  limit, we have
$\gamma(H^{\circ}(n))\approx 0.350$ \cite{GarcMW13,WeiRA22}, so the above equation implies that  
\begin{align}\label{eq:nuLBeven2}
\nu(\Omega(\mu, \scrM))\gtrsim 0.0697\nu_2(\mu).
\end{align}
If in addition $\mu$ forms a $4$-design, then $\nu_2(\mu)=2/5$ and 
\begin{align}\label{eq:nuLBeven3}
\nu(\Omega(\mu, \scrM))\gtrsim 0.0279.
\end{align}
Accordingly, the number of tests required to verify the AKLT state within infidelity $\epsilon$ and significance level $\delta$ satisfies 
\begin{align}\label{eq:NUBeven}
N\lesssim \frac{36\ln (\delta^{-1})}{\epsilon}. 
\end{align}
Numerical calculation presented in \fref{fig:MoreEffPro} suggests that the bounds in \esref{eq:nuLBeven3}{eq:NUBeven} are tight within a factor  of 3.

To construct optimal matching protocols, in principle we need to consider all maximal matchings; see \tref{tab:matching} for the number of maximal matchings when $n=3,4,\ldots, 10$.    Nevertheless, numerical calculation based on \eref{eq:normSDP} suggests that the above protocol is still optimal even in that case;  in other words, other maximal matchings do not help.

\begin{figure}
	\includegraphics[width=0.42\textwidth]{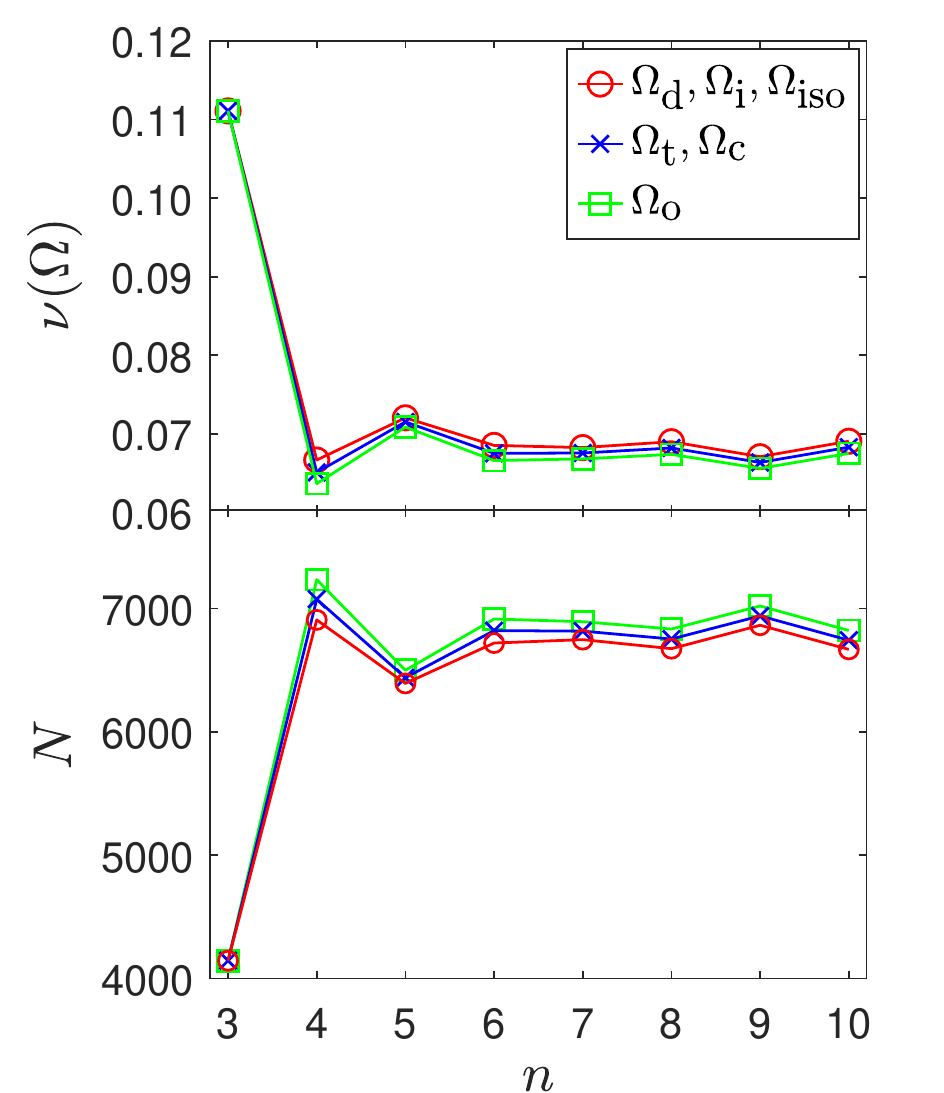}
	\caption{\label{fig:MoreEffPro}Verification of the AKLT state on the closed chain  based on  the matching cover composed of all maximum matchings.
 Infidelity and significance level are chosen to be $\delta=\epsilon=0.01$  as in \fref{fig:cycleSimple}; the choice of bond verification protocols is also the same.}
\end{figure}

When $n$ is odd, each matching of $G$ can contain at most $(n-1)/2$ edges, so at least three matchings are required to cover the edge set. Now, there exist $n$ maximum matchings, namely
\begin{align}
M_j=&\bigl\{\{j,j+1\}, \{j+2,j+3\},\ldots,\nonumber\\ &\{j+n-3,j+n-2\}\bigr\}, \quad j=1,2,\ldots, n,\label{eq:Mj}
\end{align}
where  $j$ and $j+n$ denote the same vertex. All  these maximum matchings  can be generated from $M_1$ by cyclic permutations. Let $\scrM_m=\{M_j\}_{j=1}^m$ for $m=3,4,\ldots, n$. Then $\scrM_m$ are matching covers of $G$ and can be employed to construct verification protocols for $|\Psi_G\>$. The resulting verification operators are denoted by $\Omega(\mu, \scrM_m)$. By virtue of   \thref{thm:AKLTnu1} with $g=2$ and $s=1/2$ we can deduce that
\begin{align}
&\nu(\Omega(\mu, \scrM_m))\geq \frac{(\sqrt{1+\gamma}-1)\nu_2(\mu)}{m(\sqrt{1+\gamma}+1)}\nonumber\\
&\geq \frac{(\sqrt{1+c_k}-1)\nu_2(\mu)}{m(\sqrt{1+c_k}+1)}
\quad m=3,4,\ldots, n,
\end{align}
where $\gamma=\gamma(H^{\circ}(n))$. Here the second inequality follows from \thref{thm:gapGM} and is applicable for $n>2k$ and $k>2$. 
In addition, numerical calculation shows that
\begin{align}
\nu(\Omega(\mu, \scrM_{m+2}))\geq \nu(\Omega(\mu, \scrM_m)), \quad m=3,4,\ldots, n-2. 
\end{align}
The performances of these verification protocols are illustrated in \fref{fig:ntests}.

By symmetry consideration we can deduce that 
\begin{align}
\nu(\Omega(\mu, \scrM_n))\geq \nu(\Omega(\mu, \scrM_m)), \quad m=3,4,\ldots, n,
\end{align}
which implies that
\begin{align}
\nu(\Omega(\mu, \scrM_n))&\geq \nu(\Omega(\mu, \scrM_3))\geq \frac{(\sqrt{1+\gamma}-1)\nu_2(\mu)}{3(\sqrt{1+\gamma}+1)}\nonumber\\
&\geq \frac{(\sqrt{1+c_k}-1)\nu_2(\mu)}{3(\sqrt{1+c_k}+1)},  \label{eq:nuLBodd}
\end{align}
where  the second inequality is applicable when $n>2k$ and $k>2$.
This bound is slightly worse than the counterpart in \eref{eq:nuLBeven}, but we believe that  \eref{eq:nuLBeven} applies for both even $n$ and odd $n$, and so do \ecref{eq:nuLBeven2}{eq:NUBeven}. Moreover, for a given bond verification protocol $\mu$, the verification operator $\Omega(\mu, \scrM_n)$ has the largest spectral gap among all verification operators 
based on maximum matchings.  Numerical calculation based on \eref{eq:normSDP} further  suggests that  $\Omega(\mu, \scrM_n)$ is still optimal even if we consider all maximal matchings. The performances of optimal matching protocols are  illustrated in \fref{fig:MoreEffPro}.

\subsection{Verification of the AKLT state on the open chain}
Next, we turn to the AKLT state on the open chain, which can also be verified following the general approach presented in \sref{sec:AKLT}.

\subsubsection{Simplest verification protocols}

\begin{figure}[b]
	\includegraphics[width=0.42\textwidth]{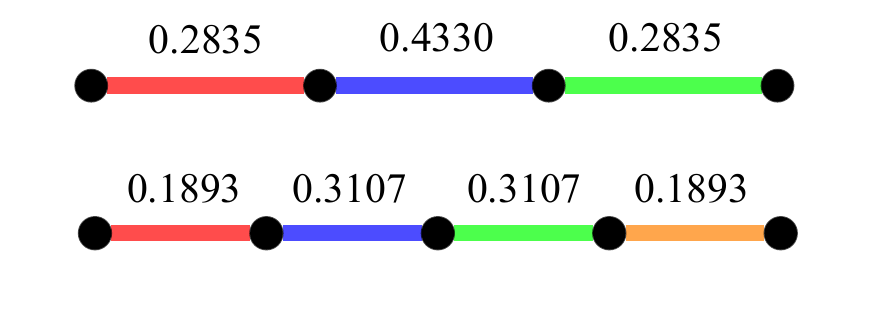}
	\caption{\label{fig:Open45prob}Trivial edge colorings of open chains with four and five vertices together with the optimal probabilities for performing the tests associated with individual colors. Here the bond verification protocol is built from the dodecahedron.}
\end{figure}

\begin{figure}[t]
	\includegraphics[width=0.42\textwidth]{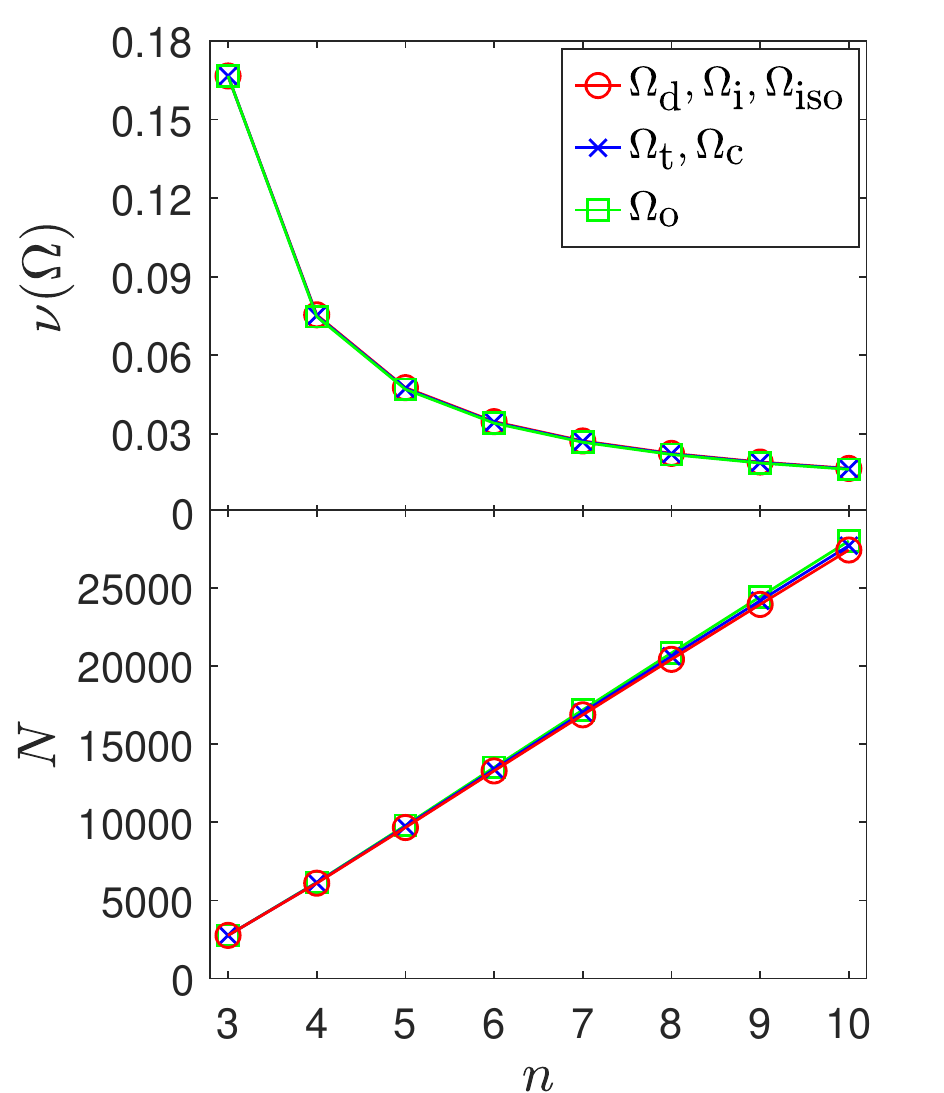}
	\caption{\label{eq:OpenSimple}
Verification of  the AKLT state on the open chain based on the trivial edge coloring with uniform probabilities. Infidelity and significance level are chosen to be $\delta=\epsilon=0.01$  as in \fref{fig:cycleSimple}; the choice of bond verification protocols is also the same.	
	}
\end{figure}

Incidentally, when the open chain has two nodes (corresponding to the only connected graph of two vertices), the AKLT Hamiltonian  coincides with the projector onto the symmetric subspace of two qubits and has  spectral gap equal to 1.
The corresponding AKLT state coincides with the singlet. In this case, each matching protocol is just a bond verification protocol. The largest spectral gap is $2/3$, which is achieved when the underlying distribution on the unit sphere forms a spherical 2-design (which is the case for all protocols based on platonic solids). Such protocols are also optimal among all protocols based on  separable measurements  \cite{PallLM18,ZhuH19O}.

Many results on the closed chain are still applicable with minor modification for the open chain. First, let us consider  verification protocols based on the  trivial edge coloring as illustrated in \fref{fig:Open45prob}. In this case, \esref{eq:VOAKLT1}{eq:nuLB} are modified as follows,
\begin{align}
\Omega(\mu,\scrM_\rmT)&=\frac{1}{n-1}\sum_{j=1}^{n-1} T_j(\mu),\quad n \geq 2,\\
\nu(\Omega(\mu,\scrM_\rmT))&\geq  \frac{\nu_2(\mu)\gamma\bigl(H_{\frac{1}{2},\frac{1}{2}}(n)\bigr)}{n-1},\quad  n \geq 3, \label{eq:nuLBopen}
\end{align}
where $\mu$ determines the bond verification protocol. When $n=3$, $\nu_2(\mu)$ in  \eref{eq:nuLBopen} can also be replaced by  $\nu_{3/2}(\mu)$, which leads to a better lower bound. The performances of several verification protocols based on platonic solids are illustrated in \fref{eq:OpenSimple}.
When  $\mu$ forms a spherical 4-design (which holds for the icosahedron and dodecahedron protocols),  we have $\nu_2(\mu)=2/5$, so  \eref{eq:nuLBopen} reduces to 
\begin{align}
\nu(\Omega(\mu,\scrM_\rmT))&\geq \frac{2\gamma(H_{\frac{1}{2},\frac{1}{2}}(n))}{5(n-1)}, \quad    n \geq 3, 
\end{align}
which is the counterpart of \eref{eq:nutri1D4design}, but the equality cannot be guaranteed in general.

\begin{figure}[t]
	\includegraphics[width=0.42\textwidth]{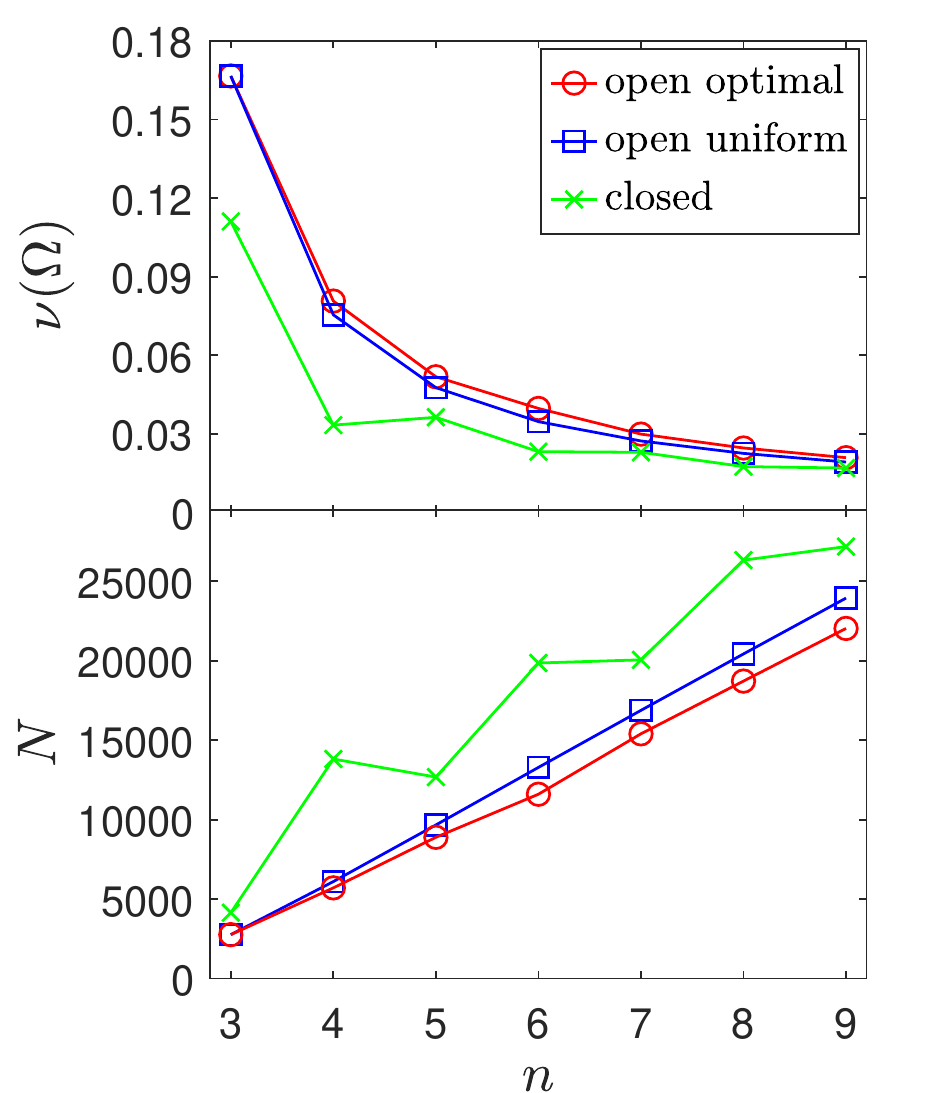}
	\caption{\label{fig:COTC}Verification of  the AKLT states on the open chain and closed chain based on the trivial edge coloring with uniform probabilities and optimal probabilities. For the closed chain, the optimal probabilities are uniform.
		Infidelity and significance level are chosen to be $\delta=\epsilon=0.01$  as in \fref{fig:cycleSimple}. The underlying bond verification protocol is built from the dodecahedron.}
\end{figure}

When $n\geq 4$, in contrast with the closed chain, the optimal probabilities associated with individual colors are not uniform as illustrated in \fref{fig:Open45prob}, and the spectral gap can be increased by optimizing the probabilities according to \eref{eq:normSDP}, as illustrated in \fref{fig:COTC}. This figure also shows that it is slightly easier to verify the AKLT state on the open chain than the counterpart on the closed chain.

\subsubsection{Optimal coloring protocols}

\begin{figure}[b]
	\includegraphics[width=0.42\textwidth]{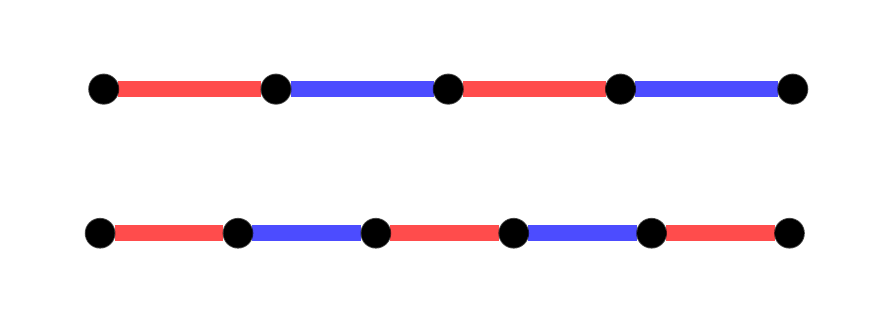}
	\caption{\label{fig:Open56color}Optimal edge colorings of open chains with five and six vertices.}
\end{figure}

\begin{figure}[t]
\includegraphics[width=0.42\textwidth]{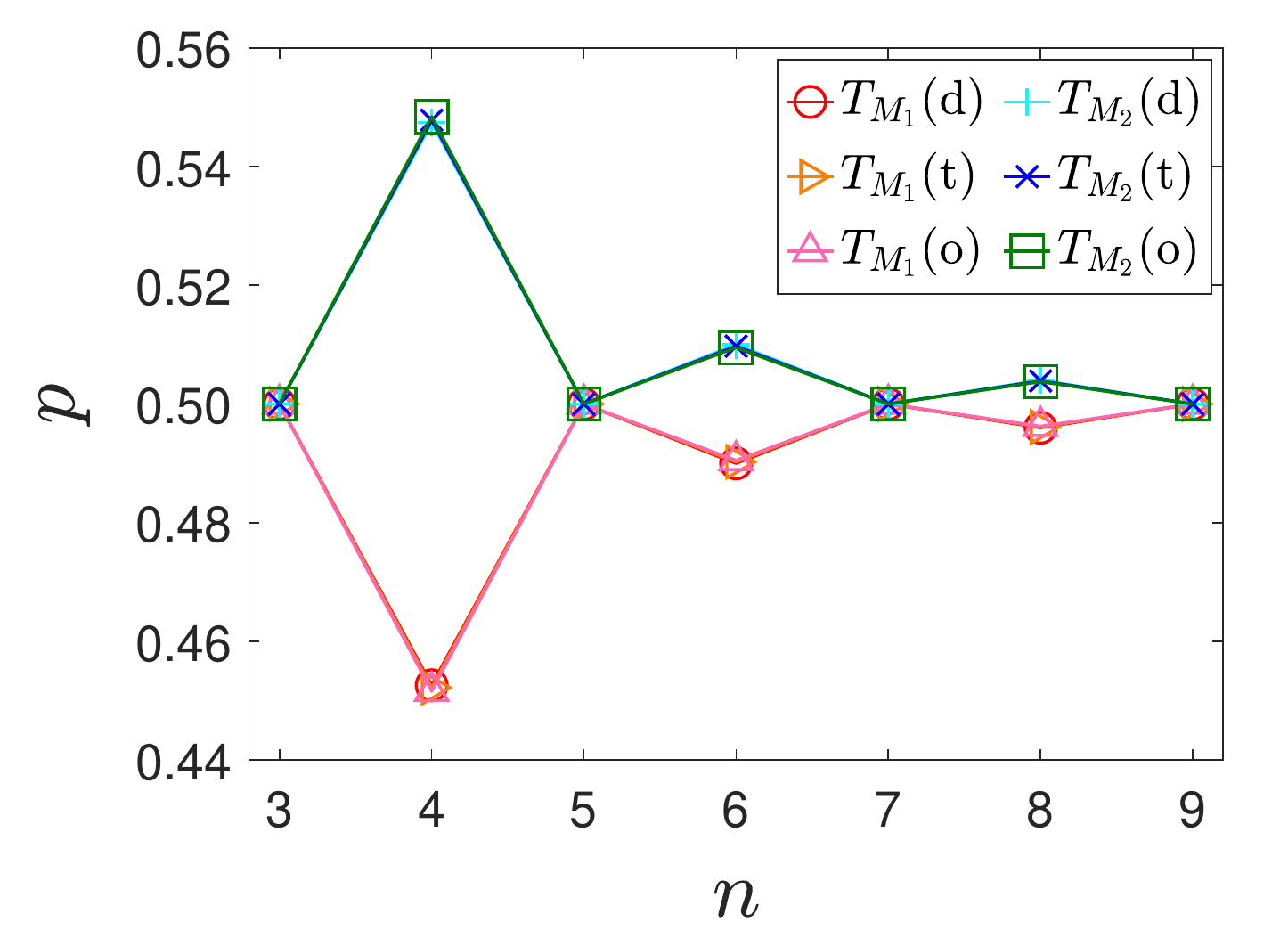}
\caption{\label{fig:matchingprob}Optimal probabilities for performing the two tests  $T_{M_1}(\mu), T_{M_2}(\mu)$ when the distribution $\mu$ is built from the 
tetrahedron (t), octahedron  (o), and  dodecahedron  (d). Here the two matchings $M_1, M_2$ are defined in \esref{eq:ColorMatchOpen}{eq:ColorMatchOpen2}.}
\end{figure}

\begin{figure}
	\includegraphics[width=0.42\textwidth]{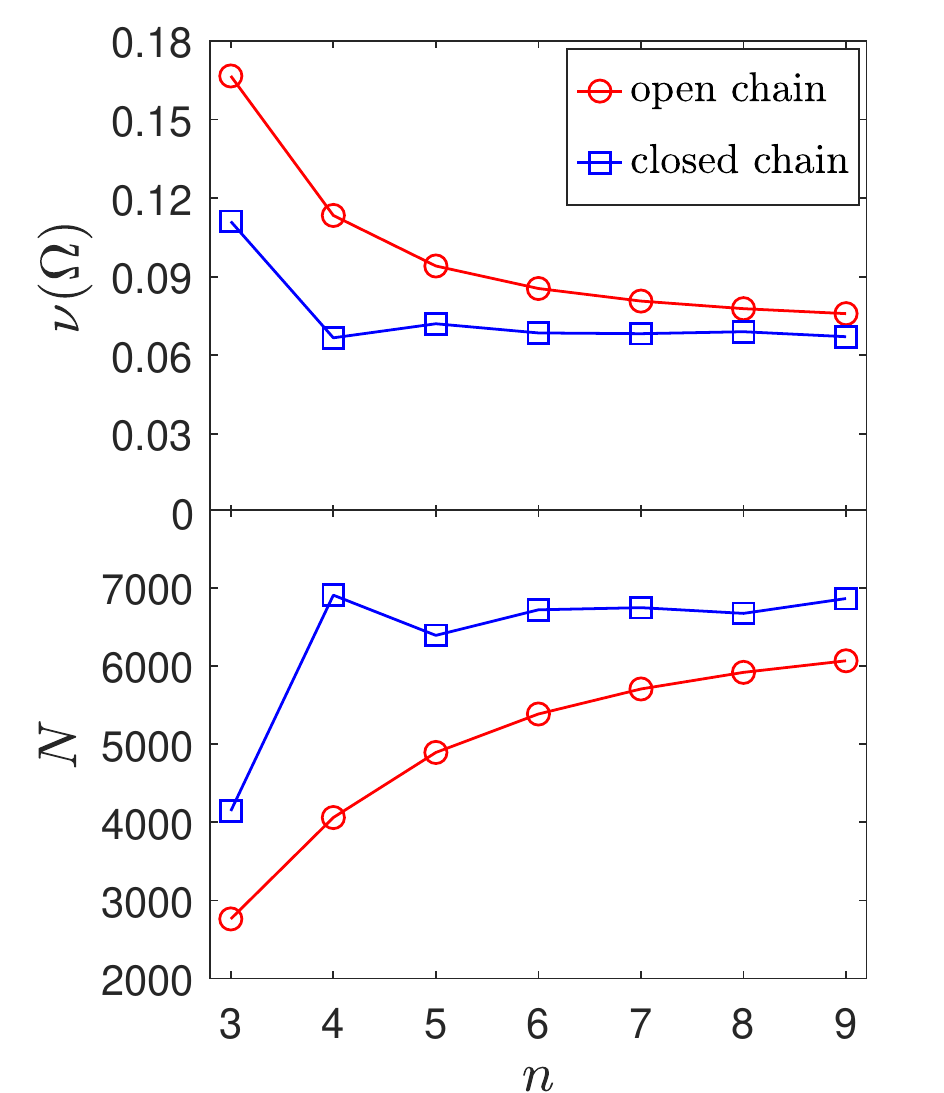}
	\caption{\label{fig:ClosedOpenOpt}Verification of  the AKLT state on the closed chain and open chain with optimal matching protocols. 	Infidelity and significance level are chosen to be $\delta=\epsilon=0.01$  as in \fref{fig:cycleSimple}. The underlying bond verification protocol is based on the dodecahedron. For the even closed chain and odd open chain, the optimal matching protocol is a coloring protocol with uniform probabilities. For the even open chain, the  protocol is a coloring protocol with optimized probabilities. For the odd closed chain, the  protocol is based on all maximum matchings with uniform probabilities. 		
	}
\end{figure}

To improve the efficiency, we can consider verification protocols based on the optimal coloring. Note that the edges of every open chain can be colored using two colors as illustrated in \fref{fig:Open56color}. When $n$ is odd, the two matchings associated with the optimal coloring read
\begin{equation}\label{eq:ColorMatchOpen}
\begin{aligned}
M_1&=\{\{1,2\}, \{3,4\},\ldots, \{n-2,n-1\}\}, \\ M_2&=\{\{2,3\},\{4,5\},\ldots, \{n-1,n\}\},
\end{aligned}
\end{equation}
which contain the same number of edges.   By symmetry the two tests $T_{M_1}(\mu), T_{M_2}(\mu)$ associated with the two matchings should be performed with the same probability to maximize the spectral gap. When $n$ is even, the two matchings read
\begin{equation}\label{eq:ColorMatchOpen2}
\begin{aligned}
M_1&=\{\{1,2\}, \{3,4\},\ldots, \{n-1,n\}\}, \\ M_2&=\{\{2,3\},\{4,5\},\ldots, \{n-2,n-1\}\}.
\end{aligned}
\end{equation}
In this case the optimal probabilities for performing the two tests  $T_{M_1}(\mu), T_{M_2}(\mu)$ are different: they are equal to 0.4900 and 0.5100 when $n=6$ for example. However, as $n$ increases, the difference gets smaller and smaller as illustrated in  \fref{fig:matchingprob}, and the improvement brought by probability optimization becomes negligible when $n\geq 8$.  Let $\scrM=\{M_1, M_2\}$; according to \thref{thm:AKLTnu1} with $m=g=2$ and $s=1/2$ ($s=1/3$ when $n=3$),  in both cases we can deduce that
\begin{align}
\nu(\Omega(\mu, \scrM))&\geq  \frac{(\sqrt{1+\gamma}-1)\nu_2(\mu)}{2\sqrt{1+\gamma}}\nonumber\\
&\geq  \frac{(\sqrt{1+c_k}-1)\nu_2(\mu)}{2\sqrt{1+c_k}}, 
\end{align}
where $\gamma=\gamma(H_{\frac{1}{2},\frac{1}{2}}(n))$. Here the second inequality follows from \thref{thm:gapGM} and is applicable when $n>2k$ and $k>2$. This bound has the same form as  the counterpart \eref{eq:nuLBeven} for the even closed chain.  In addition, \ecref{eq:nuLBeven2}{eq:NUBeven} are also applicable in the large-$n$ limit as long as $\gamma(H_{\frac{1}{2},\frac{1}{2}}(n))\approx\gamma(H^{\circ}(n))$ (cf. \tref{tab:AKLTgap}).

Numerical calculation based on \eref{eq:normSDP} further suggests that the optimal coloring protocol is also optimal among all matching protocols. \Fref{fig:ClosedOpenOpt} shows the performance of the optimal matching protocol in comparison with the counterpart for the closed chain. For a given number of nodes, it is easier to verify the AKLT state on the open chain than the one on the closed chain.

\section{\label{sec:VAKLTgenGraph}Verification of AKLT states on general graphs}
To further illustrate the power of our general approach, here we consider in more detail the verification  of AKLT states associated with general connected graphs $G(V,E)$ up to five vertices, that is, $n=|V|\leq 5$. Recall that there are 1 connected graph of two  vertices,  2 connected graphs of three  vertices,  6 connected graphs of four  vertices, and  21 connected graphs of five  vertices up to isomorphism. In \tref{tab:VAKLTgen} in \aref{app:VAKLTgen} we have summarized relevant basic information about the 30 graphs and the corresponding AKLT states, including the degree, matching number, chromatic number, chromatic index, the dimension of the underlying Hilbert space, and the spectral gap of the Hamiltonian.

To construct a verification protocol for the AKLT state associated with a given graph, it is essential  to choose a suitable  bond verification protocol.  Here we focus on the protocol based on the distribution $\mu_{32}$,  which corresponds to the pentakis dodecahedron as described in \sref{sec:BondVC}. This bond verification protocol can achieve the largest bond spectral gap (as the isotropic protocol) for all edges in  graphs up to five vertices, since the corresponding distribution $\mu_{32}$ forms a spherical 9-design. Incidentally,  for graphs up to four vertices, the alternative bond verification protocol based on the distribution $\mu_{24}$ (cf.~\sref{sec:BondVC}) can achieve the same performance. 

\begin{figure}[t]
	\hspace{-0.8ex}	\includegraphics[width=0.46\textwidth]{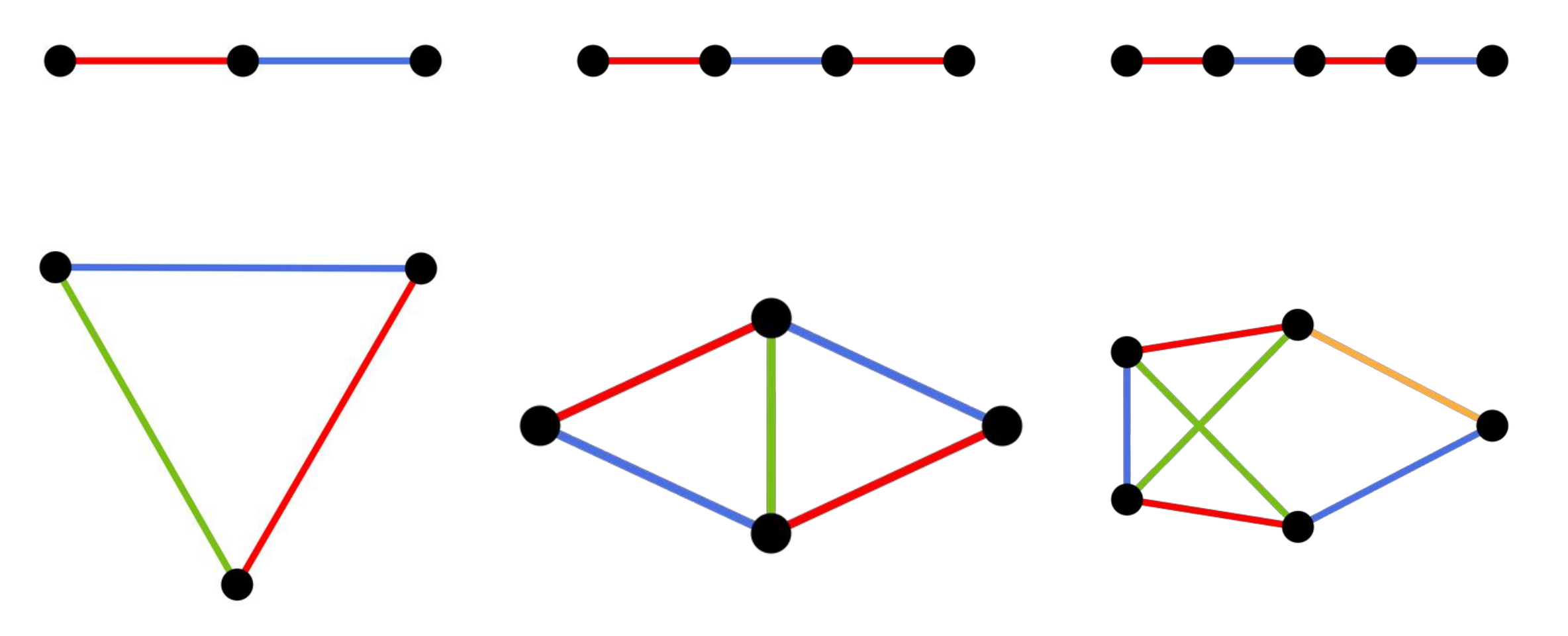} 
	\caption{\label{fig:nuMaxMin}Connected graphs of three, four, and five vertices for which the verification operators (based on optimized coloring protocols) of the corresponding AKLT states have the largest spectral gaps (up) and smallest spectral gaps (down).}
\end{figure}

\begin{figure}[t]
	\hspace{-0.8ex}	\includegraphics[width=0.45\textwidth]{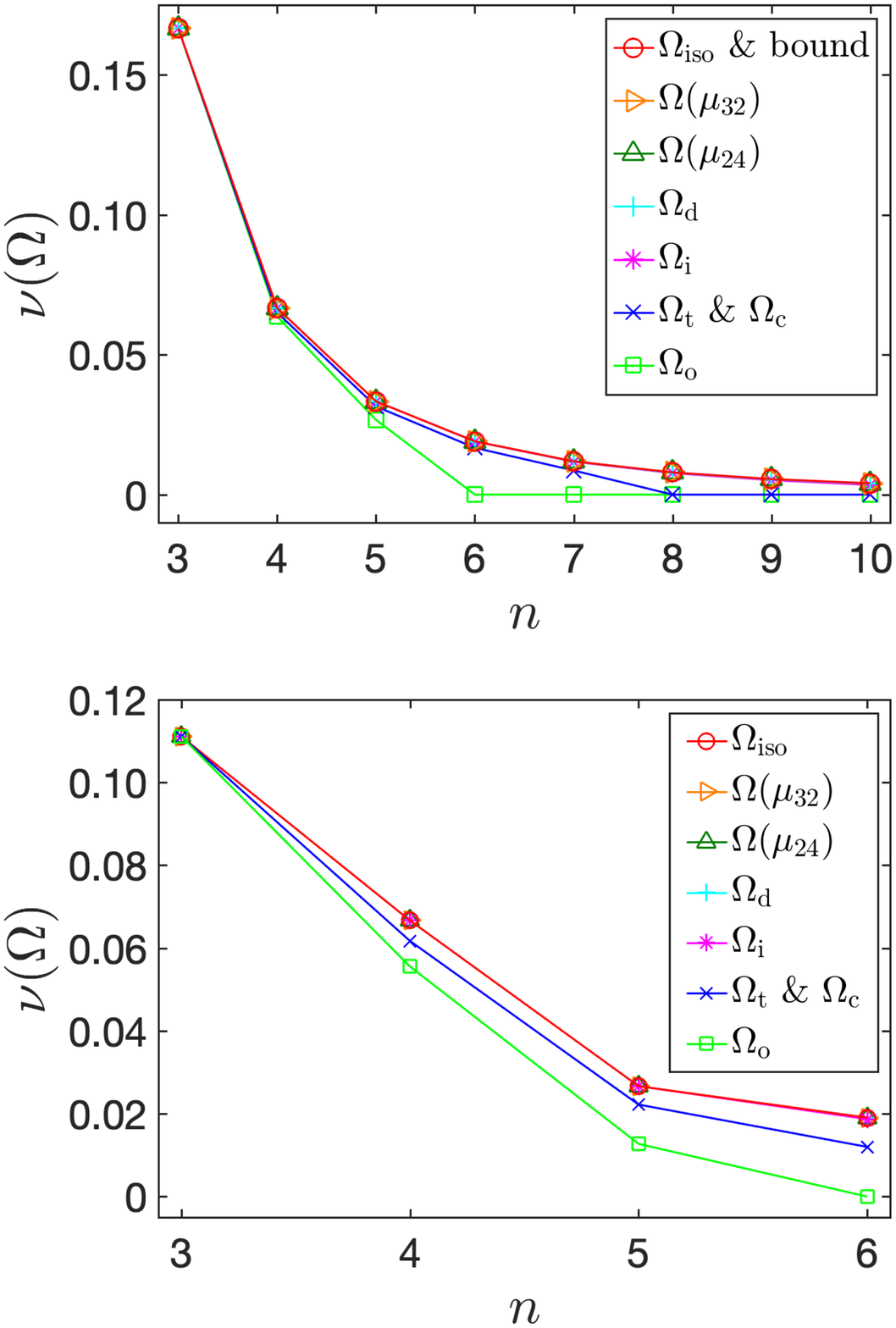} 
	\caption{\label{fig:StarComplete} 	Spectral gaps of verification operators of AKLT states defined on star graphs with with 3 to 10 vertices (upper plot) and complete graphs with 3 to 6 vertices (lower plot). Every  verification operator is based on the optimal edge coloring (which coincides with the trivial edge coloring for the star graph) with uniform probabilities (which are also optimal). The bond verification protocols are based on platonic solids,  distributions $\mu_{24}, \mu_{32}$, and the isotropic distribution defined in \sref{sec:BondVC} as indicated in the legends. Note that many different bond verification protocols lead to nearly identical spectral gaps.
		For each star graph, the spectral gap of the verification operator based on the isotropic distribution coincides with the first lower bound in 	\eref{eq:SpectralGapColor}. 	
	}
\end{figure}

Given the bond verification protocol, a verification protocol for the  AKLT state $|\Psi_G\>$ is determined by a matching cover of the underlying graph $G$.
We are particularly interested in coloring protocols, which correspond to matching covers composed of disjoint matchings. The simplest protocol is based on the trivial edge coloring: all edges have different colors. The spectral gap of the resulting verification protocol (with 
 uniform probabilities for all the colors) for each graph up to five vertices is shown in \tref{tab:VAKLTgen} in \aref{app:VAKLTgen}. For comparison, the table also shows the spectral gaps of two verification protocols associated with an optimal edge coloring: one protocol employs  the uniform probabilities for all the colors, while the other one employs the optimized probabilities, which can be determined by SDP according to \eref{eq:normSDP}. 
  For each star graph and the 3-cycle, all three protocols coincide with each other, and so do the corresponding spectral gaps. For most other graphs, by contrast,  the performance can be improved by considering an optimal edge coloring together with the optimized probabilities. For graphs with a given number of vertices, calculation shows that the  spectral gap is maximized at the linear graph as shown in \fref{fig:nuMaxMin}. We have not found a general pattern for the graph that leads to the smallest spectral gap.

Next, we discuss briefly verification protocols based on different bond verification protocols as discussed in \sref{sec:BondVC}.
\Fref{fig:StarComplete} illustrates the dependence of the spectral gap on the bond verification operator in the case of star graphs with 3 to 10 vertices and complete graphs with 3 to 6 vertices. Here each  verification protocol is based on the optimal coloring (which coincides with the trivial coloring for each star graph) with uniform probabilities (which are also optimal). In general, the relative deviations in the spectral gaps tend to increase as the number of nodes increases. The spectral gap of the verification operator based on the octahedron ($\Omega_\rmo$) vanishes when $n\geq 6$ for both star graphs and complete graphs. By contrast, the spectral gaps of verification operators based on icosahedron ($\Omega_\rmi$),
dodecahedron ($\Omega_\rmd$), distribution $\mu_{24}$  ($\Omega(\mu_{24})$), distribution $\mu_{32}$  ($\Omega(\mu_{32})$), and isotropic distribution ($\Omega_{\mathrm{iso}}$) are  close to each other in all the cases under consideration.

\section{\label{sec:summary}Summary}
We proposed a general method for constructing efficient verification protocols for  AKLT states defined on  arbitrary graphs based on local spin measurements. Explicit expressions for the AKLT states are not necessary to apply our approach. Given an AKLT state,  our verification protocols can be constructed from  probability distributions on the unit sphere and matching covers (including edge colorings) of the underlying graph, which have a simple geometric and graphic picture. 
We also provide rigorous performance guarantee that is required for practical applications. 
With our approach, most AKLT states of wide interest, including those defined on 1D and 2D lattices, can be verified with constant sample cost, which is independent 
of the system size and is dramatically more efficient than all  approaches known in the literature. Our verification protocols will be useful to various tasks in quantum information processing that employ AKLT states, including measurement-based quantum computation in particular.

\section*{Acknowledgments}
H. Zhu is grateful to Zheng Yan and Penghui Yao for stimulating discussions. This work is  supported by   the National Natural Science Foundation of China (Grants No.~92165109 and No.~11875110) and  Shanghai Municipal Science and Technology Major Project (Grant No.~2019SHZDZX01).

\appendix

\section{\label{sec:SpinProjProbProof}Proof of \lref{lem:SpinProjProb}}

\begin{proof}[Proof of \lref{lem:SpinProjProb}]
In the special case $\vec{r}=\vec{s}$, the probability $p_{\vec{r},\vec{s}}(S_1,m_1;S_2,m_2)$ is independent of the unit vector $\vec{r}$ and so can be abbreviated as $p(S_1,m_1;S_2,m_2)$;  to simplify the computation, we can assume that $\vec{r}=\hat{z}$. 	Let $S=S_1+S_2$,  $m=m_1+m_2$, and denote by $|S,m\>$ the eigenstate of $S_z=S_{1,z}+S_{2,z}$ with eigenvalue $m$. Then 
\begin{align}
&p(S_1,m_1;S_2,m_2)=|\<S, m|S_1,m_1; S_2,m_2\>|^2 \nonumber\\
&=\binom{2S}{2S_1}\binom{S+m}{S_1+m_1}\binom{S-m}{S_1-m_1}\leq 1,
\end{align}
where the second equality follows from  the well known formula for the Clebsch-Gordon coefficients \cite{Bohm93book}. In addition,  the last inequality is saturated iff $m=\pm S$, which means either \eref{eq:SpinCa} or \eqref{eq:SpinCb} holds.

 When $\vec{r}=-\vec{s}$,  the inequality $p_{\vec{r},\vec{s}}(S_1,m_1;S_2,m_2)\leq1$ is saturated iff either \eref{eq:SpinCc} or \eqref{eq:SpinCd} holds according to the above analysis and following equalities,
\begin{align}
p_{\vec{r},\vec{s}}(S_1,m_1;S_2,m_2)&=p_{-\vec{r},\vec{s}}(S_1, -m_1;S_2,m_2)\nonumber\\
&=p_{\vec{r},-\vec{s}}(S_1,m_1;S_2,-m_2). \label{eq:InversionSym}
\end{align}

In general, $|S_2,m_2\>_{\vec{s}}$ can be expanded as follows,
	\begin{align}
	|S_2,m_2\>_{\vec{s}}=\sum_{k=-S_2}^{S_2}c_k |S_2,k\>_{\vec{r}},
	\end{align}
	where the coefficients $c_k$  satisfy the normalization condition $\sum_{k=-S_2}^{S_2} |c_k|^2=1$. 
	Since the projector $P_S$ commutes with $(\vec{S}_1+\vec{S}_2)\cdot \vec{r}$, it follows that
	\begin{align}
	p_{\vec{r},\vec{s}}(S_1,m_1;S_2,m_2)=\sum_{k=-S_2}^{S_2}|c_k|^2 p(S_1,m_1;S_2,k)\leq 1, \label{eqprsmk}
	\end{align}
	and the inequality is saturated iff 
\begin{align}
p(S_1,m_1;S_2,k)=1\quad \forall c_k\neq 0,
\end{align}	
 which means $c_k=0$ whenever $p(S_1,m_1; S_2, k)<1$.

Recall that $p(S_1,m_1;S_2,k)\leq 1$, and the inequality is saturated iff $m_1+k=\pm (S_1+S_2)$. 
	Suppose the inequality in \eref{eqprsmk} is saturated, then $m_1=\pm S_1$. By symmetry, we also have $m_2=\pm S_2$. 
	If $m_1=S_1$, then 
	\begin{align}
	|c_k|=\begin{cases}
	1 & k=S_2, \\
	0 & \mbox{otherwise},
	\end{cases}
	\end{align}
	which implies that $|S_2, m_2\>_{\vec{s}}=|S_2,S_2\>_{\vec{r}}$ up to an overall phase factor, so either \eref{eq:SpinCa} or \eqref{eq:SpinCc} holds in view of \eref{eq:rsOverlap}. If $m_1=-S_1$, then either \eref{eq:SpinCb} or \eqref{eq:SpinCd} holds according to a similar reasoning.
\end{proof}

\section{\label{asec:nuSmuProof}Proof of \lref{lem:nuSmu}}
\begin{proof}[Proof of \lref{lem:nuSmu}]
	To prove \lref{lem:nuSmu}, it suffices to prove \eref{eq:nuSmuIneq}. According to \eref{eq:nuSmu}, we have
	\begin{align}
	\nu_{S_j}(\mu)=\lambda_{\min}(O_{S_j}),\quad j=1,2,
	\end{align}
	where 
	\begin{align}
	O_{S_j}=O_{S_j}(\mu)=2\int |S_j\>_{\vec{r}}\<S_j| d\mu_\sym(\vec{r}),\;\;  j=1, 2
	\end{align}
	according to  \eref{eq:MSmu}. Define
	\begin{align}
	W_j:=&\int \Bigl(\Bigl|\frac{1}{2}\bigr\>_{\vec{r}}\Bigl\<\frac{1}{2}\Bigr|\Bigr)^{\otimes 2S_j} d\mu_\sym(\vec{r})\nonumber\\
	=&\int \Bigl(\frac{1+\vec{r}\cdot\vec{\sigma}}{2}\Bigr)^{\otimes 2S_j} d\mu_\sym(\vec{r}), \quad  j=1, 2,
	\end{align}
	where $\vec{\sigma}=(\sigma_x,\sigma_y,\sigma_z)$ is the vector composed of the three Pauli operators. 
	Then $W_j$ for each $j$ is a positive semidefinite operator acting on the symmetric subspace of $\bbC^{\otimes2S_j}$. Note that this symmetric subspace has dimension $2S_j+1$, which is the same as the Hilbert space associated with spin value $S_j$. In addition, $W_1$ is the partial trace of $W_2$ after tracing out $2(S_2-S_1)$ qubits. Moreover, $W_j$ and $O_{S_j}$ have the same nonzero eigenvalues, including multiplicities. Let $\Pi_{2S_j}$ be the projector onto the  symmetric subspace of $\bbC^{\otimes2S_j}$; then 
	\begin{align}
	\nu_{S_j}(\mu) =\lambda_{\min}(O_{S_j})=\max\{\lambda | W_j\geq \lambda \Pi_{2S_j} \}. 
	\end{align}
	Notably, we have
	\begin{align}
	W_2\geq \nu_{S_2}(\mu) \Pi_{2S_2}. 
	\end{align}
	Taking partial trace over $2(S_2-S_1)$ qubits we can deduce that
	\begin{align}
	W_1\geq \nu_{S_2}(\mu) \frac{2S_2+1}{2S_1+1}\Pi_{2S_1}, 
	\end{align}
	which implies \eref{eq:nuSmuIneq} and completes the proof of \lref{lem:nuSmu}. 
\end{proof}

\section{\label{asec:OmegaSdesign}Proof of  \thref{thm:OmegaSdesign}}

\begin{proof}[Proof of \thref{thm:OmegaSdesign}]
From the definition in \eref{eq:OmegaS} we can deduce that
\begin{align}
\tr[\Omega_S(\mu)]=2S-1. \label{eq:OmegaStrace}
\end{align}	
Suppose statement~1 holds; then
	\begin{equation}
	\|\Omega_S(\mu)\|=\frac{2S-1}{2S+1}=\frac{\tr[\Omega_S(\mu)]}{2S+1},
	\end{equation}
according to \eref{eq:nuSmu}, so statement~2 must hold given that $\Omega_S(\mu)$ is a positive operator  acting on a Hilbert space of dimension  $2S+1$.	 The implication  $2\imply 3$	 is obvious.

Next, if $\Omega_S(\mu)$ is proportional to  the identity operator, then both statements~1 and 2 hold by \eref{eq:OmegaStrace}. It follows that
 statements 1, 2, 3 are equivalent to each other.

To complete the proof, it remains to  show the equivalence of statements 2 and 4. If statement 2 holds, then
\begin{align}
\tr[\Omega_S(\mu)^2]=\frac{(2S-1)^2}{2S+1}. 
\end{align}
So statement~4 holds according to  \lref{lem:trOmegaSmu2LB}  below.

Conversely, suppose $\mu_\sym$ forms a spherical $t$-design with $t=2S$. By virtue of \lref{lem:trOmegaSmu2LB} below we can deduce 
	\begin{align}
	\tr\bigl[\Omega_S(\mu)^2\bigr]=\frac{(2S-1)^2}{2S+1}=\frac{1}{2S+1}\tr[\Omega_S(\mu)]^2,
	\end{align}
which implies statement~2 given that $\Omega_S(\mu)$ is a positive operator acting on a Hilbert space of dimension  $2S+1$. This observation completes the proof of \thref{thm:OmegaSdesign}. 
\end{proof}

In the rest of this appendix we prove the following lemma, which is employed in the proof of \thref{thm:OmegaSdesign}. 
\begin{lem}\label{lem:trOmegaSmu2LB}
Suppose $\mu$ is a probability distribution on the unit sphere. Then the operator $\Omega_S(\mu)$ satisfies
	\begin{align}
	\tr\bigl[\Omega_S(\mu)^2\bigr]\geq\frac{(2S-1)^2}{2S+1};  \label{eq:trOmegaSmu2LB}
	\end{align}	
the inequality is saturated iff the distribution $\mu_\sym$ forms a spherical $t$-design with $t=2S$. 
\end{lem}

\begin{proof}[Proof of \lref{lem:trOmegaSmu2LB}]
	The inequality in \eref{eq:trOmegaSmu2LB} follows from \eref{eq:OmegaStrace} and the fact that $\Omega_S(\mu)$ is a positive operator  acting on a Hilbert space of dimension  $2S+1$. 
By virtue of  \esref{eq:OmegaS}{eq:rsOverlap} in the main text, we can deduce that
	\begin{align}
&\tr\bigl[\Omega_S(\mu)^2\bigr]=\tr\Biggl\{\biggl[1-2\int |S\>_{\vec{r}}\<S| d\mu_\sym(\vec{r})\biggr]^2\Biggr\}\nonumber\\
&=2S-3+4	\iint |{}_{\bm{r}}\<S|S\>_{\bm{s}}|^2\rmd\mu_\sym(\vec{r})\rmd\mu_\sym(\vec{s})
\nonumber\\
&=2S-3+4\iint\Bigl(\frac{1+\vec{r}\cdot\vec{s}}{2}\Bigr)^{2S}\rmd\mu_\sym(\vec{r})\rmd\mu_\sym(\vec{s})
\nonumber\\
&=2S-3+2^{2-2S}\sum_{j=0}^{\lfloor S\rfloor}\binom{2S}{2j}f_{2j}(\mu_\sym) \nonumber\\
&=2S-3+2^{2-2S}\sum_{j=0}^{\lfloor S\rfloor}\binom{2S}{2j}f_{2j}(\mu), \label{eq:OmegaSProof1}
\end{align}	
where 	
	\begin{align}
	f_t(\mu):=\iint \rmd\mu(\vec{r})\rmd\mu(\vec{s}) (\bm{r}\cdot\bm{s})^t
	\end{align}
denotes the $t$th frame potential of the distribution $\mu$, assuming that $t$ is a nonnegative integer. By convention we have $f_0(\mu)=1$ for any distribution $\mu$.

When $t$ is even,  the frame potential $	f_t(\mu)$ satisfies the following inequality \cite{Seid01,BannB09}, 
\begin{align}
	f_t(\mu)=f_t(\mu_\sym)\geq \frac{1}{t+1}, \label{eq:ftmuLB}
\end{align}
 and the inequality is saturated if $\mu_\sym$ forms a spherical $t$-design. Combining \esref{eq:OmegaSProof1}{eq:ftmuLB} we can deduce that
\begin{align}
	\tr\bigl[\Omega_S(\mu)^2\bigr]&\geq 2S-3+2^{2-2S}\sum_{j=0}^{\lfloor S\rfloor}\binom{2S}{2j}\frac{1}{2j+1}\nonumber\\
	&=2S-3+\frac{2^{2-2S}}{2S+1}\sum_{j=0}^{\lfloor S\rfloor}\binom{2S+1}{2j+1}\nonumber\\
	&=2S-3+\frac{4}{2S+1}=\frac{(2S-1)^2}{2S+1}, \label{eq:trOmegamu2LBProof}
\end{align}
which confirms the inequality in \eref{eq:trOmegaSmu2LB} again. 
	
Suppose the symmetrized probability distribution $\mu_\sym$ forms a spherical $t$-design with $t=2S$, then we have
	\begin{align}
	\!\! f_{2j}(\mu)=f_{2j}(\mu_\sym)=\frac{1}{2j+1}, \;\; j=0, 1, \ldots, \lfloor S\rfloor.   \label{eq:f2jmudesign}
	\end{align}
Consequently,  the inequality in \eref{eq:trOmegamu2LBProof} is saturated, which means the inequality in \eref{eq:trOmegaSmu2LB} is also saturated. 
	
To prove the other direction, suppose the inequality in \eref{eq:trOmegaSmu2LB} is saturated, so that the inequality in \eref{eq:trOmegamu2LBProof} is saturated. Then  \eref{eq:f2jmudesign} must hold, which means  $\mu_\sym$ forms a spherical $t$-design with $t=2S$ since  the distribution $\mu_\sym$ is symmetric under center inversion.  
\end{proof}

\clearpage

\onecolumngrid
\section{\label{app:VAKLTgen}Verification of AKLT states on general graphs}

\begin{longtable*}{c >{\centering\arraybackslash}p{2cm} >{\centering\arraybackslash}p{0.5cm} >{\centering\arraybackslash}p{0.5cm} >{\centering\arraybackslash}p{1cm} >{\centering\arraybackslash}p{1cm} >{\centering\arraybackslash}p{1cm} >{\centering\arraybackslash}p{1cm} >{\centering\arraybackslash}p{1cm} >{\centering\arraybackslash}p{1cm} >{\centering\arraybackslash}p{1cm} >{\centering\arraybackslash}p{1cm} >{\centering\arraybackslash}p{1cm} >{\centering\arraybackslash}p{2.5cm}}
	\caption{\label{tab:VAKLTgen}Verification of AKLT states on general graphs of two to five vertices.  The graphs \cite{ConnectedGraph} with  optimal edge colorings are shown in the second column. For each graph
	 $\nu(\Omega_{\mathrm{tri}})$ is the spectral gap of the verification operator based on the trivial edge coloring with uniform probabilities; $\nu(\Omega)$ is  based on the optimal edge coloring (shown in the second column) with uniform probabilities; $\nu(\tilde{\Omega})$ is based on the optimal edge coloring with optimized probabilities as shown in the last column according to the order: red (R), blue (B), green (G), orange (O), and magenta (M). For graphs No. 1, 2, 3, 4, 7, 9, 10, 11, 18, 27, 30, the optimized probabilities are uniform due to symmetry; 
	 for graphs No. 6, 14, 22, 26, the optimized probabilities are uniform by coincidence. 
	 All  bond verification protocols employed are based on the distribution $\mu_{32}$,  which corresponds to the pentakis dodecahedron as described in \sref{sec:BondVC}. For completeness, the table also shows the vertex number $|V|$,  edge number $|E|$,  degree $\Delta(G)$, matching number $\upsilon(G)$, chromatic number $\chi(G)$,  chromatic index $\chi'(G)$ of the graph $G$; in addition, the table shows the dimension $\dim \caH$ of the underlying Hilbert space and the spectral gap $\gamma(H_G)$ of the AKLT Hamiltonian.}\\
	\hline \hline
	\vphantom{\LARGE f} No. & graph $G$  & $|V|$ & $|E|$ & $\Delta(G)$ & $\upsilon(G)$ & $\chi(G)$ & $\chi'(G)$ & dim $\caH$ & $\gamma(H_G)$    & $\nu(\Omega_{\mathrm{tri}})$ & $\nu(\Omega)$  & $\nu(\tilde{\Omega})$ & $p$(R, B, G, O, M)\\[0.2ex]     
	\hline   
	\endfirsthead	
		\hline \hline
	\vphantom{\LARGE f} No. & graph $G$  & $|V|$ & $|E|$ & $\Delta(G)$ & $v(G)$ & $\chi(G)$ & $\chi'(G)$ & dim $\caH$ & $\gamma(H_G)$    & $\nu(\Omega_{\mathrm{tri}})$ & $\nu(\Omega)$  & $\nu(\tilde{\Omega})$ & $p$(R, B, G, O, M)\\[0.2ex]     
	\hline   
	\endhead	                                         
	1   & \begin{minipage}[b]{0.1\columnwidth}\centering\raisebox{-.45\height}{\includegraphics[width=\linewidth]{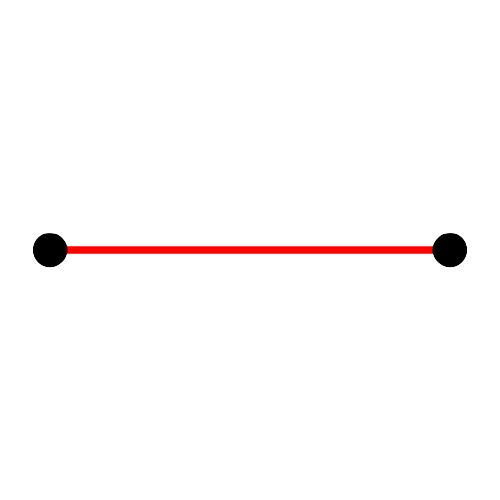}}\end{minipage}  & 2     & 1     & 1     & 1      & 2         & 1          & 4          & 1              & $\frac{2}{3}$                      & $\frac{2}{3}$  & $\frac{2}{3}$         & $\begin{pmatrix} 1 \end{pmatrix}$\\
	\hline                             
	2   & \begin{minipage}[b]{0.1\columnwidth}\centering\raisebox{-.45\height}{\includegraphics[width=\linewidth]{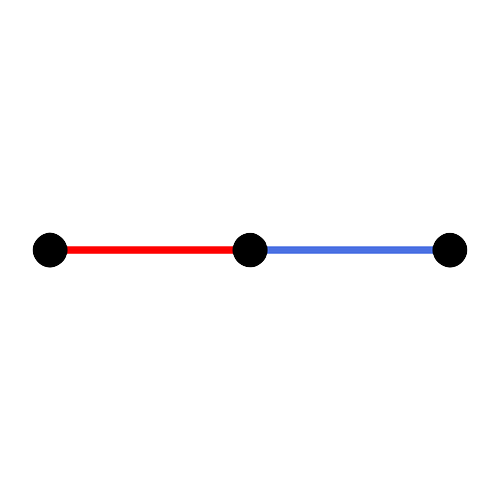}}\end{minipage}  & 3     & 2     & 2      & 1     & 2         & 2          & 12         & $\frac{2}{3}$  & $\frac{1}{6}$                      & $\frac{1}{6}$  & $\frac{1}{6}$         & $\frac{1}{2}\begin{pmatrix} 1\\ 1 \end{pmatrix}$\\                        
	3   & \begin{minipage}[b]{0.1\columnwidth}\centering\raisebox{-.45\height}{\includegraphics[width=\linewidth]{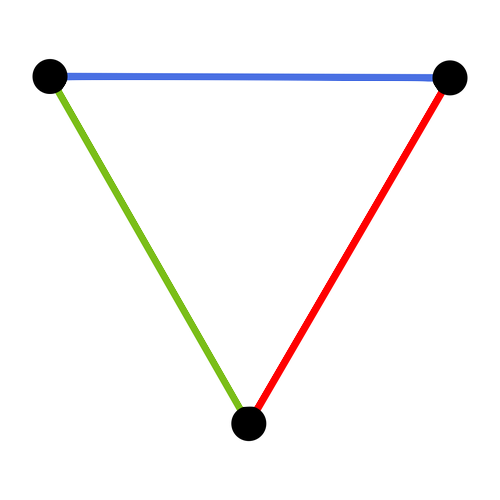}}\end{minipage}  & 3     & 3     & 2       & 1    & 3         & 3          & 27         & $\frac{5}{6}$  & $\frac{1}{9}$                      & $\frac{1}{9}$  & $\frac{1}{9}$         & $\frac{1}{3} \begin{pmatrix} 1\\ 1\\ 1\end{pmatrix}$\\ 
	\hline                   
	4   & \begin{minipage}[b]{0.1\columnwidth}\centering\raisebox{-.45\height}{\includegraphics[width=\linewidth]{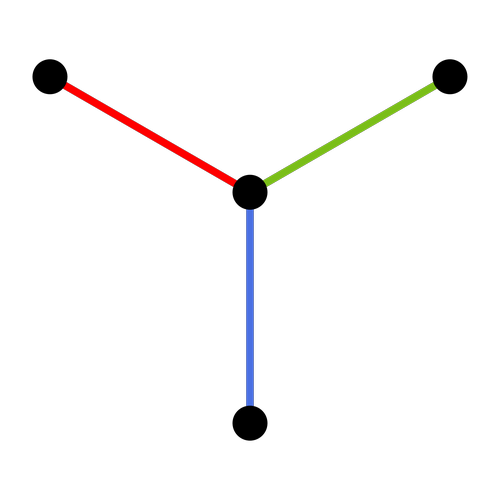}}\end{minipage}  & 4     & 3     & 3      & 1     & 2         & 3          & 32         & $\frac{1}{2}$  & $\frac{1}{15}$                      & $\frac{1}{15}$ & $\frac{1}{15}$        & $\frac{1}{3} \begin{pmatrix} 1\\ 1\\ 1\end{pmatrix}$\\                    
	5   & \begin{minipage}[b]{0.1\columnwidth}\centering\raisebox{-.45\height}{\includegraphics[width=\linewidth]{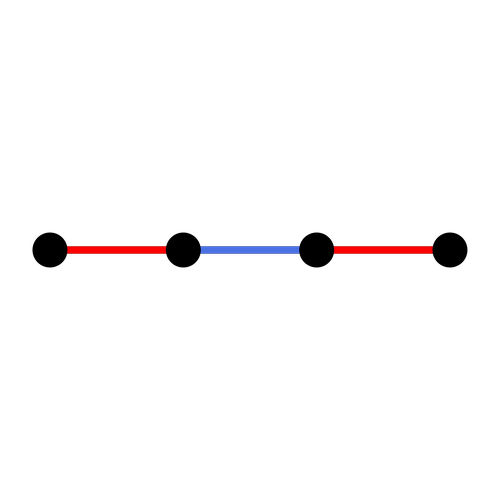}}\end{minipage}  & 4     & 3     & 2      & 2     & 2         & 2          & 36         & 0.5168         & 0.0755                      & 0.1119         & 0.1134                & $\begin{pmatrix} 0.4526\\ 0.5474 \end{pmatrix}$\\                           
	6   & \begin{minipage}[b]{0.1\columnwidth}\centering\raisebox{-.45\height}{\includegraphics[width=\linewidth]{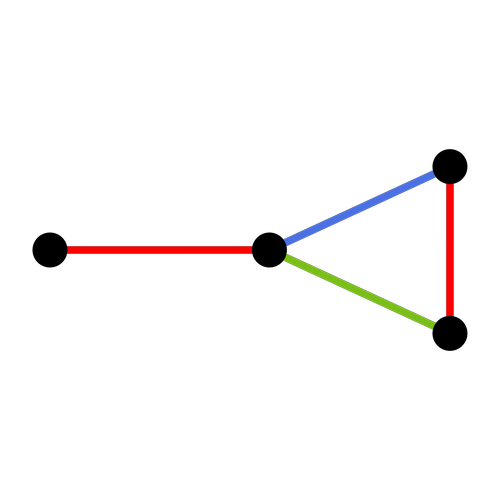}}\end{minipage}  & 4     & 4     & 3      & 2     & 3         & 3          & 72         & 0.5595         & $\frac{1}{20}$                      & $\frac{1}{15}$ & $\frac{1}{15}$        & $\frac{1}{3} \begin{pmatrix} 1\\ 1\\ 1\end{pmatrix}$\\                    
	7   & \begin{minipage}[b]{0.1\columnwidth}\centering\raisebox{-.45\height}{\includegraphics[width=\linewidth]{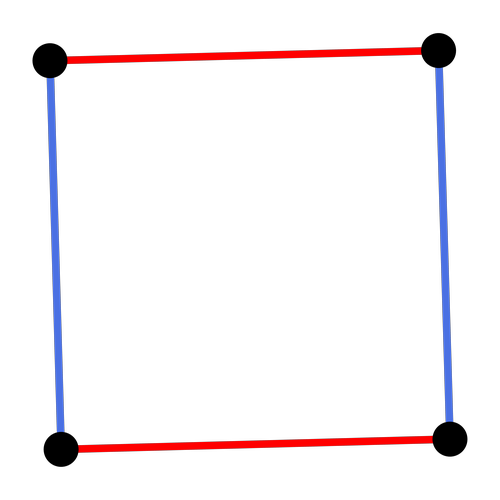}}\end{minipage}  & 4     & 4     & 2       & 2    & 2         & 2          & 81         & $\frac{1}{3}$  & $\frac{1}{30}$                      & $\frac{1}{15}$ & $\frac{1}{15}$        & $\frac{1}{2}\begin{pmatrix}   1\\ 1\end{pmatrix} $\\                       
	8   & \begin{minipage}[b]{0.1\columnwidth}\centering\raisebox{-.45\height}{\includegraphics[width=\linewidth]{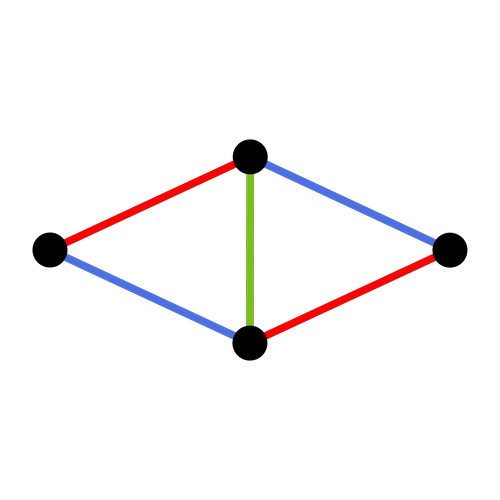}}\end{minipage}  & 4     & 5     & 3      & 2     & 3         & 3          & 144        & $\frac{1}{2}$  & $\frac{1}{30}$                      & 0.0556         & 0.0618                & $\begin{pmatrix} 0.3708\\ 0.3708\\ 0.2583\end{pmatrix}$\\                   
	9   & \begin{minipage}[b]{0.1\columnwidth}\centering\raisebox{-.45\height}{\includegraphics[width=\linewidth]{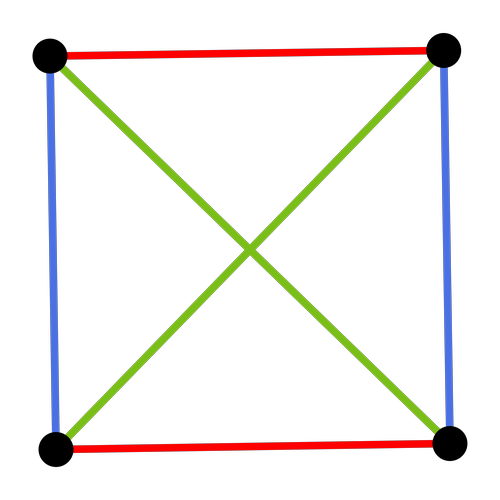}}\end{minipage}  & 4     & 6     & 3      & 2     & 4         & 3          & 256        & $\frac{7}{10}$ & $\frac{1}{30}$                     & $\frac{1}{15}$ & $\frac{1}{15}$        & $\frac{1}{3} \begin{pmatrix} 1\\ 1\\ 1\end{pmatrix}$ \\                   
	\hline
	10  & \begin{minipage}[b]{0.1\columnwidth}\centering\raisebox{-.45\height}{\includegraphics[width=\linewidth]{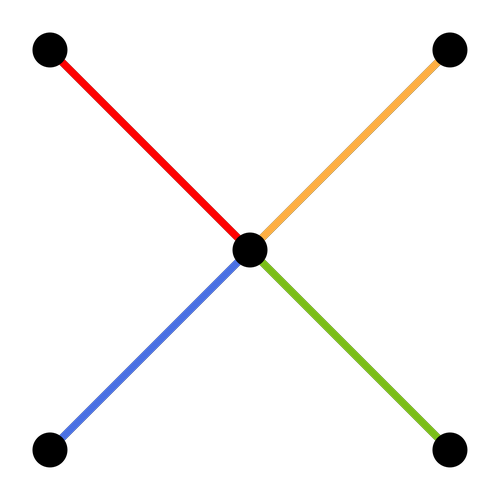}}\end{minipage} & 5     & 4     & 4       & 1    & 2         & 4          & 80         & $\frac{2}{5}$  & $\frac{1}{30}$                      & $\frac{1}{30}$ & $\frac{1}{30}$        & $\frac{1}{4} \begin{pmatrix} 1\\ 1\\ 1\\ 1\end{pmatrix}$  \\             
	11  & \begin{minipage}[b]{0.1\columnwidth}\centering\raisebox{-.45\height}{\includegraphics[width=\linewidth]{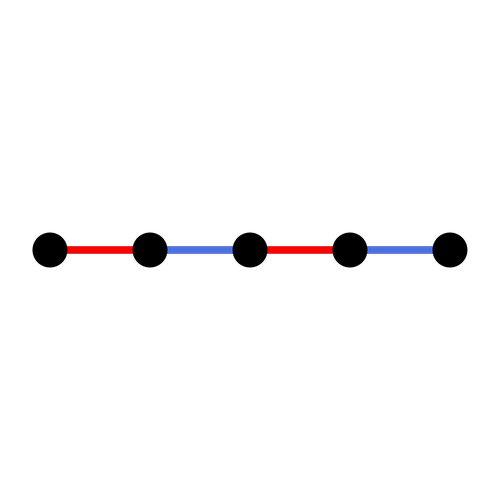}}\end{minipage} & 5     & 4     & 2      & 2     & 2         & 2          & 108        & 0.4539         & 0.0476                      & 0.0941         & 0.0941                & $\frac{1}{2} \begin{pmatrix} 1\\ 1 \end{pmatrix}$ \\                      
	12  & \begin{minipage}[b]{0.1\columnwidth}\centering\raisebox{-.45\height}{\includegraphics[width=\linewidth]{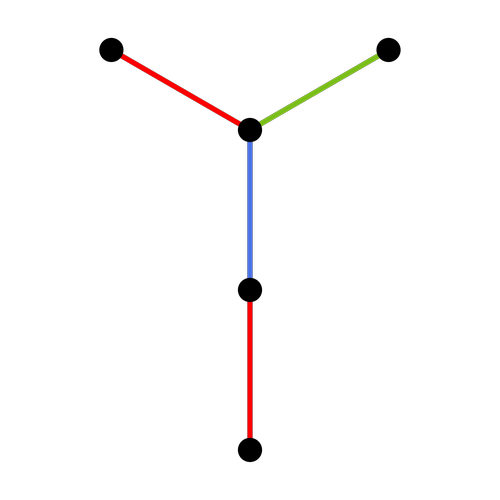}}\end{minipage} & 5     & 4     & 3      & 2     & 2         & 3          & 96         & 0.4117         & 0.0385                      & 0.0511         & 0.0529                & $\begin{pmatrix} 0.3170\\ 0.4018\\ 0.2812 \end{pmatrix}$\\                   
	13  & \begin{minipage}[b]{0.1\columnwidth}\centering\raisebox{-.45\height}{\includegraphics[width=\linewidth]{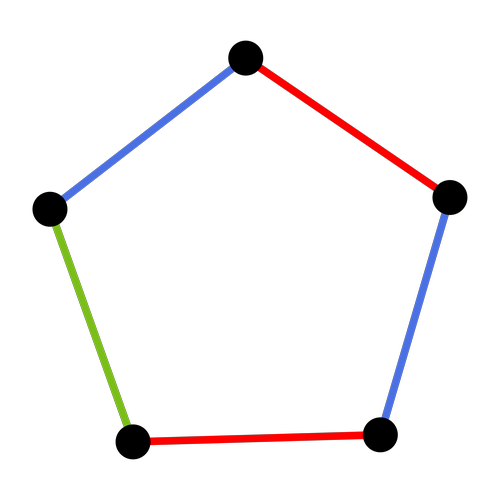}}\end{minipage} & 5     & 5     & 2      & 2     & 3         & 3          & 243        & 0.4540          & 0.0363                      & 0.0597         & 0.0603                & $\begin{pmatrix} 0.3368\\ 0.3368\\ 0.3264 \end{pmatrix}$\\                   
	14  & \begin{minipage}[b]{0.1\columnwidth}\centering\raisebox{-.45\height}{\includegraphics[width=\linewidth]{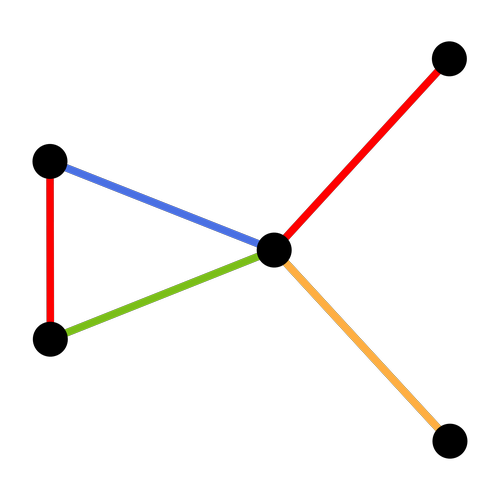}}\end{minipage} & 5     & 5     & 4      & 2     & 3         & 4          & 180        & 0.4295         & $\frac{2}{75}$                      & $\frac{1}{30}$ & $\frac{1}{30}$        & $\frac{1}{4} \begin{pmatrix} 1\\ 1\\ 1\\ 1\end{pmatrix}$ \\              
	15  & \begin{minipage}[b]{0.1\columnwidth}\centering\raisebox{-.45\height}{\includegraphics[width=\linewidth]{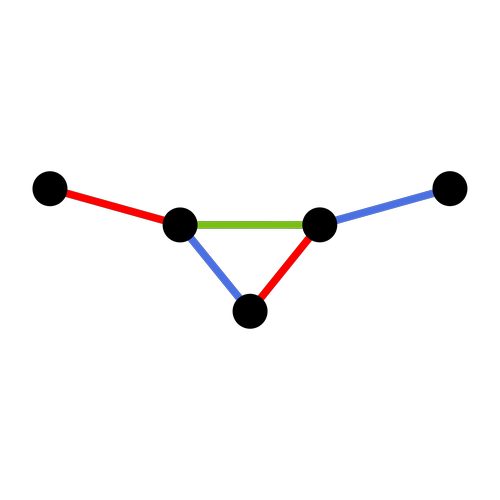}}\end{minipage} & 5     & 5     & 3     & 2      & 3         & 3          & 192        & 0.4796         & 0.0316                      & 0.0527         & 0.0547                & $\begin{pmatrix} 0.2975\\ 0.2975\\ 0.4050 \end{pmatrix}$\\                   
	16  & \begin{minipage}[b]{0.1\columnwidth}\centering\raisebox{-.45\height}{\includegraphics[width=\linewidth]{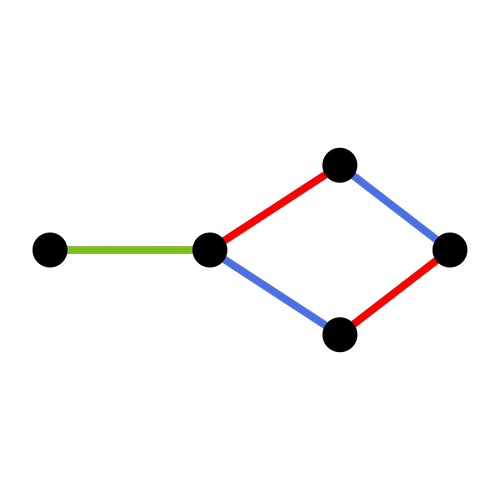}}\end{minipage} & 5     & 5     & 3      & 2     & 2         & 3          & 216        & 0.2871         & 0.0206                      & 0.0344         & 0.0369                & $\begin{pmatrix} 0.3892\\ 0.3892\\ 0.2214\end{pmatrix}$\\                   
	17  & \begin{minipage}[b]{0.1\columnwidth}\centering\raisebox{-.45\height}{\includegraphics[width=\linewidth]{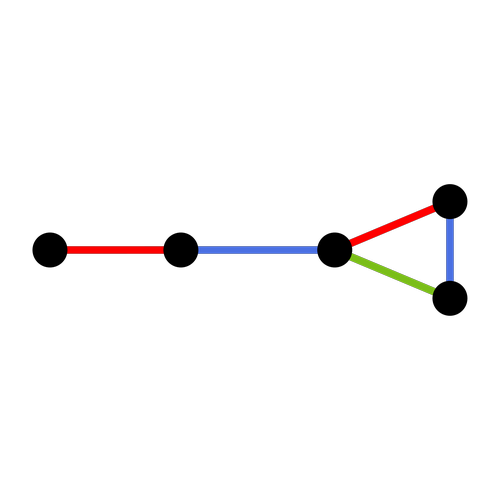}}\end{minipage} & 5     & 5     & 3      & 2     & 3         & 3          & 216        & 0.4396         & 0.0308                      & 0.0511         & 0.0529                & $\begin{pmatrix} 0.3122\\ 0.4018\\ 0.2860\end{pmatrix}$\\                   
	18  & \begin{minipage}[b]{0.1\columnwidth}\centering\raisebox{-.45\height}{\includegraphics[width=\linewidth]{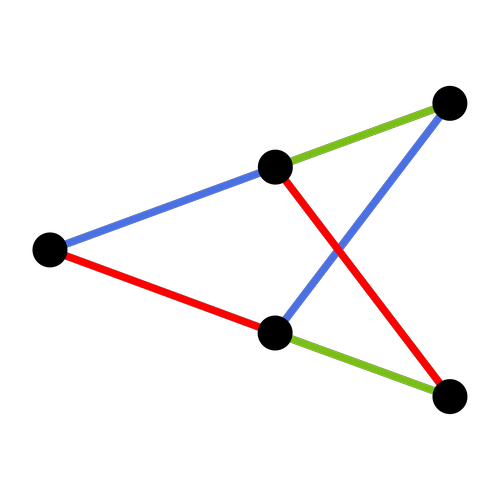}}\end{minipage} & 5     & 6     & 3      & 2     & 2         & 3          & 432        & 0.1931         & 0.0107                      & 0.0214         & 0.0214                & $\frac{1}{3} \begin{pmatrix} 1\\ 1\\ 1\end{pmatrix}$\\                     
	19  & \begin{minipage}[b]{0.1\columnwidth}\centering\raisebox{-.45\height}{\includegraphics[width=\linewidth]{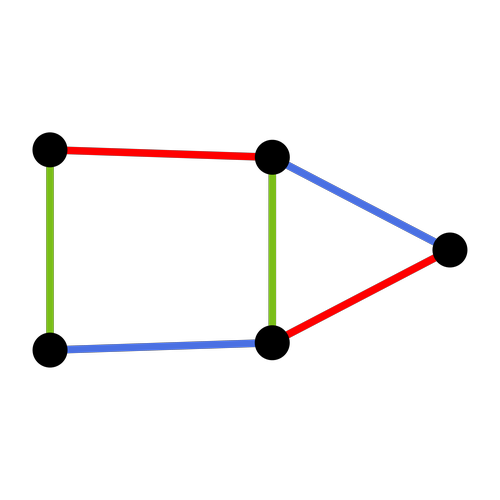}}\end{minipage} & 5     & 6     & 3     & 2      & 3         & 3          & 432        & 0.3106         & 0.0172                      & 0.0343         & 0.0347                & $\begin{pmatrix} 0.3130\\  0.3130\\ 0.3740 \end{pmatrix}$\\                 
	20  & \begin{minipage}[b]{0.1\columnwidth}\centering\raisebox{-.45\height}{\includegraphics[width=\linewidth]{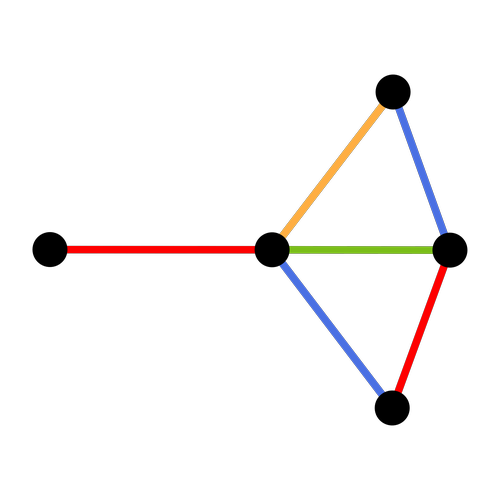}}\end{minipage} & 5     & 6     & 4     & 2      & 3         & 4          & 360        & 0.42           & 0.0208                      & 0.0312         & 0.0319                & $\begin{pmatrix} 0.2470\\ 0.2721\\ 0.2134\\ 0.2674\end{pmatrix}$ \\         
	21  & \begin{minipage}[b]{0.1\columnwidth}\centering\raisebox{-.45\height}{\includegraphics[width=\linewidth]{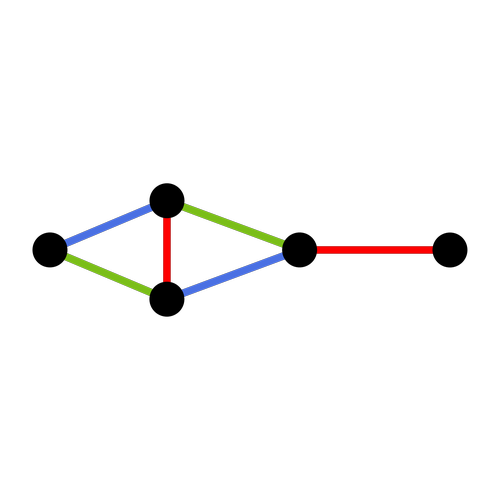}}\end{minipage} & 5     & 6     & 3     & 2      & 3         & 3          & 384        & 0.4036         & 0.0211                      & 0.0422         & 0.0441                & $\begin{pmatrix} 0.2481\\ 0.3760\\ 0.3760 \end{pmatrix}$\\                  
	22  & \begin{minipage}[b]{0.1\columnwidth}\centering\raisebox{-.45\height}{\includegraphics[width=\linewidth]{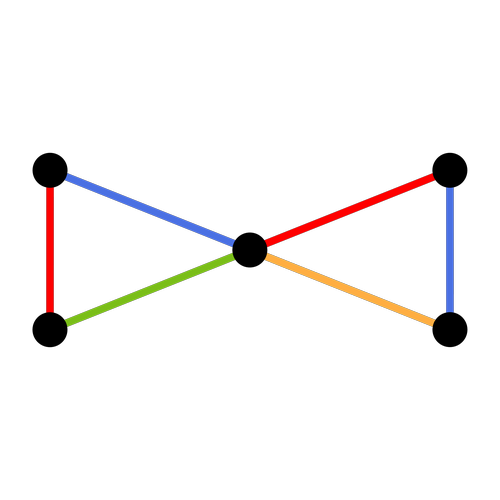}}\end{minipage} & 5     & 6     & 4     & 2      & 3         & 4          & 405        & $\frac{7}{15}$ & 0.0222                      & $\frac{1}{30}$ & $\frac{1}{30}$        & $\frac{1}{5} \begin{pmatrix} 1\\ 1\\ 1\\ 1\\ 1\end{pmatrix}$\\  
	23  & \begin{minipage}[b]{0.1\columnwidth}\centering\raisebox{-.45\height}{\includegraphics[width=\linewidth]{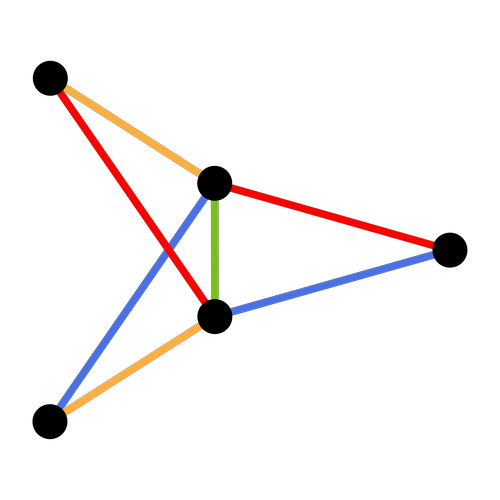}}\end{minipage} & 5     & 7     & 4       & 2     & 3         & 4          & 675        & 0.3236         & 0.0132                      & 0.0231         & 0.0265                & $\begin{pmatrix} 0.2873\\ 0.2873\\ 0.1382\\ 0.2873\end{pmatrix}$\\        
	24  & \begin{minipage}[b]{0.1\columnwidth}\centering\raisebox{-.45\height}{\includegraphics[width=\linewidth]{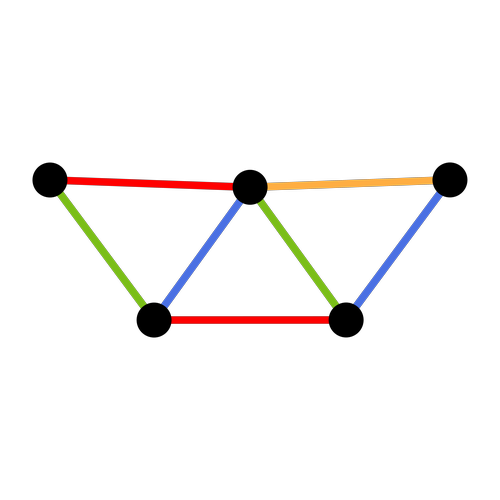}}\end{minipage} & 5     & 7     & 4      & 2     & 3         & 4          & 720        & 0.4263         & 0.0180                      & 0.0315         & 0.0318                & $\begin{pmatrix} 0.2657\\ 0.2462\\ 0.2387\\ 0.2494\end{pmatrix}$\\        
	25  & \begin{minipage}[b]{0.1\columnwidth}\centering\raisebox{-.45\height}{\includegraphics[width=\linewidth]{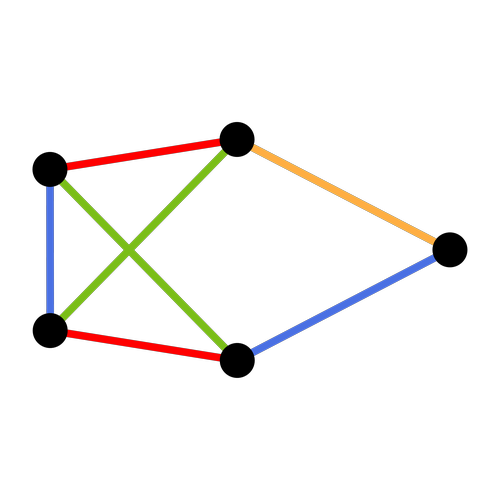}}\end{minipage} & 5     & 7     & 3     & 2      & 3         & 4          & 768        & 0.2501         & 0.0110                      & 0.0192         & 0.0193                & $\begin{pmatrix} 0.2625\\ 0.2375\\ 0.2625\\ 0.2375\end{pmatrix}$\\         
	26  & \begin{minipage}[b]{0.1\columnwidth}\centering\raisebox{-.45\height}{\includegraphics[width=\linewidth]{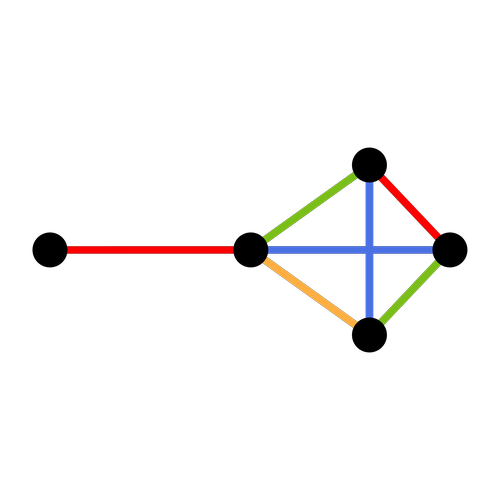}}\end{minipage} & 5     & 7     & 4     & 2      & 4         & 4          & 640        & 0.4877         & 0.0190                      & $\frac{1}{30}$ & $\frac{1}{30}$        & $\frac{1}{4} \begin{pmatrix} 1\\ 1\\ 1\\ 1\end{pmatrix}$\\               
	27  & \begin{minipage}[b]{0.1\columnwidth}\centering\raisebox{-.45\height}{\includegraphics[width=\linewidth]{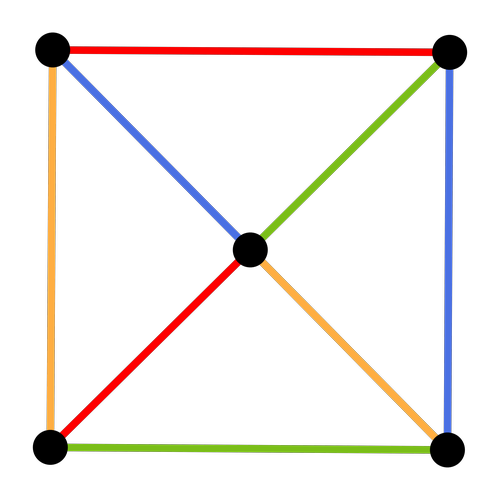}}\end{minipage} & 5     & 8     & 4     & 2      & 3         & 4          & 1280       & 0.2836         &0.0100                     & 0.0199         & 0.0199                & $\frac{1}{4} \begin{pmatrix} 1\\ 1\\ 1\\ 1\end{pmatrix}$\\              
	28  & \begin{minipage}[b]{0.1\columnwidth}\centering\raisebox{-.45\height}{\includegraphics[width=\linewidth]{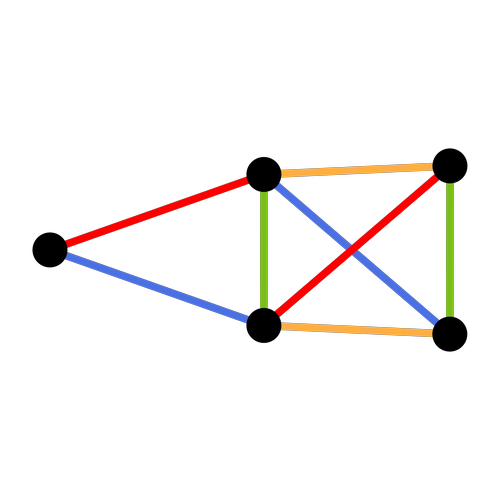}}\end{minipage} & 5     & 8     & 4      & 2     & 4         & 4          & 1200       & 0.4053         & 0.0135                      & 0.0269         & 0.0298                & $\begin{pmatrix} 0.2818\\  0.2818\\  0.1687\\  0.2677 \end{pmatrix}$\\     
	29  & \begin{minipage}[b]{0.1\columnwidth}\centering\raisebox{-.45\height}{\includegraphics[width=\linewidth]{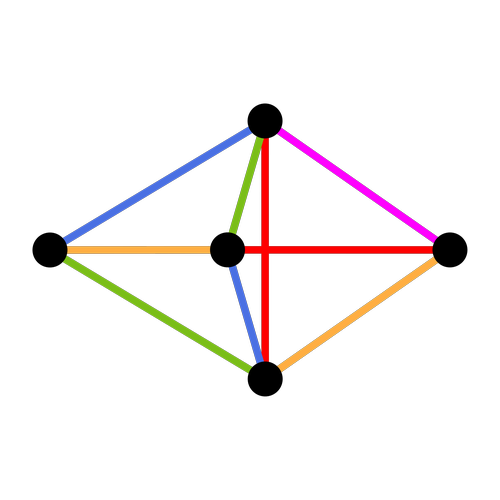}}\end{minipage} & 5     & 9     & 4      & 2     & 4         & 5          & 2000       & $\frac{2}{5}$  & 0.0111                      & 0.02           & 0.0203                & $\begin{pmatrix} 0.1850\\ 0.1965\\ 0.1850\\ 0.2372\\ 0.1965\end{pmatrix}$\\ 
	30  & \begin{minipage}[b]{0.1\columnwidth}\centering\raisebox{-.45\height}{\includegraphics[width=\linewidth]{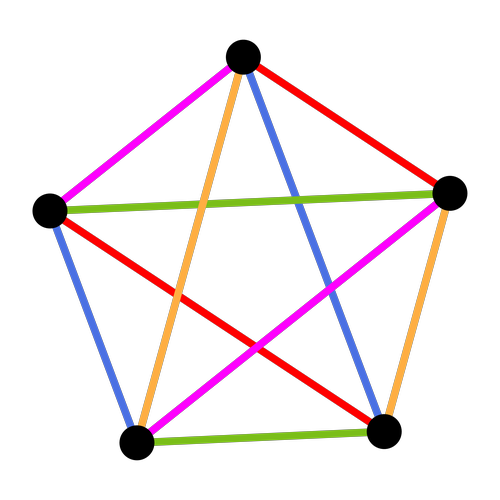}}\end{minipage} & 5     & 10    & 4       & 2    & 5         & 5          & 3125       & $\frac{3}{5}$  & 0.0133                      & $\frac{2}{75}$ & $\frac{2}{75}$        & $\frac{1}{5} \begin{pmatrix}   1\\ 1\\ 1\\ 1\\ 1\end{pmatrix}$\\    
	
	\hline \hline
\end{longtable*}
%
%

\twocolumngrid

\bibliography{all_references}

\end{document}